\documentclass[11pt]{book}
\usepackage{geometry}                % See geometry.pdf to learn the layout options. There are lots.
\usepackage{graphicx}	
\usepackage{amssymb}
\usepackage{amsmath}
\usepackage{amsfonts}
\usepackage{amsthm}
\usepackage{bbm}
\usepackage{cancel}
\usepackage{color}
\usepackage{hyperref}
\usepackage{appendix}
\newcommand{\indicator}[1]{\mathbbm{1}_{\left[ {#1} \right] }}
\newcommand{\Real}{\hbox{Re}}

\newtheorem{lemma}{ Lemma}[section]

\newtheorem{proposition}{ Proposition}[section]

\newtheorem{remark}{ Remark}

\newtheorem{theorem}{ Theorem}[section]

\newtheorem{definition}{Definition}[section]
\newtheorem{example}{Example}[section]

\usepackage{etoolbox}
\makeatletter
\patchcmd{\chapter}{\if@openright\cleardoublepage\else\clearpage\fi}{}{}{}
\makeatother
\newcommand{\HRule}{\rule{\linewidth}{0.5mm}}
\begin{document}
\bibliographystyle{apalike}

\begin{titlepage}
\begin{center}

% Upper part of the page. The '~' is needed because \\
% only works if a paragraph has started.
%\includegraphics[width=0.15\textwidth]{./logo}~\\[1cm]

\textsc{\LARGE Department of Finance and Risk Engineering\\Tandon School of Engineering\\New York University\\[1.5cm]}

%\textsc{\Large .}\\[0.5cm]

% Title
\HRule \\[0.4cm]
{ \huge \bfseries Introduction to Stochastic Differential Equations (SDEs) for Finance \\[0.4cm] }

\HRule \\[1.5cm]

\textsc{\Large \emph{Author:}\\
Andrew \textsc{Papanicolaou}}\\ap1345@nyu.edu\\[0.5cm]

\vfill

% Bottom of the page

\end{center}
{\large This work was partially supported by NSF grant DMS-0739195. }

\end{titlepage}

\tableofcontents

\chapter{Financial Introduction}
\label{sec:intro}
In this section we'll discuss some of the basic ideas of option pricing. The main idea is \textit{replication}, whereby a derivative security can be priced because it is shown to have the same cash flows as a portfolio of assets that already have a price-tag. Another word to describe such a notion of pricing is \textit{benchmarking}, where to say that you've `benchmarked the security' might mean that you've found a portfolio of other assets that do not replicate but have have some similarities to the derivative's cash flows.

The idea of finding a portfolio that is close in value to the derivative security is essentially the \textit{law of one price}, which states that ``In an efficient market, all identical goods must have only one price." Indeed, we will assume that our markets are efficient, and in some cases we will assume that arbitrage has zero probability of occurring; these assumptions are routine and are generally not considered to be restrictive. 

The manner in which these notes can be considered `oversimplified' is in the \textit{completeness} of the resulting markets. In practice there are derivatives (e.g. insurance products) which cannot be hedged, and hence the market is incomplete. Both the discrete time-space market and the Black-Scholes market are simple enough for completeness to hold. In practice, reverse-engineering these models from real-life market data will require interpretation.
%%%%%%%%%%%%%%%%%%%
\section{A Market in Discrete Time and Space}
Consider a very simple market where at time $t=0$ it is known that there are only two possible states for the market at a later time $t=T$. In between times $t=0$ and $t=T$ there is no trading of any kind. This market is described by a probability space $\Omega = \{\omega_1,\omega_2\}$ with probability measure $\mathbb P$ given by
\[p_1 = \mathbb P(\omega_1) = \frac 23\ ,\qquad p_2=\mathbb P(\omega_2) =\frac 13\ .\]
The elementary events $\omega_1$ and $\omega_2$ are the two states of the market. The traded assets in this simple market are a bank account, a stock, and and a call option on the stock with exercise at $T$ and strike $K=2$. The market outcomes are shown in Table \ref{tab:assets}.
\begin{table}[htb]
\center
\caption{Assets in the Discrete Time-Space Market}
\label{tab:assets}
\begin{tabular}{|c|c|c|}
\hline
time &t=0&t=T\\
\hline
bank account&$B_0=1$&$B_T=1$\hbox{ (interest rate $r=0$)}\\
stock&$S_0=2$&$S_T = \Big\{
\begin{array}{cc}
3&\hbox{in }\omega_1\\
1&\hbox{in }\omega_2
\end{array}$\\
call option, $K=2$&$C_0=~?$&$C_T = \Big\{
\begin{array}{cc}
1&\hbox{in }\omega_1\\
0&\hbox{in }\omega_2
\end{array}$\\
\hline
\end{tabular}
\end{table}

The way to determine the price of the call option, $C_0$, is to replicate it with a portfolio of the stock and the bank account. Let $V_t$ denote the value of such a portfolio at time $t$, so that,

\begin{eqnarray*}
V_0&=&\alpha S_0+\beta\\
V_T&=&\alpha S_T+\beta\ ,
\end{eqnarray*}
where $\alpha$ is \# of shares and $\beta$ is \$ in bank. The portfolio $V_t$ will replicate the call option if we solve for $\alpha$ and $\beta$ so that $V_T=C_T$ for both $\omega_1$ and $\omega_2$,
\begin{align*}
&3\alpha+\beta=\alpha S_T(\omega_1)+\beta=C_T(\omega_1) =1\\
&\alpha+\beta=\alpha S_T(\omega_1)+\beta=C_T(\omega_2) = 0\ .
\end{align*}
This system has solution $\alpha = \frac 12$, $\beta =-\frac 12$. Hence, by the law of one price, it must be that,
\[C_0 = V_0 = \frac 12S_0-\frac 12 = \frac 12. \]
If $C_0\neq V_0$, then there would be an arbitrage opportunity in the market. Arbitrage in a financial models means it is possible to purchase a portfolio for which there is a risk-less profit to be earned with positive probability. If such an opportunity is spotted then an investor could borrow an infinite amount of money and buy the arbitrage portfolio, giving him/her an infinite amount of wealth. Arbitrage (in this sense) does not make for a sound financial model, and in real life it is known that true arbitrage opportunities disappear very quickly by the efficiency of the markets. Therefore, almost every financial model assumes absolute efficiency of the market and `no-arbitrage'. There is also the concept of `No Free Lunch', but in discrete time-space we do not have to worry about theses differences. 
%%%%%%%%

The notion of `no-arbitrage' is defined as follows:
\begin{definition}Let $wealth_t$ denote the wealth of an investor at time $t$. We say that the market has arbitrage if,
\begin{align*} 
&\mathbb P(\hbox{wealth}_T>0|\hbox{wealth}_0=0)>0\\
\qquad\hbox{and  }\qquad&\mathbb P(\hbox{wealth}_T<0|\hbox{wealth}_0=0)=0\ ,
\end{align*}
i.e. there is positive probability of a gain and zero probability of a loss.
\end{definition}

In our discrete time-space market, if $C_0<V_0$ then the arbitrage portfolio is one that buys the option, shorts the portfolio, and invest the difference in the bank. The risk-less payoff of this portfolio is shown in Table \ref{tab:arb}.

\begin{table}[h]
\caption{Arbitrage Opportunity if Option Mispriced}
\label{tab:arb}
\center
\begin{tabular}{c|ccc}
&$t=0$&$t=T$\\
\hline
buy option&$-C_0$&$\max\{S_T-2,0\}$\\
sell portfolio&$V_0$&$-\max\{S_T-2,0\}$\\
\hline
net:&$\color{red}{V_0-C_0>0}$&$0$
\end{tabular}
\end{table}

%%%%%%%%%%%%%%%%%%%%%%
\section{Equivalent Martingale Measure (EMM)}
A common method for pricing an asset is to use a risk-neutral or an equivalent martingale measure (EMM). The EMM is convenient because all asset prices are simply an expectation of the payoff.Two important questions are: what is the EMM? Is there more one?

\begin{definition}
The probability measure $\mathbb Q$ is an EMM of $\mathbb P$ if $S_t$ is a $\mathbb Q$-martingale, that is
\[\mathbb E^QS_T = S_0\]
and $\mathbb Q$ is \textbf{equivalent} to $\mathbb P$,
\[\mathbb P(\omega)>0\Leftrightarrow \mathbb Q(\omega)>0\]
for all elementary events $\omega\in\Omega$ (it's more complicated when $\Omega$ is not made up of elementary events, but it's basically the same idea).
\end{definition}

For our discrete time-space example market, we have
\begin{itemize}
\item Under the original measure 
 \[\mathbb ES_T = p_1S_T(\omega_1)+p_2S_T(\omega_2)=\frac 23 3+\frac 13 1 = \frac 73> 2=S_0\]
 \item Under an EMM $\mathbb Q(\omega_1)=q_1$ and $\mathbb Q(\omega_2)=q_2$,
 \[ \mathbb E^QS_T=q_13+(1-q_1)1 = 2 = S_0\ .\]
Solution is $q_1 = \frac 12$ and $q_2 = 1-q_1 = \frac 12$.
 \end{itemize}

Given the EMM, a replicable option is easily priced:
\[\mathbb E^QC_T = \mathbb E^QV_T=\mathbb E^Q[\alpha S_T+\beta]=\alpha\mathbb E^QS_T+\beta =\alpha S_0+\beta = V_0 = C_0\ ,\]
and so $C_0 = \mathbb E^Q\max(S_T-K,2)$.
%%%%%%%%%%%%%%%%%%%%%%
\section{Contingent Claims}
A corporation is interested in purchasing a derivative product to provide specific cash-flows for each of the elementary events that (they believe) the market can take. Let $\Omega=\{\omega_1,\omega_2,\dots,\omega_N\}$ be these elementary events that can occur at time $t=T$, and let the proposed derivative security be a function $C_t(\omega)$ such that,
\[C_T(\omega_i) = c_i\qquad\forall i\leq N\ ,\]
where each $c_i$ is the corporations desired cash flow. The derivative product $C$ is a \textit{contingent claim}, because it pays a fixed amount for all events in the market. Below are some general examples of contingent claims:

\begin{example}\textbf{European Call Option.} At time $T$ the holder of the option has the right to buy an asset at a pre-determined strike price $K$. This is a contingent claim that pays, 
\[(S_T(\omega)-K)^+=\max(S_T(\omega)-K,0)\ ,\]
where $S_T$ is some risk asset (e.g. a stock or bond). The payoff on this call option is seen in Figure \ref{fig:callPayoff}.

\begin{figure}[htbp] %  figure placement: here, top, bottom, or page
   \centering
   \includegraphics[width=5in]{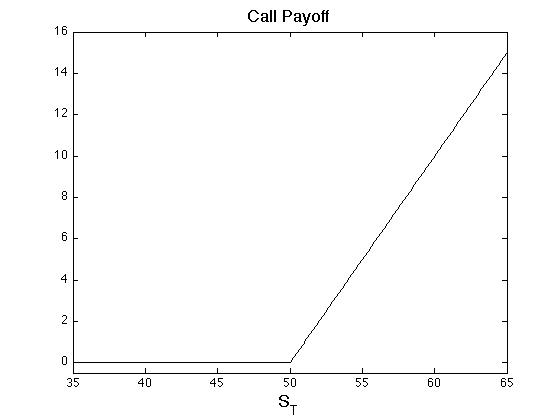} 
   \caption{The payoff on a European call option with strike $K=50$.}
   \label{fig:callPayoff}
\end{figure}
\end{example}
\begin{example}\textbf{European Put Option.} At time $T$ the holder of the option has the right to sell an asset at a pre-determined strike price $K$. This is a contingent claim that pays, 
\[(K-S_T(\omega))^+=\max(K-S_T(\omega),0)\ .\]
Figure \ref{fig:putPayoff} shows the payoff for the put.
\begin{figure}[htbp] %  figure placement: here, top, bottom, or page
   \centering
   \includegraphics[width=5in]{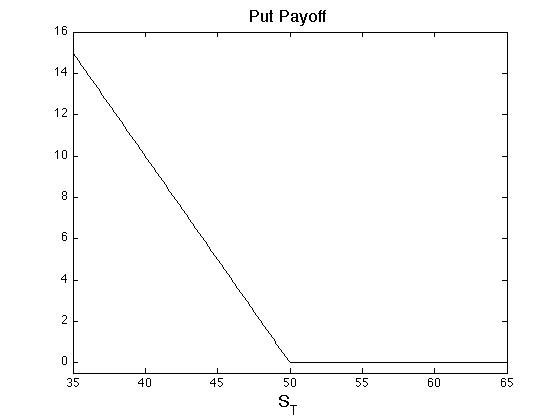} 
   \caption{The payoff on a European put option with strike $K=50$.}
   \label{fig:putPayoff}
\end{figure}
\end{example}

\begin{example}\textbf{American Call/Put Option.} At any time $t\leq T$ the holder of the option has the right to buy/sell an asset at a pre-determined strike price $K$. Exercise of this claim at time $t\leq T$ is contingent on the event $\mathcal A_t=\{\omega_i|\hbox{ it is optimal to exercise at time $t$}\}$.
\end{example}

\begin{example}\textbf{Bermuda Call/Put Option.} At either time $T/2$ or time $T$, the holder of the option has the right to buy/sell an asset at a pre-determined strike price $K$. Similar to an American option except only one early strike time.
\end{example}

\begin{example}\textbf{Asian Call Option.} At time $T$ the holder of the option has the right to buy the \textit{averaged} asset at a pre-determined strike price $K$,
\[\left(\frac 1M\sum_{\ell=1}^MS_{t_\ell}(\omega)-K\right)^+\ ,\]
where $t_\ell = \ell\frac TM$ for integer $M>0$.
\end{example}

\begin{example}\textbf{Exchange Option.} At time $T$ the holder of the option has the right to buy one asset at the price of another,
\[(S_T^1(\omega)-S_T^2(\omega))^+ ,\]
where $S_t^1$ is the price of an asset and $S_t^2$ is the price of another.
\end{example}

%%%%%%%%%%%%%%%%%%%%%%
\section{Option Pricing Terminology}
The following is a list of terms commonly used in option pricing:
\begin{itemize}
\item \textbf{Long position,} a portfolio is `long asset X' if it has net positive holdings of contracts in asset $X$.
\item \textbf{Short position,} a portfolio is `short asset X' if it has net negative holdings of contracts in asset $X$ (i.e. has short sales of contracts).
\item \textbf{Hedge,} or `hedging portfolio' is a portfolio that has minimal or possibly a floor on the losses it might obtain.
\item \textbf{In-the-money (ITM),} a derivative contract that would have positive payoff if settlement based on today's market prices (e.g. a call option with very low strike).
\item \textbf{Out-of-the-money (OTM),} a derivative contract that would be worthless if settlement based on today's market prices (e.g. a call option with very high strike).
\item \textbf{At-the-money (ATM),} a derivative contract exactly at it's breaking point between ITM and OTM.
\item \textbf{Far-from-the-money,} a derivative contract with very little chance of finishing ITM.
\item \textbf{Underlying,} the stock, bond, ETF, exchange rate, etc. on which a derivative contract is written.
\item \textbf{Strike,} The price upon which a call or put option is settled.
\item \textbf{Maturity,} the latest time at which a derivative contract can be settled.
\item \textbf{Exercise,} the event that the long party decides to use a derivative's embedded option (e.g. using a call option to buy a share of stock at lower than market value).
\end{itemize}
%%%%%%%%%%%%%%%%%%%%%%
\section{Completeness \& Fundamental Theorems}
The nice thing about the discrete time-space market is that any contingent claim can be replicated. In general, for $\Omega=\{\omega_1,\omega_2,\dots,\omega_N\}$, and with $N-1$ risky-assets $(S^1,S^2,\dots,S^{N-1})$ and a risk-free bank account (with $r=0$), replicating portfolio weights the contingent claim $C$ can be found by solving, 
\[\begin{pmatrix}
1&S_T^1(\omega_1)&\dots&\dots&S_T^{N-1}(\omega_1)\\
1&S_T^1(\omega_2)&\dots&\dots&S_T^{N-1}(\omega_2)\\
\vdots&\vdots&\ddots&&\vdots\\
\vdots&\vdots&&\ddots&\vdots\\
1&S_T^1(\omega_N)&\dots&\dots&S_T^{N-1}(\omega_N)
\end{pmatrix}
\begin{pmatrix}
\beta\\
\alpha_1\\
\vdots\\
\vdots\\
\alpha_{N-1}\\
\end{pmatrix}
=\begin{pmatrix}
c_1\\
c_2\\
\vdots\\
\vdots\\
c_N
\end{pmatrix}\ ,
\]
which has a solution $(\beta,\alpha_1,\dots,\alpha_{N-1})$ provided that none of these assets are redundant (e.g., there does not exist a portfolio consisting of the first $N-2$ assets and the banks account that replicated the $S_T^{N-1}$). This $N$-dimensional extension of the discrete time-space market serves to further exemplify the importance of replication in asset pricing, and should help to make clear the intentions of the following definition and theorems:

\begin{definition} A contingent claim is \textbf{reachable} if there is a hedging portfolio $V$ such the $V_T(\omega) = C(\omega)$ for all $\omega$, in which case we say that $C$ can be \textbf{replicated}. If all contingent claims can be replicated, then we say the market is \textbf{complete}. 
\end{definition}

\begin{theorem}
The 1st fundamental theorem of asset pricing states the market is arbitrage-free if and only if there exists an EMM.
\end{theorem}
\begin{theorem}
The 2nd fundamental theorem of asset pricing states the market is arbitrage-free and complete if and only if there exists a unique EMM.
\end{theorem}
One of the early works that proves these theorems is \cite{harrisonPliska}. Another good article on the subject is \cite{schachermayer1992}. In summary, these fundamental theorems mean that derivative prices are the expectation under an EMM. 
\chapter{Brownian Motion \& Stochastic Calculus}
\label{chapt:brownianMotion}
Continuous time financial models will often use Brownian motion to model the trajectory of asset prices. One can typically open the finance section of a newspaper and see a time series plot of an asset's price history, and it might be possible that the daily movements of the asset resemble a random walk or a path taken by a Brownian motion. Such a connection between asset prices and Brownian motion was central to the formulas of \cite{blackScholes1973} and \cite{merton1973}, and have since led to a variety of models for pricing and hedging. The study of Brownian motion can be intense, but the main ideas are a simple definition and the application of It\^o's lemma. For further reading, see \cite{bjork,oksendal}.

%%%%%%%%%%%%%%%%%%%%%%
\section{Definition and Properties of Brownian Motion}
On the time interval $[0,T]$, Brownian motion is a continuous stochastic process $(W_t)_{t\leq T}$ such that
\begin{enumerate}
\item $W_0 = 0$,
\item Independent Increments: for $0\leq s'< t'\leq s<t\leq T$, $W_t-W_s$ is independent of $W_{t'}-W_{s'}$,
\item Conditionally Gaussian: $W_t-W_s\sim\mathcal N(0,t-s)$, i.e. is normal with mean zero and variance $t-s$.
\end{enumerate}
There is a vast study of Brownian motion. We will instead use Brownian motion rather simply; the only other fact that is somewhat important is that Brownian motion is nowhere differentiable, that is
\[\mathbb P\left(\frac{d}{dt}W_t\hbox{ is undefined for almost-everywhere }t\in[0,T]\right) = 1\ ,\]
although sometimes people write $\dot W_t$ to denote a white noise process. It should also be pointed out that $W_t$ is a martingale,

\[\mathbb E[W_t|(W_\tau)_{\tau\leq s}] = W_s\qquad\forall s\leq t\ .\]

\begin{figure}[h] %  figure placement: here, top, bottom, or page
   \centering
   \includegraphics[width=5in]{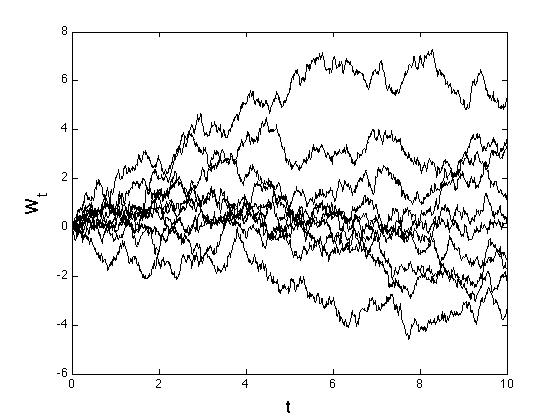} 
   \caption{A sample of 10 independent Brownian motions.}
   \label{fig:brownianTraj}
\end{figure}
\noindent\textbf{Simulation.} There are any number of ways to simulate Brownian motion, but to understand why Brownian motion can be treated like a `random walk', consider the process,
\[W_{t_n}^N = W_{t_{n-1}}^N+\Bigg\{\begin{array}{cc}
\sqrt{T/N}&\hbox{with probability }1/2\\
- \sqrt{T/N}&\hbox{with probability }1/2
\end{array}\]
with $W_{t_0} = 0$ and $t_n = n\frac TN$ for $n=0,1,2\dots,N$. Obviously $W_{t_n}^N$ has independent increments, and conditionally has the same mean and variance as Brownian motion. However it's not conditional Gaussian. However, as $N\rightarrow \infty$ the probability law of $W_{t_n}^N$ converges to the probability law of Brownian motion, so this simple random walk is actually a good way to simulate Brownian motion if you take $N$ large. However, one usually has a random number generator that can produce $W_{t_n}^N$ that also has conditionally Gaussian increments, and so it probably better to simulate Brownian motion this way. A sample of 10 independent Brownian Motions simulations are shown in Figure \ref{fig:brownianTraj}.

%%%%%%%%%%%%%%%%%%%%%%
\section{The It\^o Integral}
Introductory calculus courses teach differentiation first and integration second, which makes sense because differentiation is generally a methodical procedure (e.g., you use the chain rule) whereas finding an anti-derivative requires all kinds of change-of-variables, trig-substitutions, and guess-work. The It\^o calculus is derived and taught in the reverse order: first we need to understand the definition of an It\^o (stochastic) integral, and then we can understand what it means for a stochastic process to have a differential. 

The construction of the It\^o integral begins with a backward Riemann sum. For some function $f:[0,T]\rightarrow \mathbb R$ (possibly random), non-anticipative of $W$, the It\^o integral is defined as
\begin{equation} 
\label{eq:itoInt}
\int_0^Tf(t)dW_t =\lim_{N\rightarrow 0}\sum_{n=0}^{N-1} f(t_n)(W_{t_{n+1}}-W_{t_n})\ ,
\end{equation}
where $t_n = n\frac TN$, with the limit holding in the strong sense. If $\mathbb E\int_0^Tf^2(t)dt<\infty$, then through an application of Fubini's theorem, it can be shown that equation \eqref{eq:itoInt} is a \textbf{martingale,}
\[\mathbb E\left[\int_0^Tf(t)dW_t \Big|(W_\tau)_{\tau\leq s}\right] = \int_0^sf(t)dW_t\qquad\forall s\leq T\ .\]
Another important property of the stochastic integral is the It\^o Isometry,
\begin{proposition} \textbf{(It\^o Isometry).} For any functions $f,g$ (possibly random), non-anticipative of $W$, with $\mathbb E\int_0^Tf^2(t)dt<\infty$ and $\mathbb E\int_0^Tg^2(t)dt<\infty$, then
\[\mathbb E\left(\int_0^Tf(t)dW_t\right)\left(\int_0^Tg(t)dW_t\right)  =\mathbb E\int_0^Tf(t)g(t)dt\ .\]
\end{proposition}

Some facts about the It\^o integral:
\begin{itemize}
\item One can look at Equation \eqref{eq:itoInt} and think about the stochastic integral as a sum of independent normal random variables with mean zero and variance $T/N$,
\[\sum_{n=0}^{N-1} f(t_n)(W_{t_{n+1}}-W_{t_n})\sim \mathcal N\left( 0 ~,\frac TN\sum f^2(t_n)\right)\ .\]
Therefore, one might suspect that $\int_0^Tf(t)dW_t$ is normal distributed. 

\item In fact, the It\^o integral is normally distributed when $f$ is a non-stochastic function, and its variance is given by te It\^o isometry,
\[\int_0^Tf(t)dW_t\sim \mathcal N\left(0~,\int_0^Tf^2(t)dt\right)\ .\]

\item The It\^o integral is also defined for functions of another random variable. For instance, $f:\mathbb R\rightarrow \mathbb R$ and another random variable $X_t$, the It\^o integral is, 
\[\int_0^Tf(X_t)dW_t =\lim_{N\rightarrow\infty}\sum_{n=0}^{N-1}f(X_{t_n})(W_{t_{n+1}}-W_{t_n})\ .\]
The It\^o isometry for this integral is,
\[\mathbb E\left(\int_0^Tf(X_t)dW_t\right)^2 = \int_0^T\mathbb Ef^2(X_t)dt\ ,\]
provided that $X$ is non-anticipative of $W$.
\end{itemize}
%%%%%%%%%%%%%%%%%%%%%%
\section{Stochastic Differential Equations \& It\^o's Lemma}
With the stochastic integral defined, we can now start talking about differentials. In applications the stochastic differential is how we think about random process that evolve over time. For instance, the return on a portfolio or the evolution of a bond yield. The idea is not that various physical phenomena are Brownian motion, but that they are \textit{driven} by a Brownian motion.

Instead of a differential equation, the integrands in It\^o integrals satisfy stochastic differential equations (SDEs). For instance,

\begin{align*}
&dS_t = \mu S_tdt+\sigma S_tdW_t\qquad\qquad\hbox{geometric Brownian motion,}\\
&dY_t = \kappa(\theta-Y_t)dt+\gamma dW_t\qquad\hbox{an Ornstein-Uhlenbeck process,}\\
&dX_t=a(X_t)dt+b(X_t)dW_t\qquad\hbox{a general SDE.}
\end{align*}
Essentially, the formulation of an SDE tells us the It\^o integral representation. For instance,
\[dX_t=a(X_t)dt+b(X_t)dW_t\qquad\Leftrightarrow\qquad X_t = X_0+\int_0^ta(X_s)ds+\int_0^tb(X_s)dW_s\ .\]
Hence, any SDE that has no $dt$-terms is a martingale. For instance, if $dS_t = rS_tdt+\sigma S_tdW_t$, then $F_t=e^{-rt}S_t$ satisfies the SDE $dF_t = \sigma F_tdW_t$ and is therefore a martingale.

On any given day, mathematicians, physicists, and practitioners may or may not recall the conditions on functions $a$ and $b$ that provide a sound mathematical framework, but the safest thing to do is to work with coefficients that are known to provide existence and uniqueness of solutions to the SDE (see page 68 of \cite{oksendal}).
\begin{theorem}
\label{thm:sdeEU}
\textbf{(Conditions for Existence and Uniqueness of Solutions to SDEs).} For $0<T<\infty$, let $t\in[0,T]$ and consider the SDE
\[dX_t = a(t,X_t)dt+b(t,X_t)dW_t\ ,\]
with initial condition $X_ 0 =x$ (x constant). Sufficient conditions for existence and uniqueness of square-integrable solutions to this SDE (i.e. solutions such that $\mathbb E\int_0^T|X_t|^2dt<\infty$) are \textbf{linear growth}
\[|a(t,x)|+|b(t,x)|\leq C(1+|x|)\qquad\forall x\in \mathbb R\hbox{ and }\forall t\in [0,T]\]
for some finite constant $C>0$, and \textbf{Lipschitz continuity}
\[|a(t,x)-a(t,y)|+|b(t,x)-b(t,y)|\leq D|x-y|\qquad\forall x,y\in\mathbb R\hbox{ and }\forall t\in [0,T]\]
where $0<D<\infty$ is the Lipschitz constant.
\end{theorem}
\noindent From the statement of Theorem \ref{thm:sdeEU} it should be clear that these are not necessary conditions. For instance, the widely used \textbf{square-root process}
\[dX_t = \kappa(\bar X-X_t)dt+\gamma\sqrt{X_t}dW_t\]
does not satisfy linear growth or Lipschitz continuity for $x$ near zero, but there does exist a unique solution if $\gamma^2\leq 2\bar X\kappa$. In general, existence of solutions for SDEs not covered by Theorem \ref{thm:sdeEU} needs to be evaluated on a case-by-case basis. The rule of thumb is to stay within the bounds of the theorem, and only work with SDEs outside if you are certain that the solution exists (and is unique). 
\begin{example}[Tanaka Equation]
The canonical example to demonstrate non-uniqueness for non-Lipschitz coefficients is the Tanaka equation,
\[dX_t = sgn(X_t)dW_t\ ,\]
with $X_0=0$, where $sgn(x) $ is the sign function; $sgn(x) = 1$ if $x>0$, $sgn(x) = -1$ if $x<0$ and $sgn(x) = 0$ if $x=0$. Consider another Brownian motion $\hat W_t$ and define
\[\widetilde W_t = \int_0^tsgn(\hat W_s)d\hat W_s\ ,\]
where it can be checked that $\widetilde W_t$ is also a Brownian motion. We can also write
\[d\hat W_t = sgn(\hat W_t)d\widetilde W_t\ ,\]
which shows that $X_t=\hat W_t$ is a solution to the Tanaka equation. However, this is referred to as a \textbf{weak solution}, meaning that the driving Brownian motion was recovered after the solution $X$ was given; a \textbf{strong solution} is a solution obtained when first the Brownian motion is given. Notice this weak solution is non-unique: take $Y_t=-X_t$ and look at the differential,
\[dY_t = -dX_t = -sgn(X_t)d\widetilde W_t=sgn(Y_t)d\widetilde W_t\ .\]
Notice that $X_t\equiv 0$ is also a solution.
\end{example}

Given a stochastic differential equation, It\^o's lemma tells us the differential of any function on that process. It\^o's lemma can be thought of as the stochastic analogue to differentiation, and is a fundamental tool in stochastic differential equations:

\begin{lemma}\textbf{(It\^o's Lemma).} Consider the process $X_t$ with SDE $dX_t = a(X_t)dt+b(X_t)dW_t$. For a function $f(t,x)$ with at least one derivative in $t$ and at least two derivatives in $x$, we have
\begin{equation}
\label{eq:itoLemma}
df(t,X_t) = \left(\frac{\partial}{\partial t}+a(X_t)\frac{\partial}{\partial x}+\frac{b^2(X_t)}{2}\frac{\partial^2}{\partial x^2}\right)f(t,X_t)dt+b(X_t)\frac{\partial}{\partial x}f(t,X_t)dW_t\ .
\end{equation}
\end{lemma}
\noindent \textbf{Details on the It\^o Lemma.} The proof of \eqref{eq:itoLemma} is fairly involved and has several details to check, but ultimately, It\^o's lemma is a Taylor expansion to the 2nd order term, e.g. 
\[f(W_t)\simeq f(W_{t_0})+f'(W_{t_0})(W_t-W_{t_0})+\frac 12f''(W_{t_0})(W_t-W_{t_0})^2\]
for $0<t-t_0\ll 1$, (i.e. $t$ just slightly greater than $t_0$). To get a sense of why higher order terms drop out, take $t_0=0$ and $t_n = nt/N$ for some large $N$, and use the Taylor expansion:
\begin{align*}
f(W_t)-f(W_0) &= \sum_{n=0}^{N-1} f'(W_{t_n})(W_{t_{n+1}}-W_{t_n})+ \frac12\sum_{n=0}^{N-1} f''(W_{t_n})(W_{t_{n+1}}-W_{t_n})^2\\
&+ \frac16\sum_{n=0}^{N-1} f'''(W_{t_n})(W_{t_{n+1}}-W_{t_n})^3+ \frac{1}{24}\sum_{n=0}^{N-1} f''''(\xi_n)(W_{t_{n+1}}-W_{t_n})^4\ ,
\end{align*}
where $\xi_n$ is some (random) intermediate point to make the expansion exact. Now we use independent increments and the fact that 
\[\mathbb E(W_{t_{n+1}}-W_{t_n})^k=\Bigg\{\begin{array}{cc}
 \left(\frac tN\right)^{k/2}(k-1)!!&\hbox{if $k$ even}\\
 0&\hbox{if $k$ odd}\ ,
 \end{array}\]
from which is can be seem that the $f'$ and $f''$ are are significant,
\begin{align*}
\mathbb E\left(\sum_{n=0}^{N-1} f'(W_{t_n})(W_{t_{n+1}}-W_{t_n})\right)^2&= \frac tN\sum_{n=0}^{N-1} \mathbb E|f'(W_{t_n})|^2 = \mathcal O\left(1\right)\ ,\\
\mathbb E\sum_{n=0}^{N-1} f''(W_{t_n})(W_{t_{n+1}}-W_{t_n})^2& =\frac tN\sum_{n=0}^{N-1} \mathbb Ef''(W_{t_n})=\mathcal O\left(1\right)\ ,
\end{align*}
and assuming there is some bound $M<\infty$ such that $|f'''|\leq M$ and $|f''''|\leq M$, we see that the higher order terms are arbitrarily small,
 \begin{align*}
\mathbb E\left|\sum_{n=0}^{N-1} f'''(W_{t_n})(W_{t_{n+1}}-W_{t_n})^3\right|&\leq M\sum_{n=0}^{N-1} \mathbb E\left|W_{t_{n+1}}-W_{t_n}\right|^3\\
&\leq M\sum_{n=0}^{N-1}\sqrt{\mathbb E\left|W_{t_{n+1}}-W_{t_n}\right|^6}=\mathcal O\left(\frac{1}{N^{1/2}}\right)\ , \\
\mathbb E\left|\sum_{n=0}^{N-1} f''''(\xi_n)(W_{t_{n+1}}-W_{t_n})^4\right|&\leq M\sum_{n=0}^{N-1} \mathbb E(W_{t_{n+1}}-W_{t_n})^4=\mathcal O\left(\frac{1}{N}\right)\ ,
\end{align*}
where we've used the Cauchy-Schwartz inequality, $\mathbb E|Z|^3\leq \sqrt{\mathbb E|Z|^6}$ for some random-variable $Z$.\\

\noindent The following are examples that should help to familiarize with the It\^o lemma and solutions to SDEs:
\begin{example}An Ornstein-Uhlenbeck process 
 $dX_t = \kappa(\theta-X_t)dt+\gamma dW_t$ has solution $$X_t=\theta+(X_0-\theta)e^{-\kappa t}+\gamma \int_0^te^{-\kappa(t-s)}dW_s\ .$$
 This solution uses an integrating factor\footnote{For an ordinary differential equation $\frac{d}{dt}X_t+\kappa X_t = a_t$, the integrating factor is $e^{\kappa t}$ and the solution is $X_t = X_0e^{-\kappa t}+\int_0^ta_se^{-\kappa(t-s)}ds$. An equivalent concept applies for stochastic differential equations.} of $e^{\kappa t}$,
 $$dX_t+\kappa X_tdt = \kappa\theta dt+\gamma dW_t\ .$$

 \end{example}
 
 \begin{example}Apply It\^o's lemma to $X_t=W_t^2$ to get the SDE
\[dX_t = dt+2W_tdW_t\ .\]
\end{example}

\begin{example}The canonical SDE in financial math, the geometric Brownian motion,
$\frac{dS_t}{S_t} = \mu dt+\sigma dW_t$ has solution
$$S_t = S_0e^{\left(\mu-\frac{1}{2}\sigma^2\right)t+\sigma W_t}$$
which is always positive. Again, verify with It\^o's lemma. Also try It\^o's lemma on $\log(S_t)$.
\end{example}

\begin{example}
 Suppose $dY_t =(\sigma^2 Y_t^3 -aY_t)dt+\sigma Y_t^2dW_t$. Apply It\^o's lemma to $X_t = -1/Y_t$ to get a simpler SDE for $X_t$,
 \[dX_t = -aX_tdt+\sigma dW_t\ .\]
 Notice that $Y_t$'s SDE doesn't satisfy the criterion of Theorem \ref{thm:sdeEU}, but though the change of variables we see that $Y_t$ is really a function of $X_t$ that is covered by the theorem. 
\end{example}

%%%%%%%%%%%%%%%%%%%%%%%%%%%%%%%%%%%%%%%%%%%%%%
\section{Multivariate It\^o Lemma}
Let $W_t = (W_t^1,W_t^2,\dots,W_t^n)$ be an $n$-dimensional Brownian motion such that $\frac 1t\mathbb EW_t^iW_j^j = \indicator{i=j}$. Now suppose that we also have a system of SDEs, $X_t = (X_t^1,X_t^2,\dots,X_t^m)$ (with  $m$ possibly not equal to $n$) such that
\[dX_t^i =a_i(X_t)dt+\sum_{j=1}^nb_{ij}(X_t)dW_t^j\qquad\hbox{for each $i\leq m$}\]
where $\alpha_i:\mathbb R^m\rightarrow \mathbb R$ are the drift coefficients, and the diffusion coefficients $b_{ij}:\mathbb R^m\rightarrow \mathbb R$ are such that $b_{ij}=b_{ji}$ and
\[\sum_{i,j=1}^m\sum_{\ell=1}^nb_{i\ell}(x)b_{j\ell}(x)v_iv_j >0\qquad\forall v\in \mathbb R^m\hbox{ and }\forall x\in\mathbb R^m, \]
i.e. for each $x\in\mathbb R^m$ the covariance matrix is positive definite. For some differentiable function, there is a multivariate version of the It\^o lemma.
\begin{lemma}\textbf{(The Multivariate It\^o Lemma).} Let $f:\mathbb R^+\times\mathbb R^m\rightarrow \mathbb R$ with at least one derivative in the first argument and at least two derivatives in the remaining $m$ arguments. The differential of $f(t,X_t)$ is 
\[df(t,X_t) = \left(\frac{\partial}{\partial t}+\sum_{i=1}^m\alpha_i(X_t)\frac{\partial}{\partial x_i}+\frac 12\sum_{i,j=1}^m\mathbf b_{ij}(X_t)\frac{\partial^2}{\partial x_i\partial x_j}\right)f(t,X_t)dt\]

\begin{equation}
\label{eq:multiVarIto}
+\sum_{i=1}^m\sum_{\ell=1}^nb_{i\ell}(X_t)\frac{\partial}{\partial x_i}f(t,X_t)dW_t^\ell\
\end{equation}
where $\mathbf b_{ij}(x) = \sum_{\ell=1}^nb_{i\ell}(x)b_{j\ell}(x)$.
\end{lemma}

\noindent Equation \eqref{eq:multiVarIto} is essentially a 2nd-order Taylor expansion like the univariate case of equation \eqref{eq:itoLemma}. Of course, Theorem \ref{thm:sdeEU} still applies to to the system of SDEs (make sure $a_i$ and $b_{ij}$ have linear growth and are Lipschitz continuous for each $i$ and $j$), and in the multidimensional case it is also important whether or not $n\geq m$, and if so it is important that there is some constant $c$ such that $\inf_x\mathbf b(x)>c>0$. If there is not such constant $c>0$, then we are possibly dealing with a system that is \textit{degenerate} and there could be (mathematically) technical problems. \\

\noindent\textbf{Correlated Brownian Motion.} Sometimes the multivariate case is formulated with a correlation structure among the $W_t^i$'s, in which the It\^o lemma of equation \eqref{eq:multiVarIto} will have extra correlation terms. Suppose there is correlation matrix,
\[\rho = \begin{pmatrix}1&\rho_{12}&\rho_{13}&\dots&\rho_{1n}\\
\rho_{21}&1&\rho_{23}&\dots&\rho_{2n}\\
\rho_{31}&\rho_{32}&1&\dots&\rho_{3n}\\
\vdots&&&\ddots&\\
\rho_{n1}&\rho_{n2}&\rho_{n3}&\dots&1
\end{pmatrix}\]
where $\rho_{ij}=\rho_{ji}$ and such that $\frac 1t\mathbb EW_t^iW_t^j=\rho_{ij}$ for all $i,j\leq n$. Then equation \eqref{eq:multiVarIto} becomes
\[df(t,X_t) = \left(\frac{\partial}{\partial t}+\sum_{i=1}^m\alpha_i(X_t)\frac{\partial}{\partial x_i}+\frac 12\sum_{i,j=1}^m\sum_{\ell,k=1}^n\rho_{\ell k}b_{i\ell}(X_t)b_{kj}(X_t)\frac{\partial^2}{\partial x_i\partial x_j}\right)f(t,X_t)dt\]
\[
+\sum_{i=1}^m\sum_{\ell=1}^nb_i(X_t)\frac{\partial}{\partial x_i}f(t,X_t)dW_t^\ell\ .\]
\begin{example} \textbf{(Bivariate Example).} Consider the case when $n=m=2$, with
\begin{align*}
&dX_t=a(X_t)dx+b(X_t)dW_t^1\\
&dY_t = \alpha(Y_t)dt+\sigma(Y_t)dW_t^2
\end{align*}
and with $\frac 1t\mathbb EW_t^1W_t^2 = \rho$. Then, 
\begin{align*}
df(t,X_t,Y_t)&= \frac{\partial}{\partial t}f(t,X_t,Y_t)dt\\
&\\
&+\underbrace{\left(\frac{b^2(X_t)}{2}\frac{\partial^2}{\partial x^2}+a(X_t)\frac{\partial}{\partial x}\right)f(t,X_t,Y_t)dt+b(X_t)\frac{\partial}{\partial x}f(t,X_t,Y_t)dW_t^1}_{\hbox{$X_t$ terms}}\\
&+\underbrace{\left(\frac{\sigma^2(Y_t)}{2}\frac{\partial^2}{\partial y^2}+\alpha(Y_t)\frac{\partial}{\partial y}\right)f(t,X_t,Y_t)dt+\sigma(Y_t)\frac{\partial}{\partial y}f(t,X_t,Y_t)dW_t^2}_{\hbox{$Y_t$ terms}}\\
&+\underbrace{\rho\sigma(X_t)b(X_t)\frac{\partial^2}{\partial x\partial y}f(t,X_t,Y_t)dt}_{\hbox{cross-term.}}
\end{align*}
\end{example}

%%%%%%%%%%%%%%%%%%%%%%%%%%%%%%%%%%%%%%%%%%%%%%
\section{The Feynman-Kac Formula}
If you one could identify the fundamental link between asset pricing and stochastic differential equations, it would be the Feynman-Kac formula. The Feynman-Kac formula says the following:
\begin{proposition}\label{prop:feynmanKac}\textbf{(The Feynman-Kac Formula).} Let the function $f(x)$ be bounded, let $\psi(x)$ be twice differentiable with compact support\footnote{Compact support of a function means there is a compact subset $K$ such that $\psi(x) = 0$ if $x\notin K$. For a real function of a scaler variable, this means there is a bound $M<\infty$ such that $\psi(x)=0$ if $|x|>M$.} in $K\subset \mathbb R$, and let the function $q(x)$ is bounded below for all $x\in\mathbb R$, and let $X_t$ be given by the SDE
\begin{equation}
\label{eq:FCsde}
dX_t = a(X_t)dt+b(X_t)dW_t\ .
\end{equation}
\begin{itemize}
\item For $t\in[0,T]$, the Feynman-Kac formula is
\begin{equation}
\label{eq:FCformula}
v(t,x) = \mathbb E\left[ \int_t^Tf(X_s)e^{-\int_t^sq(X_u)du}ds+e^{-\int_t^Tq(X_s)ds}\psi(X_T)\Big|X_t=x\right]
\end{equation}
and is a solution to the following partial differential equation (PDE):
\begin{eqnarray}
\label{eq:FCpde}
\frac{\partial}{\partial t}v(t,x) +\left( \frac{b^2(x)}{2}\frac{\partial^2}{\partial x^2}+a(x)\frac{\partial}{\partial x}\right)v(t,x) - q(x)v(t,x)+f(t,x) &=&0\\
\label{eq:FCpdeTC}
v(T,x)&=&\psi(x)\ .
\end{eqnarray}
\item If $\omega(t,x)$ is a bounded solution to equations \eqref{eq:FCpde} and \eqref{eq:FCpdeTC} for $x\in K$, then $\omega(t,x) = v(t,x)$.
\end{itemize}

\end{proposition}

\noindent The Feynman-Kac formula will be instrumental in pricing European derivatives in the coming sections. For now it is important to take note of how the SDE in \eqref{eq:FCsde} relates to the formula \eqref{eq:FCformula} and to the PDE of \eqref{eq:FCpde} and \eqref{eq:FCpdeTC}. It is also important to conceptualize how the Feynman-Kac formula might be extended to the multivariate case. In particular for scalar solutions of \eqref{eq:FCpde} and \eqref{eq:FCpdeTC} that are expectations of a multivariate process $X_t\in\mathbb R^m$, the key thing to realize is that $x$-derivatives in \eqref{eq:FCpde} are the same is the $dt$-terms in the It\^o lemma. Hence, multivariate Feynman-kac can be deduced from the multivariate It\^o lemma.
%%%%%%%%%%%%%%%%%%%%%%%%%%%%%%%%%%%%%%%%%%%%%%
\section{Girsanov Theorem}
Another important link between asset pricing and stochastic differential equations is the Girsanov theorem, which provides a means for defining the equivalent martingale measure.

For $T<\infty$, consider a Brownian motion $(W_t)_{t\leq T}$, and consider another process $\theta_t$ that does not anticipate future outcomes of the $W$ (i.e. given the filtration $\mathcal F_t^W$ generated by the history of $W$ up to time $t$, $\theta_t$ is adapted to $\mathcal F_t^W$). The first feature to the Girsanov theorem is the Dolean-Dade exponent:
\begin{equation}
\label{eq:doleanExp}
Z_t \doteq \exp\left(-\frac 12\int_0^t\theta_s^2ds+\int_0^t\theta_sdW_s\right)
\end{equation}
A sufficient condition for the application of Girsanov theorem is \textbf{the Novikov condition,} 
\[\mathbb E\exp\left(\frac 12\int_0^T\theta_s^2ds\right)<\infty\ .\]
Given the Novikov condition, the process $Z_t$ is a martingale on $[0,T]$ and a new probability measure is defined using the density
\begin{equation}
\label{eq:Qmeasure}
\frac{d\mathbb Q}{d\mathbb P}\Big|_t=Z_t\qquad\forall t\leq T\ ;
\end{equation}
in general $Z$ may not be a true martingale but only a \textit{local martingale} (see Appendix \ref{app:martingalesStoppingTimes} and \cite{harrisonPliska}). The Girsanov Theorem is stated as follows:
\begin{theorem}
\label{thm:girsanov}
\textbf{(Girsanov Theorem).} If $Z_t$ is a true martingale on $[0,T]$ then the process $\widetilde W_t = W_t-\int_0^t\theta_sds$ is Brownian motion under the measure $\mathbb Q$ on $[0,T]$.
\end{theorem}

\begin{remark}The Novikov condition is only sufficient and not necessary. There are other sufficient conditions such as the Kazamaki condition, $\sup_{\tau\in\mathcal T}\mathbb Ee^{\frac12\int_0^{T\wedge\tau}\theta_sdW_s}<\infty$ where $\mathcal T$ is the set of finite stopping times. \end{remark}
\begin{remark}It is important to have $T<\infty$. In fact, the Theorem does not hold for infinite time.\end{remark}

\begin{example}\textbf{(EMM for Asset Prices of Geometric Brownian Motion).} The important example for finance the (unique) EMM for the geometric Brownian. Let $S_t$ be the price of an asset,
\[\frac{dS_t}{S_t}=\mu dt+\sigma dW_t\ ,\]
and let $r\geq 0$ be the risk-free rate of interest. For the exponential martingale
\[Z_t =\exp\left(-\frac t2\left(\frac{r-\mu}{\sigma}\right)^2+\frac{r-\mu}{\sigma}W_t\right)\ ,\]
the process $ W_t^Q \doteq \frac{\mu-r}{\sigma}t+W_t$ is $\mathbb Q$-Brownian motion, and the price process satisfies,
\[\frac{dS_t}{S_t}=rdt+\sigma dW_t^Q\ .\]
Hence
\[S_t = S_0\exp\left(\left(r-\frac 12\sigma^2\right)t+\sigma W_t^Q\right)\]
and $S_te^{-rt}$ is a $\mathbb Q$-martingale.
\end{example}

\chapter{The Black-Scholes Theory}
\label{chapt:blackScholes}
This section builds a pricing theory around the assumptions of no-arbitrage with perfect liquidity and trades occurring in continuous time. The Black-Scholes model is a complete market, and there turns out to be a fairly general class of partial differential equations (PDEs) that can price many contingent claims. The focus of Black-Scholes theory is often on European call and put options, but exotics such as the Asian and the exchange option are simple extensions of the basic formulae. The issues with American options are covered later in Section \ref{sec:americans}.

%%%%%%%%%%%%%%%%%%%%%%
\section{The Black-Scholes Model}

The Black-Scholes model assumes a market consisting of a single risky asset and a risk-free bank account. This market is given by the equations
\begin{eqnarray*}
\frac{dS_t}{S_t}&=&\mu dt+\sigma dW_t\qquad\hbox{geometric Brownian-Motion}\\
dB_t&=&rB_tdt\qquad\qquad\hbox{non-stochastic}
\end{eqnarray*}
where $W_t$ is Brownian motion as described in Chapter \ref{chapt:brownianMotion} and the interpretation of the parameters is as follows:
\begin{align*}
&\mu\hbox{ is the expected rate of return in the risk asset,}\\
&\sigma>0\hbox{ is the volatility of the risky asset,}\\
&r\geq0\hbox{ is the bank's rate of interest.}
\end{align*}
It turns out that this market is particularly well-suited for pricing options and other variations, as well as analyzing basic risks associated with the writing of such contracts. Even though this market is an oversimplification of real life, it is still remarkable how a such a parsimonious model is able to capture so much of the very essence of the risky behavior in the markets. In particular, the parameter $\sigma$ will turn out to be a hugely important factor in secondary markets for options, swaps, etc.

The default focus in these notes will be the European call option with strike $K$ and maturity $T$, that is, a security that pays
\[(S_T-K)^+\doteq\max(S_T-K,0)\ ,\]
with strike and contract price agreed upon at some earlier time $t<T$. In general, the price of any European derivative security with payoff $\psi(S_T)$ (i.e. a derivative with payoff determined by the terminal value of risky asset) will be a function of the current time and the current asset price,
\[C(t,S_t) = \hbox{ price of derivative security.}\]
The fact that the price can be written a function of $t$ and $S_t$ irrespective of $S$'s history is due to the fact the model is Markov. Through an arbitrage argument we will arrive at a PDE for the pricing function $C$. Pricing equations for general non-European derivatives (such as the Asian option discussed in Section \ref{sec:asian}) are determined on a case-by-case basis.

%%%%%%%%%%%%%%%%%%%%%%
\section{Self-Financing Portfolio}
Let $V_t$ denote the \$-value of a portfolio with shares in the risk asset and the rest of it's value in the risk-free bank account. At any time the portfolio can be written as
\[V_t=\alpha_tS_t+\beta_t\]
where $\alpha_t$ is the number of shares in $S_t$ (could be any real number) and $\beta_t$ is the \$-amount in bank. The key characteristic that will be associated with $V$ throughout these notes is the following condition:

\begin{definition}
\label{def:selfFin}
The portfolio $V_t$ is \textbf{self-financing} if 
\[dV_t = \alpha_tdS_t+r\beta_tdt\]
with $\beta_t = V_t-\alpha_tS_t$.
\end{definition}

\noindent The self-financing condition is not entirely obvious at first, but it helps to think of one's personal decision-making in a financial market. Usually, one chooses a portfolio allocation in stocks and bonds, and then allows a certain amount of change to occur in the market before adjusting their allocation. If you don't remove any cash for consumption and you don't inject any cash for added investment, then your portfolio is self-financing. Indeed, in discrete time the self-financing condition is
\[V_{t_{n+1}} = V_{t_n}+\alpha_{t_n}(S_{t_{n+1}}-S_{t_n})+(e^{r\Delta t}-1)\beta_{t_n}\]
where $\Delta t = t_{n+1}-t_n$. This becomes the condition described in Definition \ref{def:selfFin} as $\Delta t\searrow 0$.

%%%%%%%%%%%%%%%%%%%%%%
\section{The Black-Scholes Equation}
For general functions $f(t,s)$, the It\^o lemma for the geometric Brownian motion process is 

\[df(t,S_t) = \left(\frac{\partial}{\partial t}+\mu S_t\frac{\partial}{\partial s}+\frac{\sigma^2S_t^2}{2}\frac{\partial^2}{\partial s^2}\right)f(t,S_t)dt+\sigma S_t\frac{\partial}{\partial s}f(t,S_t)dW_t\]
(recall equation \eqref{eq:itoLemma} from Chapter \ref{chapt:brownianMotion}). Hence, applying the It\^o lemma to the price function $C(t,S_t)$, the dynamics of the option and the self-financing portfolio are
\begin{eqnarray}
\label{eq:dC}
dC(t,S_t)&=&\left(\frac{\partial}{\partial t}+\mu S_t\frac{\partial}{\partial s}+\frac{\sigma^2S_t^2}{2}\frac{\partial^2}{\partial s^2}\right)C(t,S_t)dt+\sigma S_t\frac{\partial}{\partial s}C(t,S_t)dW_t\\
\nonumber
&&\\
\label{eq:dV}
dV_t &=&\alpha_tdS_t+r\beta_tdt\ .
\end{eqnarray}
The idea is to find $\alpha_t$ that can be known to us at time $t$ (given our observed history of prices) so that $V_t$ replicates $C(t,S_t)$ as closely as possible. Setting $\alpha_t=\frac{\partial}{\partial s}C(t,S_t)$ and $\beta_t = V_t-\alpha_tS_t$, then buying the portfolio and shorting $C$ gets a risk-less portfolio
\[d(V_t-C(t,S_t)) = r\left(V_t-S_t\frac{\partial}{\partial s}C(t,S_t)\right)dt-\left(\frac{\partial}{\partial t}+\frac{\sigma^2S_t^2}{2}\frac{\partial^2}{\partial s^2}\right)C(t,S_t)dt\]
and by arbitrage arguments, this must be equal to the risk-free rate,
\[=r(V_t-C(t,S_t))dt\ .\]
Hence, we arrive at the Black-Scholes PDE
\begin{equation}
\label{eq:BSpde}
\left(\frac{\partial}{\partial t}+\frac{\sigma^2s^2}{2}\frac{\partial^2}{\partial s^2}+rs\frac{\partial}{\partial s}-r\right)C(t,s) = 0
\end{equation}
with $C(T,s) = \psi(s)$ (recall we denote payoff function for general European claim with function $\psi(s)$).

The power of the Black -Scholes PDE is that it replicates perfectly. Observe: if $V_0 = C(0,S_0)$, then 
\[d(V_t-C(t,S_t)) = 0\qquad\forall t\leq T,\]
and so $V_T = C(T,S_T) =\psi(S_T)$. In fact, it can be shown that any contingent claim (not just Europeans) is replicable under the Black-Scholes model. Hence, the market is complete.
%%%%%%%%%%%%%%%%%%%%%%
\section{Feynman-Kac, the EMM, \& Heat Equations}
Feynman-Kac is a probabilistic formula for solving PDEs like \eqref{eq:BSpde}. It also has financial meaning because it explicitly provides a unique equivalent martingale measure (EMM). Since we have assume no-arbitrage, the 1st Fundamental theorem of asset pricing (see Section \ref{sec:intro}) necessarily asserts the existence of an EMM. The structure of this probability measure is given to us by the Feynman-Kac formula:
\begin{proposition}
\label{prop:FC}
\textbf{(Feynman-Kac).} The solution to the Black-Scholes PDE of \eqref{eq:BSpde} is the expectation 
\[C(t,s) = e^{-r(T-t)}\mathbb E^Q[\psi(S_T)|S_t=s]\]
where $\mathbb E^Q$ is an EMM under which $e^{r(T-t)}S_t$ is a martingale,
\[dS_t = rS_tdt+\sigma S_tdW_t^Q\ ,\]
with $W_t^Q \doteq \frac{\mu-r}{\sigma}t+W_t$ being Brownian motion under the EMM.
\end{proposition}

\noindent\textbf{Non-Smooth Payoffs.} For call options, the function $\psi$ is not twice differentiable nor does it have compact support, so Proposition \ref{prop:FC} is not a direct application of Feynman-Kac as stated in Proposition \ref{prop:feynmanKac} of Chapter \ref{chapt:brownianMotion}. There needs to be a further massaging of PDE to show that the formula holds for this special case. It is quite technical, but the end result is that Feynman-Kac applies to most payoffs of financial assets for log-normal models.\\

\noindent \textbf{Uniqueness.} The uniqueness of the EMM in Proposition \ref{prop:FC} can be argued by using the uniqueness of solutions to \eqref{eq:BSpde}. The conclusion that the EMM is unique and that the market is complete.\\

\noindent\textbf{Relationship with Heat Equation.} The Black-Scholes PDE \eqref{eq:BSpde} is a type of heat equation from physics. The basic heat equation is
\[\frac{\partial}{\partial t}u(t,x) = \sigma^2\frac{\partial^2}{\partial x^2}u(t,x)\]
with some initial condition $u|_{t=0} = f$. Solutions to the heat equation are interpreted as the evolution of Brownian motion's probability distribution. In the same manner that passage of time will coincide with the diffusion of heat from a source, the heat equation can describe the diffusion of possible trajectories of Brownian motion away from their common starting point of $W_0 = 0$.

Equation \eqref{eq:BSpde} obviously has some extra term and an `$s$' in front of the 2nd derivative, but a change of variables of $\tau=T-t$ and $x=\log(s)$ leads to a representation of the solution as
\[\widetilde C(\tau,x) = C(T-\tau,e^x)\]
where $C$ solves the Black-Scholes PDE. Doing the calculus we arrive at a more basic PDE for $\widetilde C$,
\[\frac{\partial}{\partial\tau}\widetilde C(\tau,x) = \frac{\sigma^2}{2}\frac{\partial^2}{\partial x^2}\widetilde C(\tau,x)+\left(b\frac{\partial}{\partial x}-r\right)\widetilde C(\tau,x)\]
with initial condition $C(0,x) = \psi(e^x)$, and with $b=\frac{2r-\sigma^2}{2}$. Hence Black-Scholes is a heat equation with drift $b\frac{\partial}{\partial x}\widetilde C(\tau,x)$, and decay $r\widetilde C(\tau,x)$.
%%%%%%%%%%%%%%%%%%%%%%
\section{The Black-Scholes Call Option Formula}
Let $\psi(s) = (s-K)^+$. From Feynman-Kac we have
\[C(t,s) = e^{-r(T-t)}\mathbb E^Q[(S_T-K)^+|S_t=s]\]
with $dS_t = rS_tdt+\sigma S_tdW_t^Q$. Through a verification with It\^o's Lemma we can see that underlying's value at time of maturity can be written as a log-normal random variable,
\[S_T = S_t\exp\left(\left(r-\frac 12\sigma^2\right)(T-t)+\sigma (W_T-W_t)\right)\ .\]
And so $\log(S_T/S_t)$ is in fact normally distributed under the risk-neutral measure,
\[\log(S_T/S_t)\sim\mathcal N\left( \left(r-\frac 12\sigma^2\right)(T-t),\sigma^2(T-t)\right)\ .\]
Hence, we compute the expectation for $\psi(s) = (s-K)^+$,

\[\mathbb E^Q\{(S_T-K)^+|S_t=s\}\]

\[ =\underbrace{ \frac{S_t}{\sqrt{2\pi\sigma^2(T-t)}}\int_{\log(K/S_t)}^\infty e^xe^{-\frac12\left(x- \left(r-\frac 12\sigma^2\right)(T-t)\right)^2/(\sigma^2(T-t))}dx}_{=(\dagger)}\]

\[-\underbrace{K\mathbb Q(\log(S_T/S_t)>\log(K/S_t))}_{=(\star)} \]
where $\mathbb Q$ is the risk-neutral probability measure.\\

\noindent $\mathbf{(\dagger).}$ First compute $(\dagger)$ (W.L.O.G. for $t=0$):

\[(\dagger) = \frac{S_0}{\sqrt{2\pi\sigma^2 T}}\int_{\log(K/S_0)}^\infty e^x e^{-\frac 12\left(\frac{x-(r-\frac 12\sigma^2)T}{\sigma\sqrt T}\right)^2}dx\]
\[= \frac{S_0}{\sqrt{2\pi\sigma^2 T}}\int_{\log(K/S_0)}^\infty  e^{-\frac{1}{2\sigma^2T}\left(-2x^2\sigma^2T+x^2-2x(r-.5\sigma^2)T+((r-.5\sigma^2)T)^2\right)}dx\]
\[= \frac{S_0}{\sqrt{2\pi\sigma^2 T}}\int_{\log(K/S_0)}^\infty  e^{-\frac{1}{2\sigma^2T}\left(x^2-2x(r+.5\sigma^2)T+((r-.5\sigma^2)T)^2\right)}dx\]
\[= \frac{S_0e^{rT}}{\sqrt{2\pi\sigma^2 T}}\int_{\log(K/S_0)}^\infty  e^{-\frac{1}{2\sigma^2T}\left(x^2-(r+.5\sigma^2)T\right)^2}dx\]

change of variables $v = \frac{x-(r+.5\sigma^2)T}{\sigma\sqrt T}$, $dv = dx/(\sigma\sqrt T)$, so that

\[(\dagger) =  \frac{S_0e^{rT}}{\sqrt{2\pi}}\int_{(\log(K/S_0)-(r+.5\sigma^2)T)/(\sigma\sqrt T)}^\infty  e^{\frac 12 v^2}dv\]

\[= S_0e^{rT}\left(1 - \frac{1}{\sqrt{2\pi}}\int_{-\infty}^{(\log(K/S_0)-(r+.5\sigma^2)T)/(\sigma\sqrt T)} e^{\frac 12 v^2}dv\right)\]

\[=S_0e^{rT}\left(1 - N(-d_1)\right)\]
where $d_1 = \frac{\log(S_0/K)+(r+.5\sigma^2)T}{\sigma\sqrt T}$ and $N(\cdot)$ is the standard normal CDF. But the normal CDF has the property that $N(-x) = 1-N(x)$, so 
\[(\dagger) = S_0e^{rT}N(d_1)\ .\]

\noindent $\mathbf{(\star)}.$ Then, computing $(\star)$ is much simpler,
\[(\star) = K\mathbb Q\left(\log(S_T/S_0)\geq \log(K/S_0)\right)\]
\[=K\mathbb Q\left(\frac{\log(S_T/S_0)-(r-.5\sigma^2)T}{\sigma\sqrt T}\geq\frac{ \log(K/S_0)-(r-.5\sigma^2)T}{\sigma\sqrt T}\right)\]
\[=K\left(1 - \mathbb Q\left(\frac{\log(S_T/S_0)-(r-.5\sigma^2)T}{\sigma\sqrt T}\leq\frac{ \log(K/S_0)-(r-.5\sigma^2)T}{\sigma\sqrt T}\right)\right)\]

\[=K\left(1-N(-d_2) \right)= KN(d_2)\]
where $d_2 =\frac{ \log(S_0/K)+(r-.5\sigma^2)T}{\sigma\sqrt T}=d_1-\sigma\sqrt T$. Hence, we have the Black-Scholes formula for a European Call Option,

\begin{proposition}
\label{prop:BScall}
\textbf{(Black-Scholes Call Option Formula).} The call option on $S_T$ with strike $K$ at time $t$ with price $S_t$ is given by

\[C(t,S_t) = S_tN(d_1)-Ke^{-r(T-t)}N(d_2)\]
where $N(\cdot)$ is the standard normal CDF and
\begin{align*}
&d_1  = \frac{\log(S_t/K)+(r+.5\sigma^2)(T-t)}{\sigma\sqrt{T-t}} \\
&d_2 = d_1-\sigma\sqrt{T-t}\ .
\end{align*}
\end{proposition}

A plot of the Black-Scholes call option price with $K=50$, $r=.02$, $T=3/12$, and $\sigma=.2$ with varying $S_0$ is shown in Figure \ref{fig:callPrice}. 

\begin{figure}[htbp] %  figure placement: here, top, bottom, or page
   \centering
   \includegraphics[width=5in]{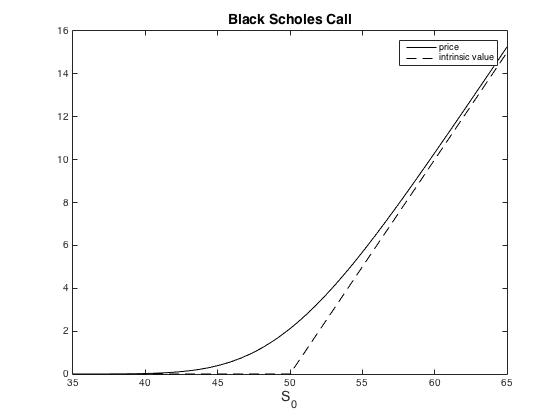} 
   \caption{The Black-Scholes call price with $K=50$, $r=.02$, $T=3/12$, and $\sigma=.2$. Intrinsic value refers to the payoff if exercised now, $(S_0-K)^+$. }
   \label{fig:callPrice}
\end{figure}

%%%%%%%%%%%%%%%%%%%%%%
\section{Put-Call Parity and the Put Option Formula}
It is straight forward to verify the relationship
\begin{equation}
\label{eq:PCterminal}
(S_T-K)^+-(K-S_T)^+ = S_T-K\ .
\end{equation}
Then applying the risk-neutral expectionan $e^{-r(T-t)}\mathbb E_t^Q$ to both sides of \eqref{eq:PCterminal} to get the put-call parity,
\begin{equation}
\label{eq:PCparity}
C(t,S_t)-P(t,S_t) = S_t-Ke^{-r(T-t)}
\end{equation}
where $C(t,S_t)$ is the price of a European call option and $P(t,S_t)$ the price of a European put option with the same strike. From put-call parity we have
\[P(t,s) = C(t,s)+Ke^{-r(T-t)}-s\]
\[=-Ke^{-r(T-t)}(N(d_2)-1)+s(N(d_1)-1)\]
\[=Ke^{-r(T-t)}N(-d_2)-sN(-d_1)\ ,\]
because $1-N(x) = N(-x)$ for any $x\in\mathbb R$.

\begin{proposition}
\label{prop:BSput}
\textbf{(Black-Scholes Put Option Formula).} The put option on $S_T$ with strike $K$ at time $t$ with price $S_t$ is given by

\[P(t,S_t) = Ke^{-r(T-t)}N(-d_2)-S_tN(-d_1)\]
where 
\begin{align*}
&d_1  = \frac{\log(S_t/K)+(r+.5\sigma^2)(T-t)}{\sigma\sqrt{T-t}} \\
&d_2 = d_1-\sigma\sqrt{T-t}\ .
\end{align*}
\end{proposition}
The Black-Scholes put option price for $K=50$, $r=.02$, $T=3/12$, and $\sigma=.2$ and varying $S_0$ is shown in Figure \ref{fig:putPrice}.

\begin{figure}[htbp] %  figure placement: here, top, bottom, or page
   \centering
   \includegraphics[width=5in]{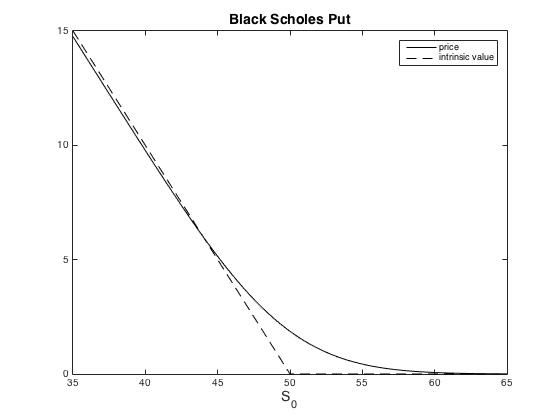} 
   \caption{The Black-Scholes put option price for $K=50$, $r=.02$, $T=3/12$, and $\sigma=.2$. Intrinsic value refers to the payoff if exercised now, $(K-S_0)^+$. }
   \label{fig:putPrice}
\end{figure}

%%%%%%%%%%%%%%%%%%%%%%
\section{Options on Futures and Stocks with Dividends}
\label{sec:dividends}
This section will explain how to compute European derivative prices on stocks with a continuously paying dividend rate, and on future prices. There is some technical issues with stocks paying dividends at a discrete times, but in the case of European options it is merely of matter of considering the future price. Dividends can be an issue in the Black-Scholes theory, particularly because they will determine whether or not Black-Scholes applies for American options (see Section \ref{sec:americans}).\\

\noindent\textbf{Continuous Dividends.} Suppose that a European option with payoff $\psi(S_T)$ is being priced in a market with 
\begin{eqnarray*}
\frac{dS_t}{S_t}&=&(\mu-q)dt+\sigma dW_t\\
dB_t&=&rB_tdt
\end{eqnarray*} 
where $q\geq 0$ is the dividend rate. Self-financing in this case has dynamics
\[dV_t = \alpha_tdS_t + r(V_t-\alpha_tS_t)dt+q\alpha_tS_tdt\ .\]
The replicating strategy from the non-dividend case applies to obtain $\alpha_t = \frac{\partial}{\partial s}C(t,S_t)$, but the arbitrage argument leads to a different equation,
\begin{equation}
\label{eq:BSpdeDividends}
\left(\frac{\partial}{\partial t}+\frac{\sigma^2s^2}{2}\frac{\partial^2}{\partial s^2}+(r-q)s\frac{\partial}{\partial s}-(r-q)\right)C(t,s) = qC(t,s)
\end{equation}
with $C(T,s) = \psi(s)$. Equation \eqref{eq:BSpdeDividends} is the Black-Scholes PDE with dividends, and is solved by 
\[C(t,S_t) = e^{-q(T-t)}C^1(t,S_t)\]
where $C^1(t,S_t)$ solves the Black-Scholes equation of \eqref{eq:BSpde} with $r$ replaced by $r-q$. For example, a European call option on a stock with dividends:
\begin{proposition}\textbf{(Call Option on Stock with Dividends).}
\[C(t,S_t) = S_te^{-q(T-t)}N(d_1)-Ke^{-r(T-t)}N(d_2)\]
where 
\begin{align*}
&d_1  = \frac{\log(S_t/K)+(r-q+.5\sigma^2)(T-t)}{\sigma\sqrt{T-t}} \\
&d_2 = d_1-\sigma\sqrt{T-t}\ .
\end{align*}
\end{proposition}

\noindent\textbf{Futures.} Now we turn out attention to futures. Suppose that the spot price on a commodity is given by 
\[\frac{dS_t}{S_t}= \mu dt+\sigma dW_t\ .\]
From arbitrage arguments we know that the future price on $S_T$ is
\[F_{t,T} \doteq S_te^{r(T-t)}\]
with $F_{T,T}=S_T$. However, $S_te^{r(T-t)}$ is a martingale under the EMM, and so $F_{t,T}$ is also a martingale with 
\[\frac{dF_{t,T}}{F_{t,T}}=\sigma dW_t^Q\ ,\]
and the price of a derivative in terms of $F_{t,T}$ is $C(t,F_{t,T}) = e^{-r(T-t)}\widetilde C(t,F_{t,T})$ where $\widetilde C$ satisfies
\begin{equation}
\label{eq:BSpdeFutures}
\left(\frac{\partial}{\partial t}+\frac{\sigma^2x^2}{2}\frac{\partial^2}{\partial x^2}\right)\widetilde C(t,x) = 0
\end{equation}
with $\widetilde C(T,x) = \psi(x)$. For example, the call option on the future:
\begin{proposition}\textbf{(Call Option on Future).}
\[C(t,F_{t,T}) = e^{-r(T-t)}\left(F_{t,T}N(d_1)-KN(d_2)\right)\]
where 
\begin{align*}
&d_1  = \frac{\log(F_{t,T}/K)+.5\sigma^2(T-t)}{\sigma\sqrt{T-t}} \\
&d_2 = d_1-\sigma\sqrt{T-t}\ .
\end{align*}
\end{proposition}

\noindent\textbf{The Similarities.} Futures contracts certainly have fundamental differences from stocks paying dividends, and vice versa. But the Black-Scholes PDE for pricing options on futures is like that for a stock with dividend rate $r$. Alternatively, the future on a dividend paying stock is
\[F_{t,T} = S_te^{(r-q)(T-t)}\]
which is a non-dividend paying asset and can hence be priced using equation \eqref{eq:BSpdeFutures}.\\

\noindent\textbf{Discrete-Time Dividends.} This interpretation of dividend-paying stocks as futures is useful when dividends are paid at discrete times. Discrete time dividends are described as
\[\log(S_t/S_0) = \left(r-\frac 12 \sigma^2\right)t+\sigma W_t^Q-\sum_{i=0}^{n(t)}\delta_i\]
where $\delta_i$ is a proportional dividend rate, and $n(t)$ is the number of dividends paid up to time $t$. The future price of dividend paying stock $S_T$ is
\[F_{t,T} = S_t\exp\left(\left(r-\frac 12\sigma^2\right)(T-t)-\sum_{i=n(t)+1}^{n(T)}\delta_i\right)\ .\]
Call optionΠon $F_{T,T}$ can be priced with equation \eqref{eq:BSpdeFutures}.

\chapter{The Black-Scholes Greeks}
\label{sec:greeks}
The Greeks a set of letters labeling the quantities of risk associate with small changes in various inputs and model parameters. The Greeks have importance in risk management where hedging and risks are determined by how much/little of the Greeks are exposed on a portfolio. There are 5 main Greeks associated with the Black-Scholes, one of which has already been instrumental in setting up the hedging portfolio, namely the $\Delta$. The other Greeks can be equally as important in determining the quality of a hedge.

%%%%%%%%%%%%%%%%%%%%%%
\section{The Delta Hedge}

Recall the derivation of the Black-Scholes PDE of equation \eqref{eq:BSpde} on page \pageref{eq:BSpde}. In particular, recall the hedging allocation in the underlying was the first derivative. This is the Black-Scholes $\Delta$,
\[\Delta(t,s) \doteq \frac{\partial}{\partial s}C(t,s)\ .\]
In general, the hedge in the underlying what's referred to when someone talks about being long, short or neutral in $\Delta$. With continuos trading the hedging portfolio is a perfect replication because it continuously rebalances to remain $\Delta$-neutral. A long position in $\Delta $ would be when the hedging portfolio holds the underlying in excess of the $\Delta$ (e.g. a covered call), and a short position in $\Delta $ would be a hedge with less (e.g. a naked call).

For the European call and put options, the Black-Scholes $\Delta$ is

\begin{align*}
&\Delta_{call}(t,s) = N(d_1)\\
&\Delta_{put}(t,s) = -N(-d_1)=\Delta_{call}(t,s)-1
\end{align*}
which are obtained by differentiating the formulas in Propositions \ref{prop:BScall} and \ref{prop:BSput}, respectively. Plots of the $\Delta$'s for varying $s$ are shown in Figure \ref{fig:deltas}. Notice that all the action in the $\Delta$-hedge occurs when the underlying price is near the strike price. Intuitively, a call or put that is far money is efficiently hedged by either holding a share of the underling or simply holding the cash, and the risk is low because there very little probability of the underlying changing dramatically enough to effect your position. Hence, the $\Delta$-hedge is like a covered call/put when the derivative is far from the money.

\begin{figure}[htbp] %  figure placement: here, top, bottom, or page
   \centering
   \includegraphics[width=5in]{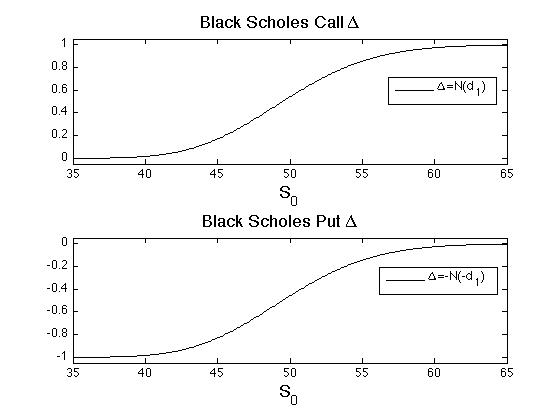} 
\caption{\textbf{Top:} The Black-Scholes $\Delta$ for a European option with strike $K=50$, time to maturity $T= 3/12$, interest rate $r=.02$, and volatility $\sigma = .2$. \textbf{Bottom:} The put option $\Delta$.}
\label{fig:deltas}
\end{figure}

%%%%%%%%%%%%%%%%%%%%%%
\section{The Theta}
The next Greek to discuss is the $\Theta$, which measures the sensitivity to time. In particular, 
\[\Theta(t,s) = \frac{\partial}{\partial t}C(t,s)\ .\] 
The financial interpretation is that $\Theta$ measures the decay of the time value of the security. Intrinsic value of a security is fairly clear to understand: it is the value of the derivative at today's price of the underlying. However, derivatives have an element of \textit{time value} on top of the intrinsic value, and in cases where the intrinsic value is zero (e.g. an out-of-the-money option) the derivative's entire cost is its time value. 

For option contracts, time value quantifies the cost to be paid for having the option to buy/sell, as opposed to taking a position that is long/short the underlying with no option involved. In most cases the time value is positive and decreases with time, hence $\Theta<0$, but there are contracts (such as far in-the-money put options) that have positive $\Theta$. The $\Theta$'s for European calls and puts are 
\begin{align}
\label{eq:callTheta}
&\Theta_{call}(t,s)=-\frac{\sigma^2S_t^2}{2\sqrt{T-t}}N'(d_1)-Ke^{-r(T-t)}N(d_2)\\
\label{eq:putTheta}
&\Theta_{put}(t,s)=-\frac{\sigma^2S_t^2}{2\sqrt{T-t}}N'(d_1)+Ke^{-r(T-t)}N(-d_2)
\end{align}
and are shown in Figure \ref{fig:thetas}. To understand why a far in-the-money put option has positive $\Theta$, look at Figure \ref{fig:posTheta} and realize that the option price must with rise with time if it is below the intrinsic value.

\begin{figure}[htbp] %  figure placement: here, top, bottom, or page
   \centering
   \includegraphics[width=5in]{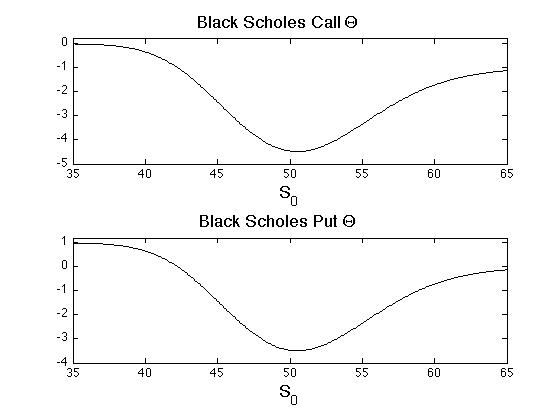} 
  \caption{\textbf{Top:} The Black-Scholes $\Theta$ for a European option with strike $K=50$, time to maturity $T= 3/12$, interest rate $r=.02$, and volatility $\sigma = .2$. \textbf{Bottom:} The put option $\Theta$.}
\label{fig:thetas}
\end{figure}

\begin{figure}[htbp] %  figure placement: here, top, bottom, or page
   \centering
   \includegraphics[width=5in]{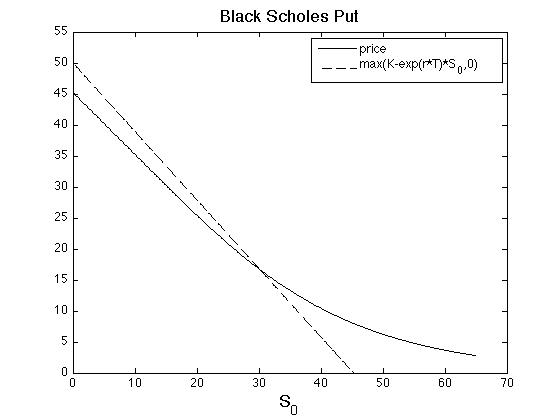} 
\caption{The Black-Scholes price of a European put option with strike $K=50$, time to maturity $T= 5$, interest rate $r=.02$, and volatility $\sigma = .2$. Notice if the price is low then the time value of exercise today is higher than the value of the option. This means that $\Theta$ of far in-the-money put options is positive.}
\label{fig:posTheta}
\end{figure}

In terms of hedging, a short position in a derivative is said to be short in $\Theta$, and a long position in the derivative is also long the $\Theta$. Hence, a $\Delta$-neutral hedge will exposed to $\Theta$, and a decrease in $\Theta$ will benefit the short position because it means that the time value of the derivative decreases at a faster rate.
%%%%%%%%%%%%%%%%%%%%%%
\section{The Gamma}

The $\Gamma$ of a security is the sensitivity of its $\Delta$-hedge to changes in the underlying, and there a handful of ways in which the $\Gamma$ explains hedging and risk management. In a nutshell, the $\Gamma$ is the higher order sensitivity of the hedging portfolio, and is essentially a convexity term for higher order hedging. 

Let $\Delta t$ be the amount of time that goes by in between hedge adjustments. From It\^o's lemma we have the following approximations of the derivative dynamics and a self-financing portfolio:
\begin{align*}
&\Delta C_t\simeq\left( \Theta_t +\frac{\sigma^2S_t^2}{2}\Gamma_t\right)\Delta t+\Delta_t\cdot\Delta S\\
&\Delta V_t\simeq \Delta_t\cdot\Delta S_t+r\left(V_t-\Delta_t\cdot S_t\right)\Delta t
\end{align*}
where $\Delta C_t = C_{t+\Delta t}-C_t$ and $\Delta S_t = S_{t+\Delta t}-S_t$, and where $\Theta_t,~\Delta_t$ and $\Gamma_t$ are the derivative's Greeks. Subtracting one from the other gets the (approximate) dynamics of a long position in the portfolio and short the derivative,

\begin{equation}
\label{eq:deltaPrem}
\Delta V_t-\Delta C_t \simeq \underbrace{r\left(V_t-\Delta_t\cdot S_t\right)\Delta t}_{\hbox{earned in bank}}+\underbrace{\left( -\Theta_t -\frac{\sigma^2S_t^2}{2}\Gamma_t\right)\Delta t}_{\hbox{premium over risk-free}}
\end{equation}
and so while the $\Delta$-hedge replicates perfectly in continuous time, there will be error if this hedge is readjusted in discrete time. However, $-\Theta_t -\frac{\sigma^2S_t^2}{2}\Gamma_t>0$ corresponds to the premium over the risk-free rate earned by the hedging portfolio. Often times $\Gamma>0$ and $\Theta<0$, and so the risk taken by rebalancing a $\Delta$-hedge in discrete time is compensated with lower (more negative) $\Theta$ and lower $\Gamma$.

For European call and put options, the $\Gamma $ is
\begin{equation}
\label{eq:gamma}
\Gamma(t,s)  =\frac{N'(d_1)}{S_t\sigma\sqrt{T-t}}
\end{equation}
which proportional to the probability density function of $\log(S_T/S_t)$ and is always positive for long positions. Furthermore, combining the $\Gamma$ in \eqref{eq:gamma} with equation \eqref{eq:callTheta} we find the premium of equation \eqref{eq:deltaPrem} is 
\[-\Theta_{call}(t,s) -\frac{\sigma^2S_t^2}{2}\Gamma(t,s) = Ke^{-r(T-t)}N(d_2)>0\ ,\]
and so the $\Delta$-hedge for a European call option earns a premium over the risk-free rate. However, the $\Delta$-hedge of the put option has a negative premium
\[-\Theta_{put}(t,s) -\frac{\sigma^2S_t^2}{2}\Gamma(t,s) = -Ke^{-r(T-t)}N(-d_2)<0\ ,\]
which means the portfolio that longs the put option and shorts the $\Delta$-hedge will earn a return over the risk-free rate.

If $\Delta$-hedging is too risky then the $\Gamma$ is used in higher order hedging, namely the $\Delta$-$\Gamma$ hedge. The idea is to consider a second derivative security $C'(t,s)$ with which to hedge. A hedge that is both $\Delta$-neutral and $\Gamma$-neutral is found by solving the equations
\begin{align*}
&C(t,S_t)= \alpha_t S_t+\beta_tC'(t,S_t)+\eta_t\\
&\Delta(t,S_t)=\alpha_t +\beta_t\Delta'(t,S_t)\\
&\Gamma(t,S_t) = \beta_t\Gamma'(t,S_t)
\end{align*}
where $\alpha_t$ is the number of contracts in the underlying, $\beta_t$ is number of contracts in $C'$, and $\eta_t$ is the \$-amount held in the risk-free bank account. This hedge is more expensive and will earn less premium over the risk-free rate,
\[\Delta V_t  - \Delta C_t \simeq r\left(V_t-\left(\Delta_t-\Delta_t'\frac{\Gamma_t}{\Gamma_t'}\right)S_t-\frac{\Gamma}{\Gamma'}C_t'\right)\Delta t+\left(\frac{\Gamma_t}{\Gamma_t'}\Theta_t'-\Theta_t\right)\Delta t\ ,\]
with $\frac{\Gamma_t}{\Gamma_t'}\Theta_t'-\Theta_t$ being the premium over the risk-free rate. A plot of the Black-Scholes $\Gamma$ is shown in Figure \ref{fig:gammaVega}.

\begin{figure}[htbp] %  figure placement: here, top, bottom, or page
   \centering
   \includegraphics[width=5in]{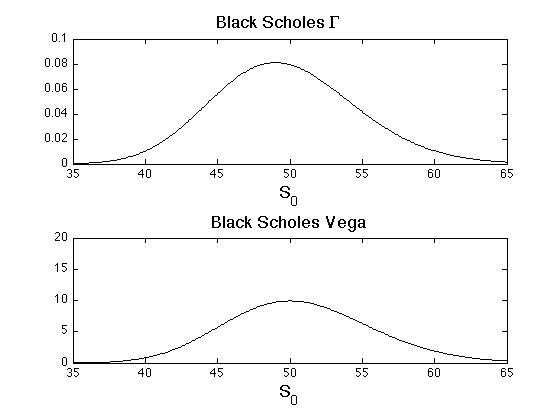} 
   \caption{\textbf{Top:} The Black-Scholes $\Gamma = \frac{N'(d_1)}{S_t\sigma\sqrt{T-t}}$, proportional to the a Gaussian probability density function of $\log(S_T/S_t)$.}
   \label{fig:gammaVega}
\end{figure}

%%%%%%%%%%%%%%%%%%%%%%
\section{The Vega}
The Vega is the Greek to describe sensitivity to $\sigma$. Of the 5 main Greeks it is the only one that corresponds to a latent parameter (although some could argue that the risk-free rate is also somewhat latent or abstract). Hence, not only is it important to have a good estimate of $\sigma$, it is also important to know if your hedging position is sensitive to estimation error. Indeed, volatility is a major question inn derivative pricing, and will be addressed in complete detail in Section \ref{sec:impVol}.

For the Black-Scholes call and puts, the Vega is
\[\frac{\partial}{\partial\sigma}C = S_tN'(d_1)\sqrt{T-t}\]
which is proportional to the $\Gamma$ and to the probability density function of $\log(S_T/S_t)$ (see Figure \ref{fig:gammaVega}). There are some portfolios that are particularly sensitive Vega, for instance a straddle consisting of a call and put with the same strike $K$ if $K$ is near the money.

%%%%%%%%%%%%%%%%%%%%%%
\section{The Rho}
The $\rho$ is the sensitivity of to changes in the risk-free rate
\[\rho = \frac{\partial}{\partial r}C\ .\]
Of the 5 main Greeks it widely considered to be the least sensitive. For the Black-Scholes call and put options, the $\rho$ is

\begin{align*}
&\rho_{call}= K(T-t)e^{-r(T-t)}N(d_2)\\
&\rho_{put}=- K(T-t)e^{-r(T-t)}N(-d_2)\ .
\end{align*}
Notice that $\rho_{call}>0$ because higher interest rates means risky assets should rise in price and hence it will be more likely that you will exercise. Similarly, $\rho_{put}<0$ because higher interest rates means it is less likely that you  will exercise. 
%%%%%%%%%%%%%%%%%%%%%%
\section{Other Greeks}

There are also other Greeks of interest in derivatives. Some of the are 

\begin{itemize}
\item $\Lambda$, $\frac{\partial}{\partial S}C\times \frac{S}{C}$, elasticity, or measure of leverage.
\item Vanna, $\frac{\partial^2}{\partial s\partial \sigma}C$, The sensitivity of the $\Delta$ to changes in volatility.
\item  Volga, $\frac{\partial^2}{\partial \sigma^2}C$, used for more elaborate volatility hedging.
\end{itemize}

\chapter{Exotic Options}
\label{sec:exotics}
This section describes some of the better-known facts regarding a few exotic derivatives. The American and the Asian options are considered `path-dependent' because the payoff of the option will vary depending on the history of the price process, whereas exchange options and binary options are not path-dependent because the payoff only depends on the asset price(s) are the terminal time.

%%%%%%%%%%%%%%%%%%%%%%%%%%%%%%%%%%

\section{Asian Options}
\label{sec:asian}
Asian option consider the average price of the underlying in the payoff,
\[\psi_T = \left(\frac 1T\int_0^TS_tdt - K\right)^+\ .\]
These options are \textit{path dependent} because the path taken by the underlying is considered. The averaging in the Asian option makes it harder to manipulate the payoff by driving the underlying's price up or down when the option is close to maturity. \\

\noindent\textbf{Delta Hedge.} In the Black-Scholes framework, the Asian option is a relatively simple extension of the hedging argument used to derive the Black-Scholes PDE. The trick is to consider an use an integrated field,
\[Z_t =\frac 1T \int_0^tS_udu\ ,\]
and consider the price
\[C(t,s,z) =\hbox{ price of Asian option given $S_t=s$ and $Z_t = z$.}\]
Let the underlying's dynamics be 
\[\frac{dS_t}{S_t}=\mu dt+\sigma dW_t\]
where $W_t$ is a Brownian motion. From It\^o's lemma we have
\[dC(t,S_t,Z_t) = \left(\frac{\partial}{\partial t}+\frac{\sigma^2S_t^2}{2}\frac{\partial^2}{\partial s^2}+\mu S_t\frac{\partial}{\partial s}+\frac{S_t}{T}\frac{\partial}{\partial z}\right)C(t,S_t,Z_t)dt+\sigma S_t\frac{\partial}{\partial s}C(t,S_t,Z_t)dW_t\ .\]
On the other hand, we have a self-financing portfolio,
\[dV_t = \alpha_tdS_t+r(V_t-\alpha_tS_t)dt\]
where $\alpha_t$ is the number of contracts in $S_t$ and $(V_t-\alpha_tS_t)$ is the \$-amount in the risk-free asset. Comparing the It\^o lemma of the option to the self-financing portfolio, we see that a portfolio that is long $V_t $ and short the Asian option will be risk-less if we have the $\Delta$-hedge
\[\alpha_t = \frac{\partial}{\partial s}C(t,S_t,Z_t)\ .\]
Furthermore, since the portfolio long $V_t$ and short the option is risk-less, it must be that it earns at the risk-free rate. Hence,
\[d(V_t-C(t,S_t,Z_t)) = r\left(V_t-S_t\frac{\partial}{\partial s}C(t,S_t,Z_t)\right)dt-\left(\frac{\partial}{\partial t}+\frac{\sigma^2S_t^2}{2}\frac{\partial^2}{\partial s^2}+\frac{S_t}{T}\frac{\partial}{\partial z}\right)C(t,S_t,Z_t)dt\]

\[=r(V_t-C(t,S_t,Z_t))dt\ .\]
Hence, the price of the Asian option satisfies the following PDE,
\begin{eqnarray}
\label{eq:asianPDE}
\left(\frac{\partial}{\partial t}+\frac{\sigma^2s^2}{2}\frac{\partial^2}{\partial s^2}+s\frac{\partial}{\partial s}+\frac sT\frac{\partial}{\partial z}-r\right)C(t,s,z)&=&0\\
\label{eq:asianPDEtc}
C(t,s,z)\Big|_{t=T}&=&\left(z-K\right)^+\ .
\end{eqnarray}
Equations \eqref{eq:asianPDE} and \eqref{eq:asianPDEtc} can be solved with a Feynman-Kac formula, 
\[C(t,s,z) = e^{-r(T-t)}\mathbb E^Q\left[\left(I_T-K\right)^+\Big|S_t=s,Z_t=z\right]\]
where $\mathbb E^Q$ is the expected value under the unique EMM $\mathbb Q$, and the asset price is $\frac{dS_t}{S_t} = rdt+\sigma dW_t^Q$ under $\mathbb Q$. We knew from Chapter \ref{chapt:blackScholes} that this market was complete, so it shouldn't come as a surprise that we were able to hedge perfectly the Asian option.\\

\noindent\textbf{PDE of Reduced Dimension.} A useful formula is the dimension reduction of the Asian option price to a function of just time and one stochastic field. Define a new function
\[\phi(t,x) \doteq \mathbb E^Q\left[\left(\frac 1T\int_t^TS_udu-x\right)^+\Big|S_t=1\right]\ .\]
Then let $\mathcal F_t$ denote the history of prices up time time $t$, and define the martingale
\[M_t \doteq \mathbb E^Q\left[\left(\frac 1T\int_0^TS_udu-K\right)^+\Big|\mathcal F_t\right]\]
\[= \mathbb E^Q\left[\left(\frac{1}{T}\int_t^TS_udu-\left(K-\frac{1}{T}\int_0^tS_udu\right)\right)^+\Big|\mathcal F_t\right]\]
\[=S_t \mathbb E^Q\left[\left(\frac{1}{T}\int_t^T\frac{S_u}{S_t}du-\frac{K-\frac{1}{T}\int_0^tS_udu}{S_t}\right)^+\Big|\mathcal F_t\right]\]

\[=S_t\phi(t,X_t)\]
where $X_t \doteq \frac{K-\frac{1}{T}\int_0^tS_udu}{S_t}$. Applying the It\^o lemma to $\xi_t$, we have
\[dX_t = -\frac 1Tdt+X_t\left(-\sigma dW_t^Q-rdt+\sigma^2dt\right)\ ,\]
and assuming that $\phi(t,x)$ is twice differentiable we apply the It\^o lemma to $\phi(t,X_t)$:
\[d\phi(t,X_t) = \left(\frac{\partial}{\partial t}+\left(-\frac 1T+X_t\left(\sigma^2-r\right)\right)\frac{\partial}{\partial x}+X_t^2\sigma^2\frac{\partial^2}{\partial x^2}\right)\phi(t,X_t)dt-X_t\sigma\frac{\partial}{\partial x}\phi(t,X_t) dW_t^Q\ .\]
Then applying It\^o's lemma to $M_t$, we have
\[dM_t = \phi(t,X_t)dS_t+S_td\phi(t,X_t) +dS_t\cdot d\phi(t,X_t)\]

\[=\phi(t,X_t)\left(rS_tdt+\sigma S_tdW_t^Q\right)+S_t\left(\frac{\partial}{\partial t}+\left(-\frac 1T+X_t\left(\sigma^2-r\right)\right)\frac{\partial}{\partial x}+X_t^2\sigma^2\frac{\partial^2}{\partial x^2}\right)\phi(t,X_t)dt\]

\[-S_tX_t\sigma\frac{\partial}{\partial x}\phi(t,X_t) dW_t^Q-X_tS_t\sigma^2\frac{\partial}{\partial x}\phi(t,X_t)dt\ , \]
and in order for $M_t$ to be a martingale the $dt$ terms must all cancel out. Therefore, we must have
\[r\phi(t,X_t)+\left(\frac{\partial}{\partial t}-\left(\frac 1T+rX_t\right)\frac{\partial}{\partial x}+X_t^2\sigma^2\frac{\partial^2}{\partial x^2}\right)\phi(t,X_t)
=0\ ,\]
which gives us the PDE of reduced dimension for a new function $\tilde \phi(t,x)\doteq e^{-r(T-t)}\phi(t,x)$,
\begin{eqnarray}
\label{eq:reducedPDE}
\left(\frac{\partial}{\partial t}-\left(\frac 1T+rx\right)\frac{\partial}{\partial x}+x^2\sigma^2\frac{\partial^2}{\partial x^2}\right)\tilde\phi(t,x)&=&0\\
\label{eq:reducedPDEtc}
\tilde\phi(t,x)\Big|_{t=T}&=&\max(-x,0)\ .
\end{eqnarray}
The PDE of equations \eqref{eq:reducedPDE} and \eqref{eq:reducedPDEtc} is of reduced dimension and is solvable with a Feynman-Kac formula,
\[\widetilde \phi(t,x) = \mathbb E[\max(-\widetilde X_T,0)|\widetilde X_t=x]\]
with $d\widetilde X_t = -\left(\frac 1T+r\widetilde X_t\right)dt+\sigma\widetilde X_t dW_t$. Finally, the option price is
\begin{equation}
\label{eq:reducedDim}
C(t,S_t,Z_t) = S_t\cdot\widetilde \phi\left(t~,~\frac{K-Z_t}{S_t}\right)\qquad\forall t\leq T\ .
\end{equation}
Equation \eqref{eq:reducedDim} is the pricing formula of reduced dimension.\\

\section{Exchange Options (Margrabe Formula)}
Consider two tradable assets, $S_t^1$ and $S_t^2$. An exchange option with strike $K$ has payoff
\[\psi(S_T^1,S_T^2) = (S_T^1-S_T^2-K)^+\]
which we price price under an EMM. In particular, for $K=0$ we have
 \begin{eqnarray}
 \nonumber
 C(t,s_1,s_2) &=& e^{-r(T-t)}\mathbb E^Q\left[(S_T^1-S_T^2)^+\Big|S_t^1=s_1,S_t^2=s_2\right]\\
\label{eq:ratioOption}
&=& e^{-r(T-t)}\mathbb E^Q\left[S_T^2\left(\frac{S_T^1}{S_T^2}-1\right)^+\Big|S_t^1=s_1,S_t^2=s_2\right]\ .
\end{eqnarray}
Options like this are found in forex markets and in commodities, the latter of which uses them to hedge risk in the risk between the retail price of a finished good verses the cost of manufacture. In electricity markets, the price of a megawatt hour of electricity needs to be hedged against the price of natural gas in what's called the \textit{spark spread}; in the soy market the cost of soy meal verses the cost of soy beans is called the \textit{crush spread}; in the oil market the cost of refined oil products verses the cost of crude oil is called the \textit{crack spread}. In commodities markets these options are called sometimes called spread options. The Margrabe formula prices spread/exchange options with $K=0$.

Suppose we have a double Black-Scholes model for the two assets,
\begin{eqnarray*}
\frac{dS_t^1}{S_t^1}&=& rdt+\sigma_1\left(\sqrt{1-\rho^2}dW_t^1+\rho dW_t^2\right)\\
\frac{dS_t^2}{S_t^2}&=& rdt+\sigma_2dW_t^2
\end{eqnarray*}
where $W_t^1$ and $W_t^2$ are independent risk-neutral Brownian motions and $\rho\in(-1,1)$. Now we define ratio in \eqref{eq:ratioOption} as
\[Y_t \doteq \frac{S_t^1}{S_t^2}\ , \]
define an exponential martingale
\[Z_t = \exp\left(-\frac 12\sigma_2^2t+\sigma_2W_t^2\right)\ ,\]
and notice that
\[C(t,s_1,s_2) = s_2\mathbb E^Q[Z_T(Y_T-1)^+|S_t^1 = s_1,S_t^2= s_2]\ .\]
Recall the Girsanov Theorem from Chapter \ref{chapt:brownianMotion} and recognize that $Z_t$ defines a new measure $\widetilde{\mathbb P}^Q$ under which $\widetilde W_t\doteq W_t^2-\sigma_2t$ is Brownian motion, and under which $W_t^1$ remains a Brownian motion independent of $\widetilde W_t$. Applying the It\^o lemma to $Y_t$ we get
\[dY_t = Y_t(\sigma_2^2-\sigma_2\sigma_1\rho)dt+Y_t\sigma_1\left(\sqrt{1-\rho^2}dW_t^1+\rho dW_t^2\right)-Y_t\sigma_2dW_t^2\]

\[=Y_t\sigma_1\left(\sqrt{1-\rho^2}dW_t^1+\rho d\widetilde W_t\right)-Y_t\sigma_2d\widetilde W_t\ ,\]
and we can now write the option price under the new measure to get \textbf{the Margrabe Formula:}
\[C(t,s_1,s_2) = s_2\widetilde{\mathbb E}^Q\left[(Y_T-1)^+\Big|Y_t = s_1/s_2 \right]\]

\[=s_2 \widetilde{\mathbb E}^Q\left[\left(\frac{s_1}{s_2}e^{-\frac{1}{2}\left(\sigma_1^2+\sigma_2^2-2\sigma_2\sigma_1\rho\right)(T-t)+\sigma_1\left(\sqrt{1-\rho^2}(W_T^1-W_t^2)+\rho (\widetilde W_T-\widetilde W_t)\right)-\sigma_2(\widetilde W_T-\widetilde W_t)}-1\right)^+\Big|Y_t = s_1/s_2 \right]\]

\[=s_1N(d_1)-s_2N(d_2)\]
where $N(\cdot)$ is the standard normal CDF function and
\begin{align*}
&d_1 = \frac{\log(s_1/s_2)+\frac 12\sigma^2(T-t)}{\sigma\sqrt{T-t}}\\
&d_2 = d_1-\sigma\sqrt{T-t}\ ,
\end{align*}
where $\sigma^2 = \sigma_1^2+\sigma_2^2-2\sigma_2\sigma_1\rho$. Essentially, the Margrabe formula is the Black-Scholes call option price with $r=0$, $S_0=s_1$, $K=s_2$, and $\sigma^2= \sigma_1^2+\sigma_2^2-2\sigma_2\sigma_1\rho$.
%%%%%%%%%%%%%%%%%%%%%%%%%%%%%%%%%%%%%%%%%%%%%%%%%%%%%%%%%%%%%%%%%%%%%%%%%%%%%%%%%%%%%%%%
%\section{Binary Options}

\section{American Options and Optimality of Early Exercise}
\label{sec:americans}

Broadly speaking, options on index funds (such as S\&P 500) are European, whereas options on individual stocks (such as Apple) are American. American options have the same features as their European counterparts, but with the additional option of \textit{early exercise}. The option of early exercise can be equated with the option to \textit{wait and see}, as the prices of these options is usually thought of the present value of the payoff when exercised at an optimal time. For this reason the American options are consider path-dependent. Methods for pricing American options are quite involved, as there are no explicit solutions. Instead, pricing is done by working backward from the terminal condition on a binomial tree, or by solving a PDE with a free boundary. Pricing methods are not covered here; these notes focus on some fundamental facts about American options.

The first thing to notice is that an American option is worth at least as much as a European. Hence, the prices are such that
\[\hbox{American Option }=\hbox{ European Option }+\hbox{ value of right to early exercise }\]

\[\geq \hbox{ European Option.}\]
In particular, we will that an American put has positive value for the right to early exercise, and so does the American call when the underlying asset pays dividend. \\

%%%%%%%%%%%%%
\noindent\textbf{American Call Option.} Let $C_t^A$ denote the price of an American call option with strike $K$ and maturity $T$. At any time $t\leq T$ we obviously must have $C_t^A>S_t-K$, otherwise there is an arbitrage opportunity by buying the option and exercising immediately for a risk-less profit. Therefore, we can write
\[C_t^A=\max\left(S_t-K,C_t^{cont}\right)\]
where $C_t^{cont}$ is the continuation value of the option at time $t$; at time $t\leq T$ the price is determined by whether or not it is optimal to exercise. Letting $C_t^E$ denote the price of the European call, it is clear that the long party cannot lose anything by having the option to exercise early (provided they do so at a good time) and so
\[C_t^A = \max\left(S_t-K,C_t^{cont}\right)\geq C_t^E\ .\]

%It should also be pointe out that if the American call where to cost less than a European call, then an arbitrage would be to buy the American, sell the European, and invest the difference. The American option will always cover the obligation in the short position on the European if it is not exercised early, and so the net difference in the initial sale is arbitrage. Hence,
%\[C_t^A\geq C_t^E\]
%where $C_t^E$ is the value of the European option.

\begin{proposition}
\label{prop:noDivCall}
For an American call option on a non-dividend paying asset, it is never optimal to exercise early. Hence, American and European calls have the same price and 
\[C_t^A = C_t^{cont}=C_t^E\ .\]
\end{proposition}
\begin{proof} The proof follows from the European Put-Call parity and is simple. Early exercise is clearly not optimal if $S_t<K$. Now suppose that $S_t>K$. Letting $P^E$ denote the price of a European put option, we have
\[C_t^A=\max\left(S_t-K,C_t^{cont}\right)\geq C_t^E = P_t^E+S_t-Ke^{-r(T-t)}>S_t-K=\hbox{ value of early exercise,}\]
which proves the continuation value is worth more than early exercise, and the result follows.
\end{proof}

To illustrate this proposition consider the following scenario: Suppose you're long an ITM American call, and you have a strong feeling that the asset price will go down and leave the option OTM. In this case you should not exercise but should short the asset and use the call to cover your short position. For example, suppose the asset is trading at \$105, you're long a 3 month American call with $K=\$100$, and the risk-free rate is 5\%. Now compare the following two strategies:
\begin{enumerate}
\item Exercise the option and invest the proceeds:
\begin{itemize}
\item Pay $K=\$100$; sell the asset at \$105.
\item Invest proceeds of \$5 in bond.
\item In three months have \$5.06. 
\end{itemize}
\item Short the asset, hold the option, and invest the proceeds:
\begin{itemize}
\item Short the asset at \$105.
\item Invest the proceeds of \$105 in bond.
\item In three months have \$106.32 in the bank; owe no more than \$100 to cover your short position.
\end{itemize}
\end{enumerate}
From these two strategies it is clear that early exercise is sub-optimal, as the latter strategy has \textit{at least} an edge of \$1.32. 
Alternatively, one could simply sell the call, which will also be better than exercising early.\\

The introduction of dividends changes things. In particular, early exercise can be optimal if the present value of the dividends is worth more than the interest earned on the strike and the time value of the option. We see this from the put-call parity
\[C_t^A\geq C_t^E=P_t^E+S_t-Ke^{-r(T-t)}-PV_t(D)\]
where $PV_t(D)$ is the present value of dividends received during the life of the option. For European options, dividends can be expressed in terms of a yield, $\delta\geq 0$, so that $PV_t(D) = S_t(1-e^{-\delta(T-t)})$ and
\[C_t^A\geq C_t^E=P_t^E+S_te^{-\delta(T-t)}-Ke^{-r(T-t)}\ .\]
It may be optimal to exercise early if $C_t^E<S_t-K$, which is possible in the presence of dividends as we can see in Figure \ref{fig:americanCallDiv}. The intuition for this figure is that $P_t^E\rightarrow 0$ as $S_t\rightarrow\infty$, and that $C_t^E\sim S_te^{-\delta(T-t)}-Ke^{-r(T-t)}<S_t-K$ for $S_t$ big enough. This only proves that the option for early exercise has positive value, therefore implying that early exercise may be optimal at some point; it does not mean that one should exercise their American option right now.

\begin{figure}[htbp] %  figure placement: here, top, bottom, or page
   \centering
   \includegraphics[width=5in]{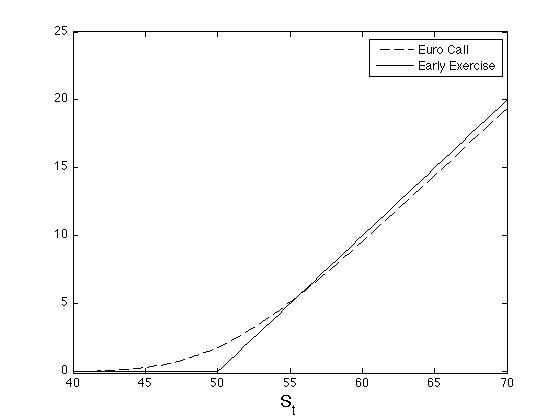} 
   \caption{The Black-Scholes price of a European call option on a dividend paying asset, with $K=50$, $T-t=3/12$, $r=.02$, $\sigma=.2$, and $\delta=.05$. Notice that early exercise exceeds the European price for high-enough price in the underlying.}
   \label{fig:americanCallDiv}
\end{figure}

\begin{proposition}
For an American call option on a dividend-paying asset, it may be optimal to exercise early.
\end{proposition}
\begin{proof}
Let $\delta>0$ be the dividend rate and assume $r>0$. From the put-call parity,
\[C_t^A\geq C_t^E = P_t^E+S_te^{-\delta (T-t)}-Ke^{-r(T-t)} = P_t^E+e^{-r(T-t)}\left(e^{-(\delta-r)(T-t)}S_t-K\right)\ .\]
Now suppose that $\delta$, $r$, $K$ and $S$ are such that
\[C_t^E<S_t-K\ .\]
Clearly, $C_t^A\geq S_t-K$, and so in this case we have strict inequality $C_t^A>C_t^E$, which indicates the early exercise of the American option adds value, and hence early exercise can be optimal if done at the correct time.
\end{proof}

In practice, dividends are paid at discrete times, and early exercise of an American options should only occur immediately prior to the time of dividend payment. The intuition is similar to the non-dividend case: at non-dividend times it is optimal to wait; time value of strike amount in bank and the right to exercise are more valuable.  \\

%%%%%%%%%%%%%
\noindent\textbf{American Put Option.} Regardless of dividends, early exercise of an American put option may be optimal. To understand why, simply consider a portfolio consisting of an American put with strike $K>0$ and a single contract in the underlying. If the underlying hits zero, then it is certainly optimal to exercise. In general, early exercise of a put option will be optimal if the price drops low enough to where the time value of the strike amount in the bank account is worth more than the option to wait on selling the asset plus the present value of any dividends that might be recieved.

\begin{proposition} \label{prop:earlyPut}Early exercise of an American put option may be optimal if $r>0$.
\end{proposition}

\begin{proof}\textbf{By contradiction.} Suppose it is never optimal to exercise an American put. Then 
\[P_t^A=\max\left(K-S_t,P_t^{cont}\right) = P_t^{cont} = P_t^E\]
where $P_t^{cont}$ is the continuation value of the option and $P_t^E$ is the price of a European put. Recall Figure \ref{fig:posTheta} on page \pageref{fig:posTheta} of Section \ref{sec:greeks}, where we see that
\[P_t^E\rightarrow e^{-r(T-t)}K\qquad\hbox{as }S_t\searrow 0.\]
If $r>0$ there is a price $S_t^0>0$ such that $P_t^E<K-S_t$ for all $S_t<S_t^0$, which is a contradiction. Hence, early exercise can be optimal.
\end{proof}

Proposition \ref{prop:earlyPut} is valid regardless of whether or not there are dividends, but the time value of dividends will counterbalance the time value of the cash received from exercise. If $r=0$, then the European put option is always greater or equal to its intrinsic value, which means
\[K-S_t\leq P_t^E\leq P_t^A= \max\left(K-S_t,P_t^{cont}\right)\ ,\]
with strict inequality between $K-S_t$ and $P_t^E$ if $S_t>0$. Hence, if $\mathbb P(S_t>0)=1$, then $\mathbb P(P_t^A=P_t^{cont})=1$ and early exercise is never optimal. \textbf{Joke:} Therefore, the Fed should drop interest rates if they are worried about rampant sell-offs by option holders.\\

To summarize, American options have three components beyond their intrinsic value,  
\begin{itemize}
\item value of option to early exercise (always positive impact on price),
\item time value of money (positive for calls, negative for puts), 
\item dividends (negative for calls, positive for puts).

\end{itemize}
The value of an American option over the price of a European option is essentially the value of early exercise. While the above results are model-free, the price (and the optimal exercise boundary) will be model-dependent.

%%%%%%%%%%%%%%%%%%%%%%
\section{Bermuda Options}
%%%%%%%%%%%%%%%%%%%%%%
A Bermuda call option has two exercise times $T_1$ and $T_2$ with
\[0<T_1<T_2<\infty\ .\]
At time $T_1$ the holder of the option has the right to exercise a call with strike $K_1$, or to continue with the possibility to exercise at time $T_2$ with exercise $K_2$. Hence, the risk-neutral price is
\[C^B(t) = \left\{
\begin{array}{ll}
e^{-r(T_1-t)}\mathbb E^Q\left[\max(S_{T_1}-K_1,C^E(T_1))\Big|\mathcal F_t\right]&\hbox{for }t\leq T_1\ ,\\
C^E(t)&\hbox{for }T_1<t\leq T_2\ ,
\end{array}\right . \]
where $C^E(t)$ is price a European call option with strike $K_2$ and maturity $T_2$. 

Consider the Black-Scholes model,
\[\frac{dS_t}{S_t} = (r-q)dt+\sigma dW_t^Q\ ,\]
where $q$ is the dividend rate. Recalling the Black-Scholes price for a call on a dividend-paying asset (see Section \ref{sec:dividends} of Chapter \ref{chapt:blackScholes}), the time-$T_1$ payoff for the Bermuda call option is
\[\psi(s) = \max\left(s-K_1,se^{-q(T_2-T_1)}N(d_1)-Ke^{-r(T-t)}N(d_2)\right)\ ,\]
where 
\begin{align*}
&d_1  = \frac{\log(s/K_2)+(r-q+.5\sigma^2)(T_2-T_1)}{\sigma\sqrt{T_2-T_1}} \\
&d_2 = d_1-\sigma\sqrt{T_2-T_1}\ .
\end{align*}
It follows that for $t<T_1$ that the Bermuda option price $C^B(t,s)$ is given by the PDE
\begin{align*}
\left(\frac{\partial}{\partial t}+\frac{\sigma^2s^2}{2}\frac{\partial^2}{\partial s^2}+(r-q)s\frac{\partial}{\partial s}-r\right)C^B(t,s) &=0\\
C^B(T_1,s)&=\psi(s)\ .
\end{align*}
When pricing the Bermuda option is intuitive to think about the appropriate early-exercise rule. We learned from Section \ref{sec:americans} that it is never optimal to exercise at time $T_1$ if $q=0$. However if $q>0$ then we need to think of a rule for when to exercise. The rule turns out to be simple: find the optimal value $s^e$ such that at time $t=T_1$,
\begin{align*}
&\hbox{exercise if }S_{T_1}\geq s^e\ ,\\
&\hbox{do not if } S_{T_1}<s^e\ .
\end{align*}
In the next subsection we consider the American put, where there will be a time-dependent exercise rule.

%%%%%%%%%%%%%%%%%%%%%%
\section{Free-Boundary Problem for the American Put}
%%%%%%%%%%%%%%%%%%%%%%
Consider the Black-Scholes model,
\[\frac{dS_t}{S_t} = rdt+\sigma dW_t^Q\ ,\]
with $r>0$. In this case it may be optimal exercise early an American put. The price $P^A$ of the put is the solution to an optimization problem,
\begin{align*}
P^A(t,s) &= \sup_{\tau\geq t} \mathbb E^Q\left[e^{-r(T\wedge\tau-t)}(K-S_{\tau\wedge T})^+\Big|S_t=s\right]\ ,
\end{align*}
where the $\tau$ is the optional exercise time chosen by the holder, and $T\wedge\tau=\min(T,\tau)$. Define $\mathcal L$ to be the PDE operator,
\[\mathcal L = \frac{\partial}{\partial t}+\frac{\sigma^2s^2}{2}\frac{\partial^2}{\partial s^2} +rs\frac{\partial}{\partial s}-r\ .\]
If at time $t$ it is optimal to continue then $\mathcal LP^A(t,S_t)=0$; if it is optimal to exercise then $\mathcal LP^A(t,S_t) = \mathcal L(K-S_t) < 0$ and $\frac{\partial}{\partial s}P^A(t,S_t) = -1$. This is equivalent to finding the pair $(P^A,s^e)$ such that
\begin{align*}
\mathcal LP^A(t,s)&=0~~\hbox{for $t<T$ and $s>s^e(t)$}\ ,\\
P^A(T,s)&=(K-s)^+\\
\frac{\partial}{\partial s}P^A(t,s^e(t))&=-1~~\hbox{for $t\leq T$}\ .
\end{align*}
This is the so-called \textit{free boundary problem}, which refers to the fact that we need to solve for the boundary $s^e(t)$ \textit{and} the PDE's solution $P^A(t,s)$. A straight-forward way to solve this PDE and to see the boundary is with a backward recursion on a binomial tree.

\chapter{Implied Volatility \& Local-Volatility Fits}
\label{sec:impVol}
A major point of discuss with Black-Scholes pricing is the role of volatility. The questions comes about when the Black-Scholes formula is inverted on the market's option prices, which produces an interesting phenomenon known as the \textit{implied volatility smile, smirk, or skew.} The implied volatility suggests that asset prices are more complex than geometric Brownian motion, and the Black-Scholes' parameter $\sigma$ needs to be dynamic. Local volatility models and stochastic volatility models are two well-known ways to address this issue; jump-diffusion models and time-change jump models are also a popular choice of model.
%%%%%%%%%%%%%%%%%%%%%%
\section{Implied Volatility}
Implied volatility is the parameter estimate obtained by inverting the Black-Scholes model on market data. The interesting part is that European call and put options will return a different parameter for different strikes and maturity. Let $C^{BS}(t,s,k,T,\sigma)$ be the Black-Scholes price of a call option with strike $K$ and maturity $T$, and let $C_t^{data}(K,T)$ be the price in the market. 
\begin{definition}
Implied Volatility is the quantity $\hat\sigma(K,T)>0$ such that 
\[C_t^{data}(K,T) = C^{BS}\left(t,S_t,K,T,\hat\sigma(K,T)\right) \ .\]
\end{definition}

Traders often quote option prices in implied volatility, probably because it is a `unit-less' measure of risk. Typically, the implied smile/skew is plotted as a function of the option's strike as shown in Figure \ref{fig:impVolSkew}. When plotting implied volatility, the term `Moneyness' refers to how far the option is in-the-money (ITM) or out-of-the-money (OTM). Moneyness can also be interpreted as a unit-less measure, for instance the log-moneyness,
\[\hbox{log-moneyess} = \log(Ke^{-r(T-t)}/S_t)\ .\]
Implied volatility is skewed to the left for many assets, but some assets (such as forex contracts) have a right skew (see Figure \ref{fig:forexSkew}). The interpretation is the followoing: implied volatility smile/skew suggests that Black-Scholes is an oversimplification, and the left skew of equity options (such as the SPX ETF options) suggests there is `crash-o-phobia' for long positions in the underlying, and long put options are an insurance item for which there is a premium in the market.\\

\begin{figure}[htbp] %  figure placement: here, top, bottom, or page
   \centering
   \includegraphics[width=5in]{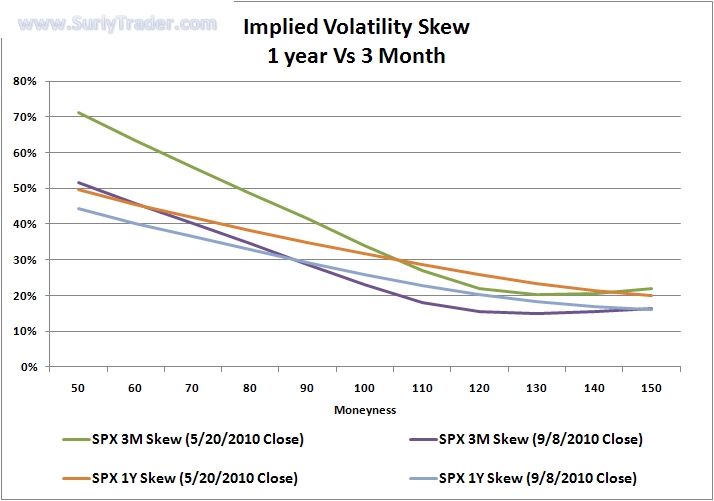} 
   \caption{Implied Volatility Skew of SPX Index Options. On May 20th, 2010, the S\&P 500 closed at 1072, so moneyness refers to something close to the intrinsic value.}
   \label{fig:impVolSkew}
\end{figure}

\begin{figure}[htbp] %  figure placement: here, top, bottom, or page
   \centering
   \includegraphics[width=5in]{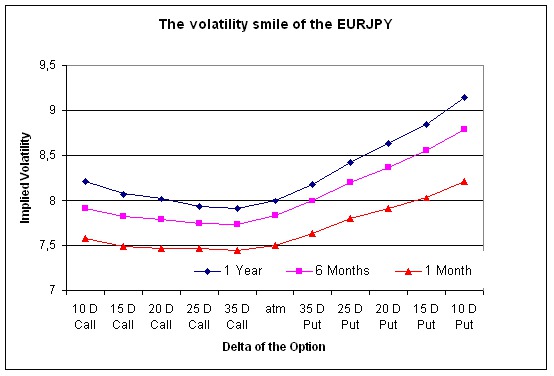} 
   \caption{The implied volatility smile for options on the Euro/Yen currency exchange. Notice the skew is to the right.}
   \label{fig:forexSkew}
\end{figure}

\newpage
\noindent \textbf{Bounds on Implied Volatility.} There are some bounds on implied volatility that can obtained from the Black-Scholes formula.
\begin{proposition}\label{prop:growthBounds}
Changes in implied volatility have bounding inequalities
\[-\frac{e^{-r(T-t)}N(-d_2)}{S_t\sqrt{T-t}N'(d_1)}\leq \frac{\partial}{\partial K}\hat\sigma(K,T)\leq \frac{e^{-r(T-t)}N(d_2)}{S_t\sqrt{T-t}N'(d_1)}\]
where $d_1$ and $d_2$ are given by the Black-Scholes formula and are functions of $\hat\sigma$.
\end{proposition}

\noindent The derivation of Proposition \ref{prop:growthBounds} is as follows:\\

\noindent For the call option,
\[\frac{\partial}{\partial K}C_t^{data}(K)=C_k^{BS_1}(t,S_t,K,T,\hat\sigma(K,T))\]
\[+C_\sigma^{BS}(t,S_t,K,T,\hat\sigma(K,T))\frac{\partial}{\partial K}\hat\sigma(K,T)\leq 0\ ,\]
hence
\[\frac{\partial}{\partial K}\hat\sigma(K,T)\leq -\frac{C_k^{BS_1}(t,S_t,K,T,\hat\sigma(K,T))}{C_\sigma^{BS}(t,S_t,K,T,\hat\sigma(K,T))}\]

\begin{equation}
\label{eq:UB}
=\frac{e^{-r(T-t)}N(d_2)}{S_t\sqrt{T-t}N'(d_1)}\ ,
\end{equation}
where it is straight forward to verify the $K$-derivative to be
\[C_k^{BS_1}=\frac{\partial}{\partial K}\left(S_tN(d_1)-Ke^{-r(T-t)}N(d_2)\right) = -N(d_2)\ .\] 
The ratio in \eqref{eq:UB} has the Greek letter vega in the denominator, which causes it to blow-up when $K$ is far-from-the-money, hence the estimate will be more useful for neat-the-money options (see Figure \ref{fig:growthBounds}). A similar inequality is derived from the put option,
\[-\frac{e^{-r(T-t)}N(-d_2)}{S_t\sqrt{T-t}N'(d_1)}\leq \frac{\partial}{\partial K}\hat\sigma(K,T)\ .\]
This growth bound is also shown in Figure \ref{fig:growthBounds}.\\
\bigskip

\begin{figure}[htbp] %  figure placement: here, top, bottom, or page
   \centering
   \includegraphics[width=5in]{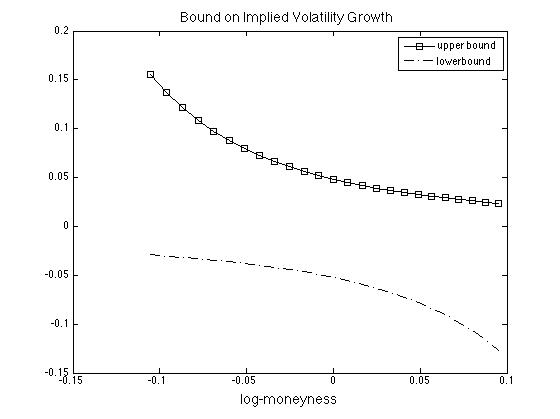} 
   \caption{The growth bound on an option with underlying price $S_t = 50$, time to maturity $T-t=3/13$, interest rate $r=0$, and implied volatility $\hat\sigma = .2$. For an ATM option, the bounds of proposition \ref{prop:growthBounds} show that derivative of implied volatility with respect to strike to be bounded by about 5\%.}
   \label{fig:growthBounds}
\end{figure}
Another important bound on implied volatility is the moment formula for options with extreme strike. Let $x(K) = \log(Ke^{-r(T-t)}/S_t)$. There exists $x^*>0$ such that
\begin{equation}
\label{eq:impVolUB}\hat\sigma(K,T)<\sqrt{\frac{2}{T-t}|x(K)|}
\end{equation}
for all $K$ s.t. $|x(K)|>x^*$. The power of this estimate is that it provides a model-free bound on implied volatility. In other words, it extrapolates the implied volatility smile without assuming any particular structure of the probability distribution. The upper bound in \eqref{eq:impVolUB} can be sharpened to the point where there is uncovered a relationship between implied volatility and the number of finite moments of the underlying $S_T$.
\begin{proposition}\textbf{(The Moment Formula for Large Strike).} 
Let 
\[\bar p = \sup\{p:\mathbb E^QS_T^{1+p}<\infty\}\ ,\]
 and let 
 \[\beta_R = \lim\sup_{K\rightarrow\infty}\frac{\hat\sigma^2(K,T)}{|x(K)|/T}\ .\]
 Then $\beta_R\in[0,2]$ and 
 \[\bar p =\frac{1}{2\beta_R}+\frac{\beta_R}{8}-\frac 12\ ,\]
 where $1/0\doteq \infty $. Equivalently,
 \[\beta_R = 2-4(\sqrt{\bar p^2+\bar p}-\bar p)\]
 where the right-hand side is taken to be zero when $\bar p = \infty$.

\end{proposition}
There is also a moment formula for extremely small strikes and $\bar q = \sup\{q:\mathbb E^QS_T^{-q}<\infty\}$. Both moment formulas are powerful because they are model-free, and because they say that tail behavior of the underlying can be characterized by implied volatility's growth rate against log-moneyness. Indeed, an asset with heavy tails will have higher implied volatility, and so going long or short an option will have more risk, and hence these options will cost more. The important details of the moment formula(s) are in the original paper of Roger Lee \cite{rogerLeeMoment}.

%%%%%%%%%%%%%%%%%%%%%%
\section{Local Volatility Function}
%%%%%%%%%%%%%%%%%%%%%%%%

The idea the volatility is non-constant was first implemented in local volatility models. A local volatility model is a simple extension of Black-Scholes
\[\frac{dS_t}{S_t}=\mu dt+\sigma(t,S_t)dW_t\]
where the function $\sigma(t,s)$ is chosen to fit the implied volatility smile.
Dependence of $\sigma$ on $S_t$ does not effect completeness of market, so there is a pricing PDE that is try similar to the Black-Scholes and is derived in much the same way:
\begin{equation}
\label{eq:localVolPDE}
\left(\frac{\partial}{\partial t}+\frac{S_t^2\sigma^2(t,S_t)}{2}\frac{\partial^2}{\partial s^2} +rS_t\frac{\partial}{\partial s}-r\right)C(t,S_t) = 0 
\end{equation}
with terminal condition $C\big|_{t=T} = (s-K)^+$. Hence, the Feynmann-Kac formula applies and we have
\[C(t,s) = e^{-r(T-t)}\mathbb E^Q[(S_T-K)^+|S_t=s]\]
where $\mathbb E^Q$ is a risk-neutral measure under which $\frac{dS_t}{S_t} = rdt+\sigma(t,S_t)dW_t^Q$. Now, observe the following facts:
\begin{enumerate}
\item The first derivative with respect to $K$ is written with a CDF,
\begin{align*}
\frac{\partial}{\partial K}C(t,s) &=e^{-r(T-t)} \frac{\partial}{\partial K}\mathbb E^Q[(S_T-K)^+|S_t= s]\\
&=e^{-r(T-t)} \mathbb E^Q\left[\frac{\partial}{\partial K}(S_T-K)^+\Big|S_t= s\right]\\
&=-e^{-r(T-t)} \mathbb E^Q\left[\indicator{S_T\geq K}\Big|S_t= s\right]\\
&=-e^{-r(T-t)} \mathbb Q(S_T>K|S_t=s)\\
&=-e^{-r(T-t)} \left(1 - \mathbb Q(S_T\leq K|S_t=s)\right)\ .
\end{align*}
\item The second derivative with respect to $K$ is a PDF,
\begin{align*}
\frac{\partial^2}{\partial K^2}C(t,s) &=e^{-r(T-t)} \frac{\partial^2}{\partial K^2}\mathbb E^Q[(S_T-K)^+|S_t= s]\\
&=-e^{-r(T-t)} \frac{\partial}{\partial K}\left(1 - \mathbb Q(S_T\leq K|S_t=s)\right)\\
&=e^{-r(T-t)} \frac{\partial}{\partial K} \mathbb Q(S_T\leq K|S_t=s)\ ,
%&=e^{-r(T-t)} \frac{\partial}{\partial K} \mathbb Q(S_{T-t}\leq K|S_0=s)\ ,
\end{align*}
which is commonly referred to as \textbf{the Breeden and Litzenberger formula.}
%where the last line follows from the homogeneity of the process $S_t$.
%\item Then, notice that
%\begin{align*}
%\frac{\partial}{\partial t}\frac{\partial^2}{\partial K^2}C(t,s) &= r\frac{\partial^2}{\partial K^2}C(t,s) -e^{-r(T-t)} \frac{\partial}{\partial K}\frac{\partial}{\partial t} \mathbb Q(S_{T-t}\leq K|S_0=s)\\
%&= r\frac{\partial^2}{\partial K^2}C(t,s) +e^{-r(T-t)} \frac{\partial}{\partial K}\frac{\partial}{\partial T} \mathbb Q(S_{T-t}\leq K|S_0=s)\\
%&=-\frac{\partial}{\partial T}\frac{\partial^2}{\partial K^2}C(t,s) \ ,
%\end{align*}

\item Finally, denote by $C(t,s,k)$ all the list call options with strike $T$ and strikes $k$ ranging from zero to infinity. For a fixed $K>0$, differentiating with respect to $T$ and repeated application of integration-by-parts yields the following 
\begin{align*}
&\frac{\partial}{\partial T}C(t,s,K)\\
=&\frac{\partial}{\partial T}\left(e^{-r(T-t)}\mathbb E^Q[(S_T-K)^+|S_t=s]\right)\\
=&-re^{-r(T-t)}\mathbb E^Q[(S_T-K)^+|S_t=s]+e^{-r(T-t)}\frac{\partial}{\partial T}\mathbb E^Q[(S_T-K)^+|S_t=s]\\
=&-re^{-r(T-t)}\mathbb E^Q\{(S_T-K)^+|S_t=s\}\\
&+e^{-r(T-t)}\frac{\partial}{\partial T}\mathbb E^Q\left[(s-K)^++\int_t^T\left(rS_u\indicator{S_u\geq K}+\frac{\sigma^2(u,S_u)S_u^2}{2}\delta_{K}(S_u)\right)du\Big|S_t=s\right]\\
&+e^{-r(T-t)}\frac{\partial}{\partial T}\underbrace{\mathbb E^Q\left[\int_t^T\sigma(u,S_u)S_u\indicator{S_u\geq K}dW_u^Q\Big|S_t=s\right]}_{=0}\\
=&-re^{-r(T-t)}\mathbb E^Q[(S_T-K)^+|S_t=s]\\
&+e^{-r(T-t)}\mathbb E^Q\left[rS_T\indicator{S_T\geq K}+\frac{\sigma^2(T,S_T)S_T^2}{2}\delta_{K}(S_T)\Big|S_t=s\right]\\
=&rKe^{-r(T-t)}\mathbb E^Q\left[\indicator{S_T\geq K}+\frac{\sigma^2(T,S_T)S_T^2}{2}\delta_{K}(S_T)\Big|S_t=s\right]\\
=&rK\int_0^\infty \indicator{k\geq K}\frac{\partial^2}{\partial k^2}C(t,s,k)dk+\int_0^\infty \frac{\sigma^2(T,k)k^2}{2}\delta_{K}(k)\frac{\partial^2}{\partial k^2}C(t,s,k)dk\\
%=&-r\int_0^\infty \delta_{K}(k)C(t,s,k)dk-r\int_0^\infty k\delta_{K}(k)\frac{\partial}{\partial k}C(t,s,k)dk+r\int_0^\infty \delta_{K}(k)C(t,s,k)dk\\
%&+\int_0^\infty \frac{\sigma^2(T,k)k^2}{2}\delta_{K}(k)\frac{\partial^2}{\partial k^2}C(t,s,k)dk\\
=&-rK\frac{\partial}{\partial K}C(t,s,K)+ \frac{\sigma^2(T,K)K^2}{2}\frac{\partial^2}{\partial K^2}C(t,s,K)\ .\\
\end{align*}
\end{enumerate}

\noindent The above steps 1 through 3 have derived what is known in financial literature as \textbf{Dupire's Equation:}
\begin{equation}
\label{eq:dupire}
\left(\frac{\partial}{\partial T}-\frac{K^2\sigma^2(T,K)}{2}\frac{\partial^2}{\partial K^2}+rK\frac{\partial}{\partial K}\right)C(t,s,K) = 0
\end{equation}
with initial condition $C(t,s,K)|_{T=t} = (s-K)^+$. From \eqref{eq:dupire} we get the Dupire scheme for \textbf{implied local volatility:}
\begin{equation}
\label{eq:dupireFormula}
I^2(T,K) = 2\frac{\left(\frac{\partial}{\partial T}+rK\frac{\partial}{\partial K}\right)C }{K^2\frac{\partial^2}{\partial K^2}C}
\end{equation}
which can be fit to the market's options data for various $T$'s and $K$'s (see \cite{dupire}). This is called `calibration of the implied volatility surface', although methods for calibrating local volatility can be quite complicated with numerical differentiation schemes in $T$ and $K$, regularization constraints, etc (for more on fitting techniques see \cite{achdouPironneauBook}). Figure \ref{fig:impVolSurface} shows a calibration of the \textit{local volatility surface.}\\

\begin{figure}[htbp] %  figure placement: here, top, bottom, or page
   \centering
   \includegraphics[width=5in]{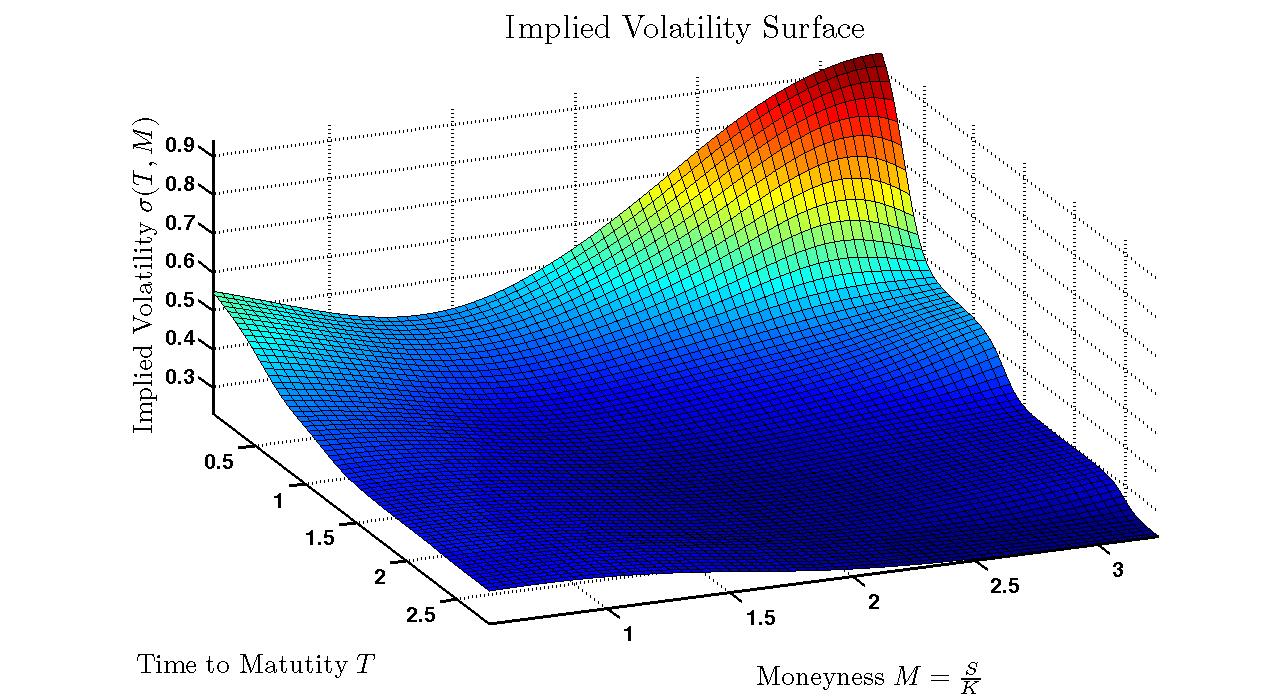} 
   \caption{Implied Volatility Surface}
   \label{fig:impVolSurface}
\end{figure}
\noindent \textbf{Shortcomings of Local Volatility.} There are some shortcomings of local volatility, one of which is the time dependence or `little $t$' effects, which refers to changes in the calibrated surface with time. In other words, the volatility surface will change from day to day, a phenomenon that is inconsistent with $\Delta$-hedging arguments for local volatility.

%%%%%%%%%%%%%%%%%%%%%%
\section{Fitting the Local Volatility Function}
%%%%%%%%%%%%%%%%%%%%%%%%
Define the log moneyness $k=\log(K/S_te^{r(T-t)})$, and the time-weighted implied volatility
\[ \omega(T,k) = \hat\sigma^2\left(T,S_te^{r(T-t)}e^k\right) (T-t)\ .\]
It is straight forward to verify that the call-option price can be written with a dimensional version of Black-Scholes formula, 
\begin{align*}
C(t,s,K) &= e^{-r(T-t)}\mathbb E^Q[(S_T-K)^+|S_t=s] \\
%&=s\mathbb E^Q\{(S_T/(se^{r(T-t)})-e^k)^+|S_t=s\}\\
&=sC^{BS_1}(\omega(T,k),k)\\
&=s(N(d_1)-e^kN(d_2))\ ,
\end{align*}
where $C^{BS_1}(\omega,k)$ is the Black-Scholes call-option price $S_t=1$, time to maturity $T-t=1$, strike $e^k$, interest rate $r=0$, initial asset price $S_t=1$, and volatility $\sqrt\omega$; the formula parameters are $d_1=-k/\sqrt\omega+ \sqrt\omega/2$ and $d_2=d_1-\sqrt{\omega}$. The dimensionless option price $C^{BS_1}(\omega, k)$ takes only log moneyness and time-weight volatility, and then the market's option price is proportional to the dimensionless price. 

Now, there are the following derivatives,
\begin{align}
\nonumber
\frac1s\frac{\partial}{\partial T}C(t,s,K) &= C_k^{BS_1}(\omega(T,k),k)\frac{\partial k}{\partial T}+ C_\omega^{BS_1}(\omega(T,k),k)\frac{\partial \omega(T,k)}{\partial T}\\
\label{eq:dT_C}
&=-rC_k^{BS_1}(\omega(T,k),k)+C_\omega^{BS_1}(\omega(T,k),k)\left(\omega_T(T,k)-r\omega_k(T,k)\right)\\
\nonumber
&\\
\nonumber
\frac1s\frac{\partial}{\partial K}C(t,s,K) &= C_k^{BS_1}(\omega(T,k),k)\frac{\partial k}{\partial K}+C_\omega^{BS_1}(\omega(T,k),k)\frac{\partial \omega(T,k)}{\partial k}\frac{\partial k}{\partial K}\\
\label{eq:dK_C}
&=\frac1KC_k^{BS_1}(\omega(T,k),k)+\frac1KC_\omega^{BS}(\omega(T,k),k) \omega_k(T,k)\\
\nonumber
&\\
\nonumber
\frac1s\frac{\partial^2}{\partial K^2}C(t,s,K) &=\frac{\partial^2}{\partial k^2}C^{BS_1}(\omega(T,k),k)\left(\frac{\partial k}{\partial K}\right)^2+\frac{\partial}{\partial k}C^{BS_1}(\omega(T,k),k)\frac{\partial^2 k}{\partial K^2}\\
\label{eq:dKdK_C}
&=\frac{1}{K^2}\left(\frac{\partial^2}{\partial k^2}C^{BS}(\omega(T,k),k)-C_k^{BS_1}(\omega(T,k),k) \right)\ .
\end{align}
Plugging \eqref{eq:dT_C}, \eqref{eq:dK_C}, and \eqref{eq:dKdK_C} into the implied local volatility function of  
\eqref{eq:dupireFormula} yields
\begin{align}
\label{eq:dupireFormula_impVol_pre}
&\sigma^2\left(t,S_te^{r(T-t)}e^k\right) =2 \frac{C_\omega^{BS_1}(\omega(T,k),k) \omega_T(T,k)}{\frac{\partial^2}{\partial k^2}C^{BS_1}(\omega(T,k),k)-\frac{\partial}{\partial k}C^{BS_1}(\omega(T,k),k)}\ .
\end{align}
The denominator in \eqref{eq:dupireFormula_impVol_pre} involves the following three derivatives:
\begin{align}
\label{eq:Cww}
C_{\omega\omega}^{BS_1}(\omega,k)&=\left(-\frac18-\frac{1}{2\omega}+\frac{k^2}{2\omega^2}\right)C_\omega^{BS_1}(\omega,k)\\
\label{eq:Cwk}
C_{\omega k}^{BS_1}(\omega,k)&=\left(\frac12-\frac k\omega\right)C_\omega^{BS_1}(\omega,k)\\
\label{eq:Ckk-Ck}
C_{kk}^{BS_1}(\omega,k)-C_k^{BS_1}(\omega,k)&=2C_\omega^{BS_1}(\omega,k)\ .
\end{align}
Using \eqref{eq:Cww}, \eqref{eq:Cwk} and \eqref{eq:Ckk-Ck}, the denominator in \eqref{eq:dupireFormula_impVol_pre} becomes
\begin{align*}
&\frac{\partial^2}{\partial k^2}C^{BS_1}(\omega(T,k),k)-\frac{\partial}{\partial k}C^{BS_1}(\omega(T,k),k) \\
&=C_{\omega\omega}^{BS_1}(\omega(T,k),k)(\omega_k(T,k))^2+C_\omega^{BS_1}(\omega(T,k),k)w_{kk}(T,k)+2C_{\omega k}^{BS_1}(\omega(T,k),k)\omega_k(T,k)\\
&\hspace{1cm}+C_{kk}^{BS_1}(\omega(T,k),k)-C_\omega^{BS_1}(\omega(T,k),k)\omega_k(T,k)-C_k^{BS_1}(\omega(T,k),k)\\
&=\left\{\left(-\frac18-\frac{1}{2\omega}+\frac{k^2}{2\omega^2}\right)(\omega_k(T,k))^2+w_{kk}(T,k)-\frac{2k}{\omega}\omega_k(T,k)+2\right\}C_\omega^{BS_1}(\omega(T,k),k)\ ,
\end{align*}
which we then plug into \eqref{eq:dupireFormula_impVol_pre} to get the following implied local volatility formula in terms of the implied volatility surface:

\begin{align}
\nonumber
&\sigma^2\left(T,K\right)\\
\label{eq:dupireFormula_impliedVol}
& = \frac{\omega_T(T,k)}{1-\frac{k}{\omega(T,k)}\omega_k(T,k) +\tfrac14\left(-\tfrac14-\tfrac{1}{\omega(T,k)}+\tfrac{k^2}{\omega^2(T,k)}\right)\left(\omega_k(T,k)\right)^2+\frac12\omega_{kk}(T,k)}\ .
\end{align}

%%%%%%%%%%%%%%%%%%%%%%
\section{Surface Parameterizations}
%%%%%%%%%%%%%%%%%%%%%%%%

Given the formula in \eqref{eq:dupireFormula_impliedVol}, it remains to decide on a differentiable model for $\omega(T,k)$. One approach is to find a parametric function $f^\theta(T,k)$ to fit the implied volatility surface. If the function is once differentiable in $T$ and twice in $k$, then the chain rule can be applied to obtain the derivatives in the \eqref{eq:dupireFormula_impliedVol}.

Let $\theta$ be a vector of parameters fit to the weighted surface,

\[\min_\theta\sum_{T,k}\left|f^\theta(T,k)-\omega(T,k)\right| \ ,\]
where $f^\theta$ is the parametric function. Given a fit $\theta$, we need to check that
\begin{itemize}
\item the fitted surface is a free from calendar-spread arbitrage,
\item and each time slice is free of butterfly arbitrage.
\end{itemize}
\begin{proposition}
\label{prop:calendarSpreadArb}
The fitted surface is free from calendar-spread arbitrage if
\[f_T^\theta(T,k) \geq 0\ ;\]
in other words the slices of the weighted surface cannot intersect. 
\end{proposition}
To prevent butterfly arbitrage, effectively we need to ensure that the option price yields a probability density function,
\[\frac{d}{dK}\mathbb P^Q(S_T\leq K|\mathcal F_t) = e^{r(T-t)}\frac{\partial^2}{\partial K^2}C(t,T,K) \ .\]
Using the parameterized surface, this density is 
\begin{align*}
e^{r(T-t)}\frac{\partial^2}{\partial k^2}C(t,T,K)&= se^{r(T-t)}\frac{\partial^2}{\partial k^2}C^{BS_1}(f^\theta(T,k),k)\\
&=\frac{g^\theta(k)}{\sqrt{2\pi f^\theta(T,k)}}e^{-\tfrac12d_2^2(k;\theta)}\ ,
\end{align*}
where $d_2(k;\theta)=-k/\sqrt{f^\theta(T,k)}- \sqrt{f^\theta(T,k)}/2$, and 
\begin{equation}
\label{eq:f_density}
g^\theta(k) = \left(1-\frac{kf_k^\theta(T,k)}{2f^\theta(T,k)}\right)^2-\frac{f_k^\theta(T,k)^2}{4}\left(\frac{1}{f^\theta(T,k)}+\frac14\right)+\frac{f_{kk}^\theta(T,k)}{2}\ .
\end{equation}
\begin{proposition}
\label{prop:butterflyArb}
The fitted surface is free from butterfly arbitrage if and only if $g^\theta(k)\geq 0$ and $\lim_{k\rightarrow+\infty}d_1(k;\theta)=-\infty$, where $d_1(k;\theta)=-k/\sqrt{f^\theta(T,k)}+ \sqrt{f^\theta(T,k)}/2$.
\end{proposition}

\begin{example}[The SVI Parameterization]
One such model is the stochastic volatility inspired (SVI) parameterization of the implied-volatility surface. The SVI uses the parametric function
\begin{equation}
\label{eq:sviFunc}
f^\theta(k) = a+b\left(\rho(k-m)+\sqrt{(k-m)^2+\xi^2}\right)
\end{equation}
with parameters $\theta=(a,b,\rho,m,\xi)$ fitted across $k$ and $T$. Clearly, the SVI has $\frac{\partial}{\partial T}f^\theta(k) = 0$ to avoid calendar-spread arbitrage (per Proposition \ref{prop:calendarSpreadArb}) and parameterizations can be found to avoid butterfly arbitrage per Proposition \ref{prop:butterflyArb}.
\end{example}

\chapter{Stochastic Volatility}
\label{chapt:stochVol}
Stochastic volatility is another popular way to fit the implied volatility smile. It has the advantage of assuming volatility brings another source of randomness, and this randomness can be hedge if it is also a Brownian motion. Furthermore, calibration of stochastic volatility models does have some robustness to the market's daily changes. However, whereas local volatility models are generally calibrated to fit the implied volatility surface (i.e. to fit options of all maturities), stochastic volatility models have trouble fitting more than one maturity at a time.

Consider the following stochastic volatility model,
\begin{eqnarray}
\label{eq:SVM}
dS_t&=&\mu S_tdt+\sigma(X_t)S_tdW_t\\
\label{eq:dX}
dX_t &=&\alpha(X_t)dt+\beta(X_t)dB_t
\end{eqnarray}
where $W_t$ and $B_t$ are Brownian motions with correlation $\rho\in(-1,1)$ such that $\mathbb E[dW_tdB_t]=\rho dt$, $\sigma(x)>0$ is the volatility function, and $X_t$ is the volatility process and is usually mean-reverting. An example of a mean-revering process is the square-root process,
\[dX_t = \kappa(\bar X-X_t)dt+\gamma\sqrt{X_t}dB_t\ .\]
The square-root process is used in the Heston model (along with $\sigma(x) =\sqrt{x}$), and usually relies on what is known as the Feller condition ($\gamma^2\leq 2\bar X\kappa$) so that $X_t$ never touches zero.

\section{Hedging and Pricing}
Assuming that $\alpha$ and $\beta$ are `well-behaved' functions, we define the differential operator
\[\mathcal L \doteq\frac{\partial}{\partial t}+\frac{\sigma^2(x)s^2}{2}\frac{\partial^2}{\partial s^2}+\mu s \frac{\partial}{\partial s}+\frac{\beta^2(x)}{2}\frac{\partial^2}{\partial x^2}+\alpha(x)\frac{\partial}{\partial x}+\rho\beta(x)\sigma(x)s\frac{\partial^2}{\partial x\partial s}\ ; \] 
the interpretation of $\mathcal L$ is the following: $\frac{\sigma^2(X_t)S_t^2}{2}\frac{\partial^2}{\partial s^2}+\mu S_t \frac{\partial}{\partial s}$ are the generator of $S_t$; $\frac{\beta^2(X_t)}{2}\frac{\partial^2}{\partial x^2}+\alpha(X_t)\frac{\partial}{\partial x}$ are the generator of $X_t$; $\rho\beta(X_t)\sigma(X_t)S_t\frac{\partial^2}{\partial x\partial s}$ is the cross-term; and the remaining are the separate stochastic terms. We use when applying bi-variate It\^o Lemma for the process in \eqref{eq:SVM} and \eqref{eq:dX}:\\
For any twice differentiable function $f(t,s,x)$, the It\^o differential is 
\begin{align}
\label{eq:2Dito} 
df(t,S_t,X_t) &= \mathcal Lf(t,S_t,X_t)dt+\sigma(X_t)S_t\frac{\partial}{\partial s}f(t,S_t,X_t)dW_t+\beta(X_t)\frac{\partial}{\partial x}f(t,S_t,X_t)dB_t\ .
\end{align}

The price $C(t,s,x)$ is the price of claim $\psi(s,x)$ paid at time $T$ given $S_t=s$ and $X_t=x$. We hedge with a self-financing portfolio $V_t$ consisting of risky-asset, the risk-free bank account, and a 2nd derivative $C'$ that is settled at time $T'>T$. This method for hedging stochastic volatility is called the Hull-White method. 

The self-financing condition is
\[dV_t = a_tdS_t+b_tdC'(t,S_t,X_t)+r(V_t - a_tS_t-b_tC'(t,S_t,X_t))dt\ .\]
Now, taking a long position in $V_t$ and a short position in $C(t,S_t,X_t)$, we apply the bivariate It\^o lemma of \eqref{eq:2Dito} to get the dynamics of this new portfolio:

\[d(V_t-C(t,S_t,X_t))\]

%\[ = -\Bigg(\frac{\partial}{\partial t}+\mu S_t\frac{\partial}{\partial s}+\frac{\sigma^2(X_t)S_t^2}{2}\frac{\partial^2}{\partial s^2}+\alpha(X_t)\frac{\partial}{\partial x}+\frac{\beta(X_t)}{2}\frac{\partial^2}{\partial x^2}+\rho \beta(X_t)S_t\sigma(X_t)\frac{\partial^2}{\partial x\partial s}\Bigg)C(t,S_t,X_t)dt\]
%
%
%\[-\sigma(X_t)S_t\frac{\partial}{\partial s}C(t,S_t,X_t)dW_t-\beta(X_t)\frac{\partial}{\partial x}C(t,S_t,X_t)dB_t\]
%
%\[+ a_tdS_t+b_tdC'(t,S_t,X_t)+r(V_t - a_tS_t-b_tC'(t,S_t,X_t))dt\]
\[=a_t(\mu S_tdt+\sigma(X_t)S_tdW_t)\]

\[+ b_t\mathcal LC'(t,S_t,X_t)dt+b_t\sigma(X_t)S_t\frac{\partial}{\partial s}C'(t,S_t,X_t)dW_t+b_t\beta(X_t)\frac{\partial}{\partial x}C'(t,S_t,X_t)dB_t\]

\[+r(V_t - a_tS_t-b_tC'(t,S_t,X_t))dt\]
\[-\mathcal LC(t,S_t,X_t)dt-\sigma(X_t)S_t\frac{\partial}{\partial s}C(t,S_t,X_t)dW_t-\beta(X_t)\frac{\partial}{\partial x}C(t,S_t,X_t)dB_t\ .\]
We now choose $a_t$ and $b_t$ so that the $dW_t$ and the $dB_t$ terms vanish, which means that

\begin{align*}
&\sigma(X_t)S_t\frac{\partial}{\partial s}C(t,S_t,X_t) = a_t\sigma(X_t)S_t+ b_t\sigma(X_t)S_t\frac{\partial}{\partial s}C'(t,S_t,X_t)\\
&\\
&\beta(X_t)\frac{\partial}{\partial x}C(t,S_t,X_t)=b_t\beta(X_t)\frac{\partial}{\partial x}C'(t,S_t,X_t)
\end{align*}
or 

\begin{align}
\label{eq:a_t}
&b_t = \frac{\frac{\partial}{\partial x}C(t,S_t,X_t)}{\frac{\partial}{\partial x}C'(t,S_t,X_t)}\\
\label{eq:b_t}
&a_t = \frac{\partial}{\partial s}C(t,S_t,X_t)-b_t\frac{\partial}{\partial s}C'(t,S_t,X_t)\ .
\end{align}
Plugging these solutions into the differential, the $\alpha(X_t)\frac{\partial}{\partial x}$ terms cancel, the $\mu S_t\frac{\partial}{\partial s}$ terms cancel, and by the same arbitrage argument as Black-Scholes, it follows that $V_t-C(t,S_t,X_t)$ must grow at the risk-free rate:
\begin{align*}
&d(V_t-C_t(t,S_t,X_t))\\
&= a_t\mu S_t dt+b_t\mathcal LC'(t,S_t,X_t)dt+r(V_t - a_tS_t-b_tC'(t,S_t,X_t))dt  -\mathcal LC(t,S_t,X_t)dt\\
&=r(V_t-C(t,S_t,X_t))dt\ .
\end{align*}
Inserting $a_t$ and $b_t$ from equations \eqref{eq:a_t} and \eqref{eq:b_t}, and after some rearranging of terms (and dropping the $dt$'s), we get 
\[\left(\mathcal L-(\mu-r)S_t\frac{\partial}{\partial s}-r\right)C(t,S_t,X_t)= b_t\left(\mathcal L-(\mu-r)S_t\frac{\partial}{\partial s}-r\right)C'(t,S_t,X_t)\]
and dividing both sides by $\frac{\partial}{\partial x}C(t,S_t,X_t)$, we get
\[\frac{\left(\mathcal L-(\mu-r)S_t\frac{\partial}{\partial s}-r\right)C(t,S_t,X_t)}{\frac{\partial}{\partial x}C(t,S_t,X_t)}= \frac{\left(\mathcal L-(\mu-r)S_t\frac{\partial}{\partial s}-r\right)C'(t,S_t,X_t)}{\frac{\partial}{\partial x}C'(t,S_t,X_t)}\doteq R(t,s,x)\ .\]
The left-hand side of this equation does not depend on specifics of $C'$ (e.g. maturity and strikes), and the right-hand side does not depend on specifics of $C$. Therefore, there is a function $R(t,s,x)$ that does not depend on parameters such as strike price and maturity, and it equates the two sides of the expression. We let $R$ take the form
\[R(t,s,x) = -\alpha(x)+\beta(x)\Lambda(t,s,x)\]
where 
\[\Lambda(t,s,x) = \rho\frac{\mu-r}{\sigma(x)}+g(t,s,x)\sqrt{1-\rho^2}\ ,\]
where $g$ is an arbitrary function. Hence, we arrive at the PDE for the price $C(t,s,x)$ of a European Derivative with stochastic vol:
\begin{proposition}\textbf{(Pricing PDE for Stochastic Volatility).}
\label{prop:stochasticVolPDE}
Given the stochastic volatility model of equations \eqref{eq:SVM} and \eqref{eq:dX}, the price of a European claim with payoff $\psi(S_T)$ satisfies 
\begin{align}
\nonumber
\Bigg(\frac{\partial}{\partial t}+\frac{\sigma^2(x)s^2}{2}\frac{\partial^2}{\partial s^2}+\rho\beta(x)\sigma(x)s\frac{\partial^2}{\partial x\partial s}+\alpha(x)\frac{\partial}{\partial x}+\frac{\beta^2(x)}{2}\frac{\partial^2}{\partial x^2}+rs\frac{\partial}{\partial s}-r\Bigg)C&\\
\label{eq:SVpde}
=\beta(x)\Lambda(t,s,x)\frac{\partial}{\partial x}C&
\end{align}
with $C\Big|_T=\psi(s,x)$. 
\end{proposition}
\noindent The PDE in \eqref{eq:SVpde} can be solved with Feynman-Kac, 
\[C(t,s,x) = e^{-r(T-t)}\mathbb E^Q\{\psi(S_T,X_T)|S_t=s,X_t=x\}\ .\]
which yields a description of an EMM, $\mathbb Q$, under which the stochastic volatility model is
\begin{eqnarray*}
\frac{dS_t}{S_t}&=&r dt+\sigma(X_t)dW_t^Q\\
dX_t&=&\left(\alpha(X_t)-\beta(X_t)\Lambda(t,S_t,X_t)\right)dt+\beta(X_t)dB_t^Q
\end{eqnarray*}
where $W_t^Q$ and $B_t^Q$ are $\mathbb Q$-Brownian motions with the same correlation $\rho$ as before. The interpretation of $\Lambda(t,s,x)$ is that it is \textbf{the market price of volatility risk.} 

One can think of a market as complete when there are as many non-redundant assets as there are sources of randomness. In this sense we have `completed' this market with stochastic volatility by including the derivative $C'$ among the traded assets. This has allowed us to replicate the European claim and obtain a unique no-arbitrage price that must be the solution to \eqref{eq:SVpde}.\\

%\begin{remark}A sufficient condition for application of the Feynman-Kac formula is 
%\[\inf_x\sigma(x)>0\hbox{ and }\inf_x\beta(x)>0\ .\]
%This insures that the PDE is \textit{non-degenerate}. Indeed, there are cases in which Feynman-Kac applies to degenerate equations, but applicability of the formula is evaluated on a case by case basis. For instance, Black-Scholes at first glance appears degenerate because it has an $s$ before the second derivative, however the change of variables $v = \log(s)$ converts the Black-Scholes PDE to a heat equation and confirms that Black-Scholes is actually non-degenerate (see the heat equation remark in Section \ref{sec:blackScholes}). In the literature, non-degeneracy is usually assumed to avoid trouble, and non-degenerate models evaluated after the basic theory has been established.\\
%\end{remark}

\begin{example}\textbf{(The Heston Model).} For the Heston model described earlier, the EMM $\mathbb Q$ is described by
\begin{align}
\label{eq:dSheston}
&dS_t=r S_tdt+\sqrt{X_t}dW_t^Q\\
\label{eq:dXheston}
&dX_t=\kappa(\bar X-X_t)dt+\lambda X_tdt+\gamma\sqrt{X_t}dB_t^Q
\end{align}
where $\lambda X_t$ is the assumed form of the volatility price of risk. This model can be re-written,
\[dX_t = \tilde\kappa(\tilde{\bar X}-X_t)dt+\gamma\sqrt{X_t}dB_t^Q\]
where $\tilde\kappa =\kappa-\lambda$ and $\tilde{\bar X} = \frac{\kappa\bar X}{\kappa-\lambda}$. It is important to have the Feller condition 
\[\gamma^2\leq 2\tilde\kappa\tilde{\bar X}\]
otherwise the PDE requires boundaries for the event $X_t = 0$. The process $X_t$ is degenerate, but it is well-known that the Feller condition is the critical assumption for the application of Feynman-Kac. The fit of the Heston model to the implied volatility of market data is shown in Figure \ref{fig:hestonFit}.
\begin{figure}[htbp] %  figure placement: here, top, bottom, or page
   \centering
   \includegraphics[width=6in]{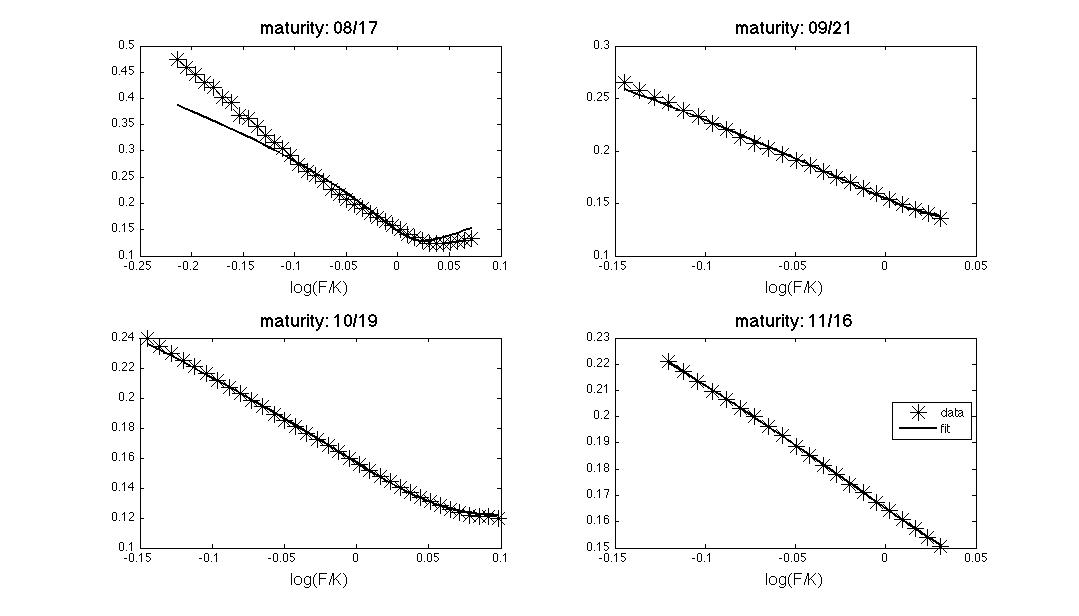} 
   \caption{The fit of the Heston model to implied volatility of market data on July 27 of 2012, with different maturities. Notice that the Heston model doesn't fit very well to the options with shortest time to maturity.}
   \label{fig:hestonFit}
\end{figure}
\end{example}
%%%%%%%%%%%%%%%%%%%%%%
\section{Monte Carlo Methods}
A slow but often reliable method for pricing derivatives is to approximate the risk-neutral expectation with independent samples. Essentially, sample paths of Brownian motion are obtained using a pseudo-random number generated, the realized derivative payoff is computed for each sample, and then the average of these sample payoffs is an approximation of the true average. For example, let $\ell=1,2,\dots,N$, let $S_T^{(\ell)}$ be the index for an independent sample from $S_T$'s probability distribution. If $\mathbb E^QS_T^2<\infty$, then by the law of large numbers we have
\[\frac 1N\sum_{\ell=1}^N(S_T^{(\ell)}-K)^+\stackrel{N\rightarrow\infty}{\longrightarrow}\mathbb E^Q(S_T-K)^+\ .\] 

For stochastic differential equations, the typical method is to generate samples that are approximately from the same distribution as the price. The basic idea is to use an \textbf{Euler-Maruyama scheme,}
\begin{equation}
\label{eq:EMscheme}
\log \tilde S_{t+\Delta t}^{(\ell)} = \log \tilde S_t^{(\ell)}+\left(\mu-\frac 12\sigma^2\right)\Delta t+\sigma(W_{t+\Delta t}^{(\ell)}-W_t^{(\ell)})\ ,
\end{equation}
for some small $\Delta t>0$, with 
\[W_{t+\Delta t}^{(\ell)}-W_t^{(\ell)}\sim iid \mathcal N(0,\Delta t)\ .\]
This scheme is basically the discrete backward Riemann sum from which the It\^o integral was derived (see Chapter \ref{chapt:brownianMotion}). In the limit, the Monte-Carlo average will be with little-o of the true expectation,
\[\lim_{N\rightarrow\infty} \frac 1N\sum_{\ell=1}^N(\tilde S_T^{(\ell)}-K)^+ = \mathbb E^Q(S_T-K)^++o(\Delta t)\ .\]
The scheme in \eqref{eq:EMscheme} can easily be adapted for a local volatility model, and can also be adapted to stochastic volatility models, but the latter will also need a scheme for generating the volatility process $X_t$. 
\begin{example} \textbf{(The exponential OU model).} Consider the stochastic volatility model
\begin{align*}
&\frac{dS_t}{S_t} = rdt+e^{X_t}dW_t^Q\\
&dX_t = \kappa(\bar X-X_t)dt+\gamma dB_t^Q
\end{align*}
with correlation $\rho$ between $W_t^Q$ and $B_t^Q$. The Euler-Maruyama scheme is
\begin{align*}
&\log \tilde S_{t+\Delta t} = \log \tilde S_t+\left(r-\frac 12e^{2\tilde X_t}\right)\Delta t+e^{\tilde X_t}\left(\rho (B_{t+\Delta t}-B_t)+\sqrt{1-\rho^2}(W_{t+\Delta t}-W_t)\right)\\
&\tilde X_{t+\Delta t} = \tilde X_t+\kappa\left(\bar X-\tilde X_t \right)\Delta t+\gamma(B_{t+\Delta t}-B_t)
\end{align*}
where $B_{t+\Delta t}^{(\ell)}-B_t^{(\ell)}$ and $W_{t+\Delta t}^{(\ell)}-W_t^{(\ell)}$ are increments of independent Brownian motions.
\end{example}

Sometimes, there more clever methods of sampling need to be employed. For instance, the Heston model:
\begin{example}\textbf{(Sampling the Heston Model).} Recall the Heston model of \eqref{eq:dSheston} and \eqref{eq:dXheston}. A scheme similar to the one used in the previous example will not work for the Heston model because Euler-Maruyama,
\[\tilde X_{t+\Delta t}=\tilde X_t+\kappa(\bar X-\tilde X_t)\Delta t+\gamma\sqrt{\tilde X_t}\left(B_{t+\Delta t}-B_t\right)\]
will allow the volatility process to become negative. An alternative is to consider the Stratonovich-type differential,\footnote{This refers to the Stratonovich-type integral that is the limit of forward Riemann sums, $\int f_s\circ dB_s = \lim_N\sum_{n=1}^Nf_{t_n}(B_{t_n}-B_{t_{n-1}})$. For the square-root process, the It\^o lemma on $\sqrt{X_t}$ is used to show that $\gamma\int \sqrt{X_s}dB_s = \gamma\lim_n\sum_n\sqrt{X_{t_{n-1}}}\Delta B_{t_{n-1}} = -\gamma\lim_n\sum_n\left(\sqrt{X_{t_n}}-\sqrt{X_{t_{n-1}}}\right)(B_{t_n}-B_{t_{n-1}}) +\gamma\lim_n\sum_n\sqrt{X_{t_n}}(B_{t_n}-B_{t_{n-1}})=-\frac{\gamma^2}{2}\int ds+\gamma\int \sqrt{X_s}\circ dB_s$.}
\[dX_t = \left(\kappa(\bar X-X_t)-\frac{\gamma^2}{2}\right)dt+\gamma\sqrt{X_t}\circ dB_t\]
and then use an implicit Euler-Maruyama scheme
\[\tilde X_{t+\Delta t}-\tilde X_t =\left( \kappa(\bar X-\tilde X_{t+\Delta t})-\frac{\gamma^2}{2}\right)\Delta t+\gamma \sqrt{\tilde X_{t+\Delta t}} \left(B_{t+\Delta t}-B_t\right)\ .\] 
which can be solved for $\sqrt{\tilde X_{t+\Delta t}}$ using the quadratic equation. Notice that a condition for the discriminant to be real is the Feller condition that $\gamma^2\leq 2\kappa\bar X$. 
\end{example}

In general, Monte Carlo methods are good because they can be used to price pretty much any type of derivative instrument, so long as the EMM is given. For instance, it is very easy to price an Asian option with Monte Carlo. Another advantages of Monte Carlo is that there is no need for an explicit solution to any kind of PDE; the methodology for pricing is to simply simulate and average. However, as mentioned above, this can be very slow, and so there is a need (particularly among practitioners) for faster ways to compute prices.

Another important topic to remark on is \textbf{importance sampling,} which refers to simulations generated under an equivalent measure for which certain rare events are more common. For instance, suppose one wants to use Monte Carlo to price a call option with extreme strike. It will be a waste of time to generate samples from the EMM's model if the overwhelming majority of them will finish out of the money. Instead, one can sample from another model that has more volatility, so that more samples finish in the money. Then, assign importance weights based on ratio of densities and take the average. See \cite{glasserman} for further reading on Monte Carlo methods in finance.
%\begin{example}\textbf{(High Strike Local Volatility).}
%Suppose one has a local volatility model,
%\[\frac{dS_t}{S_t}= rdt+\sigma(t,S_t)dW_t^Q\ ,\]
%and suppose a call option with extremely high strike needs to be priced. For argument sake, say that about 1 in every 1000 samples will finish in the money. An importance sampling scheme will take
%\[\tilde \sigma(t,s)\gg\sigma(t,s)\]
%so that about 1 in 5 samples finishes in the time money. The scheme to generate samples is then
%\[\log(\tilde S_{t+\Delta t}^{(\ell)})=\log(\tilde S_t^{(\ell)})+\left(\mu-\frac 12\tilde \sigma(t,\tilde S_t^{(\ell)})\right)\Delta t+\tilde \sigma(t,\tilde S_t^{(\ell)})(W_{t+\Delta t}^{(\ell)}+W_t^{(\ell)}),\]
%and the importance weights are
%\[\omega^\ell =\frac{1}{c}\exp\left( \right)\]
%
%\end{example}
%%%%%%%%%%%%%%%%%%%%%%
\section{Fourier Transform Methods}
The most acclaimed way to compute a European option price is with the so-called \textit{Fourier transform methods.} In particular, the class of affine models in \cite{dps2000} have tractable pricing formulae because they have explicit expressions for their Fourier transforms. The essential ingredient to these methods is the characteristic function,
\[\phi_T(u)\doteq\mathbb E^Q\exp(iu\log(S_T))\qquad\forall u\in\mathbb C,~i=\sqrt{-1}\ ,\]
where $\mathbb E^Q$ now denotes expectation under an EMM. The function $\phi_T(u)$ is the Fourier transform of the $\log S_T$'s distribution function,
\[\phi_T(u) = \int_{-\infty}^\infty  e^{iu s}f(s)ds\]
where $f(s)$ is the density function for $\log S_T$. The inverse Fourier transform is $f(s) = \tfrac{1}{2\pi}\int_{-\infty}^\infty e^{-iu s}\phi_T(u)du$, which leads to the cumulative distribution function
\[F(s)-F(0)=\int_0^sf(v)dv = \frac{1}{2\pi}\int_{-\infty}^\infty  \frac{1-e^{-ius}}{iu}\phi_T(u)du\ .\]

The crucial piece of information \emph{is} the characteristic function, as any European claim can be priced provided there is also a Fourier transform for the payoff function. Let $p_T$ denote the density function on the risk-neutral distribution of $\log(S_T)$ and let $\widehat\psi$ denote the Fourier transform of the payoff function. It is shown in \cite{lewis} that 
\[\mathbb E^Q\psi(\log(S_T)) = \int_{-\infty}^\infty\psi(s)f(s)ds=\int_{-\infty}^\infty\frac{1}{2\pi} \int_{iz-\infty}^{iz+\infty}e^{-ius}\widehat \psi(u) du f(s)ds\]
\[=\frac{1}{2\pi}\int_{iz-\infty}^{iz+\infty}\widehat \psi(u)\int_{-\infty}^\infty e^{-ius}  p_T(s)ds du=\frac{1}{2\pi}\int_{iz-\infty}^{iz+\infty}\widehat \psi(u)\phi_T(-u) du\]
where the constant $z\in\mathbb R$ in the limit of integration is the appropriate strip of regularity for non-smooth payoff functions (see Table \ref{tab:payoffs}). This usage of characteristic functions is far reaching, as the class of L\'evy jump models are defined in terms of their characteristic function, or equivalently with the L\'evy-Khintchine representation (for more on financial models with jumps see \cite{contTankov}). 

\begin{table}[htb]
\center
\caption{Fourier Transforms of Various Payoffs, $\widehat\psi(u)=\int_{-\infty}^\infty e^{ius}\psi(s)ds$ and $\psi(s)=\frac{1}{2\pi} \int_{iz-\infty}^{iz+\infty}e^{-ius}\widehat \psi(u) du$, $z=\Im(u)$.}
\label{tab:payoffs}
\begin{tabular}{|p{4cm}|c|c|c|}
\hline
asset&$\psi(s)$&$\widehat\psi(u)$&regularity strip\\
\hline
call &$(e^s-K)^+$&$-\frac{K^{iu+1}}{u^2-iu}$&$z>1$\\
put&$(K-e^s)^+$&$-\frac{K^{iu+1}}{u^2-iu}$&$z<0$\\
covered call; cash secured put&$\min(e^s,K)$&$\frac{K^{iu+1}}{u^2-iu}$&$0<z<1$\\
cash or nothing call &$\indicator{e^s\geq K}$&$-\frac{K^{iu}}{iu}$&$z>0$\\
cash or nothing put &$\indicator{e^s\leq K}$&$\frac{K^{iu}}{iu}$&$z<0$\\
asset or nothing call&$e^s\indicator{e^s\geq K}$&$-\frac{K^{iu+1}}{iu+1}$&$z>1$\\
asset or nothing put&$e^s\indicator{e^s\leq K}$&$\frac{K^{iu+1}}{iu+1}$&$z<0$\\
Arrow-Debreu&$\delta(s-\log(K))$&$K^{iu}$&$z\in\mathbb R$\\
\hline
\end{tabular}
\end{table}

%%%%%%%%%%%%%%%%%%%%%%%%%%%%%%%%%%%%%%%%%%%%%%%%%%%%%%%%%%%%
\section{Call-Option Pricing}
An important tool is the formula of Gil-Pelaez:
\begin{proposition}[Gil-Pelaez Inversion Theorem]
The distribution function $F(s)$ is given by
\begin{equation}
\label{eq:gilPelaezInversion}
F(s) = \frac12+\frac{1}{2\pi}\int_0^\infty \frac{e^{ius}\phi_T(-u)-e^{-ius}\phi_T(u)}{iu}du\ ,
\end{equation}
for all $s\in(-\infty,\infty)$.
\begin{proof}
First let's prove an identity,
\[\mbox{sign}(v-s) = \frac2\pi\int_0^\infty \frac{\sin(u(v-s))}{u}du = \left\{\begin{array}{cc}-1&\mbox{if }v<s\\0&\mbox{if }v=s\\1&\mbox{if }v>s\end{array}\right. \ .\]
This integral is shown using residue calculus: 
\begin{align*}
\int_0^\infty \frac{\sin(u)}{u}du &= \frac12\int_{-\infty}^\infty \frac{\sin(u)}{u}du=\frac12 \lim_{\epsilon\rightarrow 0}\lim_{R\rightarrow 0}\mbox{Im}\left(\int_{-R}^{-\epsilon}\frac{e^{iz}}{z}dz+\int_\epsilon^R\frac{e^{iz}}{z}dz\right)\ .
\end{align*}
A closed path can be constructed from $-R$ to $R$, making a semi-circle of radius $\epsilon$ around zero, and then from $R$ back to $-R$ on a semi-circle of radius $R$. Along this path the integral of $e^{iz}/z$ is zero,
\[\int_{-R}^{-\epsilon}\frac{e^{iz}}{z}dz+\int_\epsilon^R\frac{e^{iz}}{z}dz=-\int_{C_\epsilon}\frac{e^{iz}}{z}dz-\int_{C_R}\frac{e^{iz}}{z}dz=0\ ,\]
where $C_\epsilon =\{\epsilon e^{i(\theta-\pi)}|0\leq \theta\leq \pi\}$ and $C_R=\{Re^{i\theta}|0\leq \theta\leq \pi\}$. Now,
\begin{align*}
\left|\int_{C_R}\frac{e^{iz}}{z}dz \right|&=\frac1R\left|\int_0^\pi\frac{e^{iR\cos(\theta)-R\sin(\theta)}}{e^{i\theta}}d\theta\right|\\
&\leq \frac{1}{R}\int_0^\pi\left|\frac{e^{iR\cos(\theta)-R\sin(\theta)}}{e^{i\theta}}\right|d\theta \\
&= \frac{1}{R}\int_0^\pi e^{-R\sin(\theta)}d\theta \\
&<\frac{\pi}{R}\rightarrow 0\qquad\hbox{as }R\rightarrow\infty \ ,
\end{align*}
and 
\[\int_{C_\epsilon}\frac{e^{iz}}{z}dz = \int_\pi^0 \frac{1+O(\epsilon)}{\epsilon e^{i\theta}}i\epsilon e^{i\theta}d\theta = -i\pi +O(\epsilon)\ .\]
Hence,
\[\int_0^\infty \frac{\sin(u)}{u}du = \frac12 \lim_{\epsilon\rightarrow 0}\lim_{R\rightarrow 0}\mbox{Im}\left(\int_{-R}^{-\epsilon}\frac{e^{iz}}{z}dz+\int_\epsilon^R\frac{e^{iz}}{z}dz\right)=\frac\pi2\ .\]
Now notice that
\begin{align*}
\frac1\pi\int_0^\infty \frac{e^{ius}\phi_T(-u)-e^{-ius}\phi_T(u)}{iu}du&=\frac1\pi\int_0^\infty\int_{-\infty}^\infty \frac{e^{-iu(v-s)}-e^{iu(v-s)}}{iu}dF(v)du\\
&=-\frac2\pi\int_0^\infty\int_{-\infty}^\infty \frac{\sin(u(v-s))}{u}dF(v)du\\
&=-\int_{-\infty}^\infty d F(v)\frac2\pi\int_0^\infty \frac{\sin(u(v-s))}{u}du\\
&=\int_{-\infty}^s dF(v)-\int_s^\infty dF(v)\\
&=2F(s)-1\ ,
\end{align*}
which proves \eqref{eq:gilPelaezInversion}.
\end{proof}
\end{proposition}
Notice that \eqref{eq:gilPelaezInversion} can be further simplified to have
\[F(s) = \frac12-\frac{1}{\pi}\int_0^\infty \Real\left[\frac{e^{-ius}\phi_T(u)}{iu}\right]du\ ,\]
which is used in pricing a European call option, as the risk-neutral probability of finishing in the money is
\[\Pi_2=\mathbb Q(S_T\geq K)=\frac 12+\frac{1}{\pi}\int_0^\infty\Real \left[\frac{e^{-iu\log(K)}\phi_T(u)}{iu}\right]du\]
where $\mathbb Q$ denotes probability under the EMM. The `delta' of the option is
\[\Pi_1=\frac{e^{-rT}}{S_0}\mathbb E^QS_T\indicator{S_T\geq K}=\frac{e^{-rT}}{S_0}\int_{\log K}^\infty e^sf(s)ds\ ,\]
but we can define a change of measure,
\[\frac{d\widetilde{\mathbb Q}}{d\mathbb Q} = \frac{e^{s_T}}{\mathbb E^Qe^{s_T}}=\frac{e^{s_T}}{S_0e^{rT}}\ ,\]
where $s_T=\log S_T$. Then,
\[\mathbb E^{\widetilde Q}e^{ius_T} = \frac{\mathbb E^ Qe^{(1+iu)s_T} }{\mathbb E^Qe^{s_T}}= \frac{\phi_T(u-i)}{\phi_T(-i)}\ ,\]
and therefore we have the formula
\[\Pi_1=\widetilde{\mathbb Q}(S_T\geq K)=\frac 12+\frac{1}{\pi}\int_0^\infty \Real \left[\frac{e^{-iu\log(K)}\phi_T(u-i)}{iu\phi_T(-i)}\right]du\ .\]
Assuming no dividends, the value of a European call option with strike $K$, at time $t=0$, is
\begin{equation}
\label{eq:fourierPrice}C= S_0\Pi_1-Ke^{-rT}\Pi_2\ .
\end{equation}
The pricing formula of equation \eqref{eq:fourierPrice} is (in one way or another) interpreted as an inverse Fourier transform. Indeed, the method of Carr and Madan \cite{carrMadan1999} shows that \eqref{eq:fourierPrice} can be well-approximated as a fast Fourier transform, which has the advantage of taking less time to compute than numerical integration, but it is also less accurate. Equation \eqref{eq:fourierPrice} has become an important piece of machinery in quantitative finance because it is (or is considered to be) an explicit solution to the pricing equation. In general the pricing formula of \eqref{eq:fourierPrice} applies to just about any situation where the characteristic function is given. 
%%%%%%%%%%%%%%%%%%%%%%%%%%%%%%
\section{Heston Explicit Formula}
Finally, an important application of the Fourier pricing formula in \eqref{eq:fourierPrice} is to the Heston model. The call option formula for the Heston model from equations \eqref{eq:dSheston} and \eqref{eq:dXheston} has a well-known inverse Fourier transform (see \cite{gatheralBook}) with
\begin{align*}
\Pi_j^{heston}&=\frac 12+\frac{1}{\pi}\int_0^\infty\Real \left[\frac{e^{-iu\log(K)}\phi_T^j(u)}{iu}\right]du\qquad\hbox{for }j=1,2,\\
\phi_T^j(u)&=\exp\left(A_j(u)+B_j(u)X_0+iu(\log(S_0)+rT)\right)\\
A_j(u)&=\frac{\kappa\bar X}{\gamma^2}\left((b_j-\rho \gamma ui-d_j)T-2\log\left(\frac{1-g_je^{-d_jr}}{1-g_j}\right)\right)\\
B_j(u) &= \frac{b_j-\rho \gamma ui-d_j}{\gamma^2}\left(\frac{1-e^{-d_jr}}{1-g_je^{-d_jr}}\right)\\
g_j&= \frac{b_j-\rho \gamma ui-d_j}{b_j-\rho \gamma ui+d_j}\\
d_j&= \sqrt{(\rho\gamma u i-b_j)^2-\gamma^2(2c_jui-u^2)}\\
c_1&=\frac 12,~~c_2=-\frac 12,~~b_1 = \kappa-\rho\gamma,~~b_2 = \kappa\ ,
\end{align*}
which is applied for the Heston model the volatility price of risk absorbed into parameters $\kappa$ and $\bar X$. \textbf{Warning:} Use a well-made Gaussian quadrature function for numerical computation of the integrals for the Heston price; do not attempt to integrate with an evenly-spaced grid.
%%%%%%%%%%%%%%%%%%%%%%%%%%%%%%
\section{The Merton Jump Diffusion}
Short time-to-maturity options are known to be well fit by the class of models where log-asset prices are L\'evy processes. For instance, the Merton jump diffusion,
\[\frac{dS_t}{S_t} = \left(r-\nu\right)dt+\sigma dW_t+\left(e^{J_t}-1\right)dN_t\]
where $\sigma>0$, $N_t$ is an independent Poisson process with parameter $\lambda\in(0,\infty)$, $J_t$ is a sequence of i.i.d. normal random variables with $J_{t_i}\sim N(\mu_J,\sigma_J^2)$ when $N_{t_i}-N_{t_i-} = 1$, and $\nu = \lambda\left(e^{\mu_J+\frac{\sigma_J^2}{2}}-1\right)$ is a compensator to ensure the process is a discounted martingale. The log-price is 
\begin{align*}
\log(S_t/S_0) &= \int_0^t\left(r-\nu-\frac{\sigma^2}{2}\right)du+\sigma W_t+\int_0^tJ_udN_u\\
&=\int_0^t\left(r-\nu-\frac{\sigma^2}{2}\right)du+\sigma W_t+\sum_{n=1}^{N_t}J_n\ .
\end{align*}

For $S_0=1$, the characteristic function is derived as follows,
\begin{align*}
\phi_T(u)&= \mathbb E\exp\left( iu\left(r-\nu-\frac{\sigma^2}{2}\right)T+iu\sigma W_T+iu\int_0^tJ_tdN_t\right)\\
%&=\exp\left( iu\left(r-\nu-\frac{\sigma^2}{2}\right)T-\frac{\sigma^2u^2}{2}T\right)\mathbb E\exp\left(iu\int_0^tJ_tdN_t\right)\\
%&=\exp\left( iu\left(r-\nu-\frac{\sigma^2}{2}\right)T-\frac{\sigma^2u^2}{2}T\right)\mathbb E\mathbb E\left[\exp\left(iu\sum_{n=1}^{N_T}J_n\right)\Big|N_T\right]\\
%&=\exp\left( iu\left(r-\nu-\frac{\sigma^2}{2}\right)T-\frac{\sigma^2u^2}{2}T\right)\mathbb E\prod_{n=1}^{N_T}\mathbb E\exp\left(iuJ_1\right)\\
%&=\exp\left( iu\left(r-\nu-\frac{\sigma^2}{2}\right)T-\frac{\sigma^2u^2}{2}T\right)\mathbb E\left(e^{iu\mu_J-\frac{\sigma_J^2u^2}{2}}\right)^{N_T}\\
%&=\exp\left( iu\left(r-\nu-\frac{\sigma^2}{2}\right)T-\frac{\sigma^2u^2}{2}T\right)\mathbb Ee^{\left(iu\mu_J-\frac{\sigma_J^2u^2}{2}\right)N_T}\\
&=\exp\left( iu\left(r-\nu-\frac{\sigma^2}{2}\right)T-\frac{\sigma^2u^2}{2}T+\lambda T\left(e^{iu\mu_J-\frac{\sigma_J^2u^2}{2}}-1\right)\right)\ .
\end{align*}
This is an example of a characteristic function that is known explicitly. In addition, it is particular example of the general L\'evy-Khintchine representation,
\begin{align*}
&\mathbb E\exp\left(iu\log(S_T)\right)\\
& = \exp\left(  iu\left(r-\nu-\frac{\sigma^2}{2}\right)T-\frac{\sigma^2u^2}{2}T +T\int_{\mathbb R\setminus\{0\}}\left(e^{iux}-1-iux\mathbf 1_{|x|<1}\right)\eta(dx)\right)\ ,
\end{align*}
where $\eta$ is the intensity measure, satisfying $\int_{\mathbb R\setminus\{0\}}1\wedge x^2\eta(dx)<\infty$. For more on the L\'evy-Khintchine representation, see \cite{contTankov}.

\chapter{Stochastic Control}

This chapter takes techniques from stochastic control and applies them to portfolio management. The portfolio can be of varying type, two possibilites are a portfolio for investment of (personal) wealth, or a hedging portfolio with a short position in a derivative contract. The basic problem involves an investor with a self-financing wealth process and a concave utility function to quantify their risk aversion, from which their goal is to maximize their expected utility of terminal wealth and/or consumption. To exemplify the need for hedging obtained from optimal control, recall the price of volatility risk $\Lambda(t,s,x)$ from Proposition \ref{prop:stochasticVolPDE} of Chapter \ref{chapt:stochVol}. The pricing PDE for stochastic volatility depends on $\Lambda$, but incompleteness of the market means that $\Lambda$ may not be uniquely specified. However, an expression can be obtained from the solution to an optimal control, hence writing $\Lambda$ as a function of the investor's risk aversion. This chapter will start by considering the basic problem of optimization of (personal) wealth, and later on will show how optimal control is used in hedging derivatives.

%%%%%%%%%%%%%%%%%%%%%%%%%%
\section{The Optimal Investment Problem}

Consider a standard geometric Brownian motion for the price of a risky asset,
\begin{equation}
\label{eq:dS_physicalMeasure}
\frac{dS_t}{S_t}=\mu dt+\sigma dW_t\ ,
\end{equation}
where $\mu\in\mathbb R$, $\sigma>0$, and $W$ is a Brownian motion under the statistical measure. There is also the risk-free bank account that pays interest at a rate $r\geq 0$. At time $t\geq 0$ the investors has a portfolio value $X_t$ with an allocation $\pi_t$ in the risky asset and a consumption stream $c_t$. The dynamics of the portfolio are self-financing,
\begin{equation}
\label{eq:dX_wealth}
dX_t=X_t\left(rdt+ \pi_t\left(\frac{dS_t}{S_t}-rdt\right)-c_tdt\right)\ ,
\end{equation}
where 
\begin{align*}
\pi_t&=\hbox{proportion of wealth in the risky asset,}\\
c_t&=\hbox{rate of consumption.}
\end{align*}
A natural constraint on consumption is $c_t\geq 0$ for all $t\geq 0$, and based on \eqref{eq:dX_wealth} a constraint of $X_t\geq 0$ is enforced automatically. Here we have taken $X_t\geq 0$ almost surely, but in general any finite lower bound on $X_t$ is necessary to ensure no-arbitrage (i.e. $X_t\geq -M>-\infty $ almost with constant $M$ finite), otherwise there could be doubling strateguies. The optimization problem is then formulated as
\begin{equation}
\label{eq:optimizationProblem}
V(t,x)=\max_{\pi,c\geq0}\mathbb E\left[\int_t^TF(u,c_u,X_u)du+U(X_T)\Big|X_t=x\right]\ ,
\end{equation}
where $F$ is a concave utility on consumption and wealth, $U$ is a concave utility on terminal wealth, and the admissible pairs $(\pi_t,c_t)_{t\geq0}$ are non-anticipating, adapted to $W$, with $\int_0^T|\pi_tX_t|^2dt<\infty$ almost surely. We refer to $\mathbb E\left[\int_t^TF(u,c_u,X_u)du+U(X_T)\Big|X_t=x\right]$ as the objective function, and \textbf{refer to $V$ as the \textit{optimal value function.}}

\begin{example}[Logarithmic utility]
\label{ex:optimalLog}
Suppose that $F=0$ and $U(x) = \log(x)$. There is no utility of consumption, so the optimal is $c_t=0$ for all $t\geq 0$. Now apply It\^o's lemma,
\[d\log(X_t) = rdt+ \pi_t\left(\frac{dS_t}{S_t}-rdt\right) -\frac{\sigma^2\pi_t^2}{2}dt\ ,\]
and then taking expectations,
\[\mathbb E[\log(X_T)|X_t=x] =\log(x)+\mathbb E\left[\int_t^T\left(r+ \pi_u\left(\mu-r\right) -\frac{\sigma^2\pi_u^2}{2}\right)du\Big|X_t=x\right]\ ,\]
where the right-hand side is concave in $\pi_t$. Hence, the optimal strategy is
\[\pi_t = \frac{\mu-r}{\sigma^2}\qquad\forall t\in[0,T]\ ,\]
which is the Sharpe ratio divided by the volatility. The optimal value function is
\[V(t,x)=\max_\pi\mathbb E[\log(X_T)|X_t=x] = \log\left(xe^{\left(r+\frac{(\mu-r)^2}{2\sigma^2}\right)(T-t)}\right)\ ,\]
and using $U^{-1}(v) = e^v$, we find the \textbf{certainty equivalent,}
\[X_t^{ce}=e^{-r(T-t)}U^{-1}\left(V(t,X_t)\right)= X_te^{\left(\frac{(\mu-r)^2}{2\sigma^2}\right)(T-t)}\ ,\]
which is the risk-free rate plus $\tfrac12$ times the Sharpe-ratio squared.
\end{example}

\begin{example}[Log-Utility of Consumption]
\label{ex:logOptimalConsumption}
Suppose that $F(t,c_t,X_t) = e^{-\beta t}\log(c_tX_t)$, $U(x) = 0$, and $T=\infty$. The optimization problem is
\[\max_{\pi,c\geq0}\mathbb E\left[\int_t^\infty e^{-\beta(u-t)}\log(c_uX_u)du\Big|X_t=x\right] = V(x)\ ,\]
which is constant in $t$. Now notice for any admissible $(\pi,c)$ on $[t,t+\Delta t]$ we have the dynamic programming principle,
\[V(X_t)\geq e^{-\beta\Delta t}\mathbb E_tV(X_{t+\Delta t}) + \mathbb E_t\int_t^{t+\Delta t}e^{-\beta(u-t)}\log(c_uX_u)du\ ,\]
with equality if and only if $(\pi,c)$ is chosen optimally over $[t,t+\Delta t]$, and hence
\begin{align*}
&\frac{\mathbb E_tV(X_{t+\Delta t}) -V(X_t)}{\Delta t}\\
& \leq \frac{1-e^{-\beta\Delta t}}{\Delta t}\mathbb E_tV(X_{t+\Delta t}) -\frac{1}{\Delta t} \mathbb E_t\int_t^{t+\Delta t}e^{-\beta(u-t)}\log(c_uX_u)du\\
& \rightarrow  \beta V(X_t) -\log(c_tX_t)\ ,
\end{align*}
as $\Delta t\rightarrow 0$. On the other hand, from It\^o's lemma we have
\begin{align*}
dV(X_t) &= \left(\frac{\sigma^2\pi_t^2X_t^2}{2}\frac{\partial^2}{\partial x^2}V(X_t)+\left(r+\pi_t(\mu-r)-c_t\right)X_t\frac{\partial}{\partial x}V(X_t)\right)dt\\
&\hspace{3cm}+\sigma\pi_tX_t\frac{\partial}{\partial x}V(X_t)dW_t\ ,
\end{align*}
and assuming the Brownian term vanishes under expectations, we have
\begin{align*}
&\frac{\mathbb E_tV(X_{t+\Delta t}) -V(X_t)}{\Delta t}\\
& =\frac{1}{\Delta t}\mathbb E_t\int_t^{t+\Delta t}\left(\frac{\sigma^2\pi_u^2X_u^2}{2}\frac{\partial^2}{\partial x^2}V(X_u)+\left(r+\pi_u(\mu-r)-c_u\right)X_u\frac{\partial}{\partial x}V(X_u)\right)du\\
& \rightarrow   \frac{\sigma^2\pi_t^2X_t^2}{2}\frac{\partial^2}{\partial x^2}V(X_t)+\left(r+\pi_t(\mu-r)-c_t\right)X_t\frac{\partial}{\partial x}V(X_t)\ , 
\end{align*}
as $\Delta t\rightarrow 0$. Hence, for all admissible pairs $(\pi,c)$ the value function $V(x)$ satisifies
\[ \frac{\sigma^2\pi^2x^2}{2}\frac{\partial^2}{\partial x^2}V(x)+\left(r+\pi(\mu-r)-c\right)x\frac{\partial}{\partial x}V(x) - \beta V(x) +\log(cx)\leq 0\ ,\]
with equality if and only $\pi$ and $c$ are optimal, which leads to the equation
\[\max_{\pi,c\geq 0}\left( \frac{\sigma^2\pi^2x^2}{2}\frac{\partial^2}{\partial x^2}V(x)+\left(r+\pi(\mu-r)-c\right)x\frac{\partial}{\partial x}V(x) - \beta V(x) +\log(cx)\right)=0\ .\]
Let's assume the ansatz
\[V(x) = a\log(x) +b \ .\]
Then through first-order optimality conditions (i.e. by differentiating with respect to $\pi$ and setting equal to zero) we find the optimal
\[\pi(x)= -\frac{\mu-r}{\sigma^2 x}\frac{\frac{\partial}{\partial x}V(x)}{\frac{\partial^2}{\partial x^2}V(x)}=\frac{\mu-r}{\sigma^2 }\ .\]
Similarly, first-order optimality conditions for $c$ yield
\[c(x)= \frac{1}{x\frac{\partial}{\partial x}V(x)}=\frac{1}{a}\ .\]
Putting optimal $\pi_t$ and $c_t$ back into the equation for $V$ along with the ansatz, we find 
\[\log(x/a)-\beta(a\log(x)+b)+\left(ar- 1\right)+\frac a2\frac{(\mu-r)^2}{\sigma^2}=0\ ,\]
and comparing $\log(x)$ terms and non-$x$-dependent terms we find,
\begin{align*}
a&=\frac1\beta \ ,\\
b&=\frac{1}{2\beta^2}\frac{(\mu-r)^2}{\sigma^2}+\frac{r}{\beta^2}+\frac1\beta\left(\log(\beta)-1\right)\ .
\end{align*}
\end{example}

%%%%%%%%%%%%%%%%%%%%%%%%%%
\section{The Hamilton-Jacobi-Bellman (HJB) Equation}
Example \ref{ex:optimalLog} is useful to get started and to get a sense for how an optimal control should look. Example \ref{ex:logOptimalConsumption} is more instructive because it shows us how (i) the function $V$ inherits concavity from $F$ and $U$, and (ii) how it also shows how to derive the PDE that $V$ should satisfy.

The derivation starts with the dynamic programming principle,
\[V(t,x)=\max_{\pi,c\geq0}\mathbb E\left[\int_t^{t+\Delta t}F(u,c_u,X_u)du+V(t+\Delta t,X_{t+\Delta t})\Big|X_t=x\right]\ ,\]
where $\max_{\pi,c}$ is taken over the interval $[t,t+\Delta t]$. Applying It\^o's lemma to $V(t,X_t)$, we find 
\begin{align*}
&V(t+\Delta t,X_{t+\Delta t})\\
 &= V(t,X_t)+ \int_t^{t+\Delta t}\left(\frac{\partial}{\partial t}+\frac{\sigma^2\pi_u^2X_u^2}{2}\frac{\partial^2}{\partial x^2}+\left(r+\pi_u(\mu-r)-c_u\right)X_u\frac{\partial}{\partial x}\right)V(u,X_u)du\\
&\hspace{4cm}+\sigma\int_t^{t+\Delta t}\pi_uX_u\frac{\partial}{\partial x}V(u,X_u)dW_u\ ,
\end{align*}
for any admissible $(\pi,c)$ over  $[t,t+\Delta t]$. Hence, for any $(\pi,c)$ on $[t,t+\Delta t]$ we have 
\begin{align*}
&\mathbb E\left[\int_t^{t+\Delta t}\Bigg(F(u,c_u,X_u)+\Bigg(\frac{\partial}{\partial t}+\frac{\sigma^2\pi_u^2X_u^2}{2}\frac{\partial^2}{\partial x^2}\right.\\
&\hspace{3cm}+\left.\left(r+\pi_u(\mu-r)-c_u\right)X_u\frac{\partial}{\partial x}\Bigg)V(u,X_u)\Bigg)du\Big|X_t=x\right]\leq 0\ ,
\end{align*}
with equality iff and only if an optimal $(\pi,c)$ is chosen. Hence, dividing by $\Delta t$ and taking the limt to zero, we obtain the so-called \textbf{Hamilton-Jacobi-Bellman (HJB) equation:}
\begin{align}
\label{eq:HJB}
\max_{\pi,c\geq0}\left(F(t,c)+\Bigg(\frac{\partial}{\partial t}+\frac{\sigma^2\pi^2x^2}{2}\frac{\partial^2}{\partial x^2}+\left(r+\pi(\mu-r)-c\right)x\frac{\partial}{\partial x}\Bigg)V(t,x)\right)&=0\ ,\\
\nonumber
V(T,x)&=U(x)\ .
\end{align}

%%%%%%%%%%%%%%%%%%%%%%%%%%
\section{Merton's Optimal Investment Problem}
Let $F=0$ and consider a power utility function,
\[U(x) = \frac{x^{1-\gamma}}{1-\gamma}\ ,\]
where $\gamma>0$, $\gamma\neq1$ is the risk aversion. The problem is to solve
\[V(t,x)=\max_\pi\mathbb E[U(X_T)|X_t=x]\ ,\]
The HJB equation for this problem is 
\begin{align}
\label{eq:HJBmerton}
\left(\frac{\partial}{\partial t}+rx\frac{\partial}{\partial x}\right)V(t,x)+\max_{\pi}\left(\frac{\sigma^2\pi^2x^2}{2}\frac{\partial^2}{\partial x^2}V(t,x)+\pi(\mu-r)x\frac{\partial}{\partial x}V(t,x)\right)&=0\ ,\\
\nonumber
V(T,x)&=U(x)\ ,
\end{align}
for which we find the optimal $\pi$,
\[\pi_t = -\frac{\mu-r}{x\sigma^2} \frac{\frac{\partial}{\partial x}V(t,x)}{\frac{\partial^2}{\partial x^2}V(t,x)}\ . \]
Inserting the optimal $\pi_t$ into \eqref{eq:HJBmerton} we obtain the nonlinear equation,
\begin{equation}
\label{eq:HJBmerton_nonlinear}
\left(\frac{\partial}{\partial t}+rx\frac{\partial}{\partial x}\right)V(t,x)-\frac{\left((\mu-r)\frac{\partial}{\partial x}V(t,x)\right)^2}{2\sigma^2\frac{\partial^2}{\partial x^2}V(t,x)}=0\ .
\end{equation}
Then using the ansatz $V(t,x) = g(t)U(x)$, we find
\begin{align*}
\frac{\partial}{\partial t}V(t,x)&= g'(t)U(x)\ ,\\
\frac{\partial}{\partial x}V(t,x)&= \frac{1-\gamma}{x}g(t)U(x)\ ,\\
\frac{\partial^2}{\partial x^2}V(t,x)&= -\frac{\gamma(1-\gamma)}{x^2}g(t)U(x)\ ,
\end{align*}
and inserting in \eqref{eq:HJBmerton_nonlinear} we find an ODE for $g$,
\[g'(t) +r(1-\gamma)+\frac{(1-\gamma)(\mu-r)^2g(t)}{2\gamma\sigma^2}=0\ ,\]
with terminal condition $g(T)=1$. The solution is
\[g(t) = e^{(1-\gamma)(T-t)\left(r+\frac{(\mu-r)^2}{2\gamma\sigma^2}\right)}\ ,\]
and the optimal value function is
\[V(t,x) = U(x)g(t)=\frac{\left(xe^{(T-t)\left(r+\frac{(\mu-r)^2}{2\gamma\sigma^2}\right)}\right)^{1-\gamma}}{1-\gamma}=U\left(xe^{(T-t)\left(r+\frac{(\mu-r)^2}{2\gamma\sigma^2}\right)}\right)\ .\]
and the certainty equivalent is
\[X_t^{ce}=e^{-r(T-t)}U^{-1}(v(t,X_t)) = X_te^{(T-t)\left(\frac{(\mu-r)^2}{2\gamma\sigma^2}\right)}\ .\]

\section{Stochastic Returns}
Consider the model
\begin{eqnarray}
\label{eq:SRM_stochControl}
dS_t&=&Y_t S_tdt+\sigma S_tdW_t\\
\label{eq:SRM_dY_stochControl}
dY_t &=&\kappa(\mu-Y_t)dt+\beta dB_t\ ,
\end{eqnarray}
with $dW_tdB_t=\rho dt$ where $\rho\in(-1,1)$. The interpretation of $Y_t$ could be any of the following: $Y_t$ is a dividend yield with uncertainty (although somewhat of strange model because it can be negative), or $Y_t$ is the return rate on a commodities or bond portfolio where there is a role yield due to contango or backwardation.
 
Let's assume the simple case $\mu=r=0$, for which the value function is
\[V(t,x,y) = \max_\pi\mathbb E\left[U(X_T)\Big|X_t=x,Y_t=y\right]\ ,\]
and has HJB equation
\begin{align}
\nonumber
\left(\frac{\partial}{\partial t}+\frac{\beta^2}{2}\frac{\partial^2}{\partial y^2}-\kappa y\frac{\partial}{\partial y}\right)V(t,x,y)\hspace{4cm}&\\
\nonumber
+\max_{\pi}\Bigg(\frac{\sigma^2x^2\pi^2}{2}\frac{\partial^2}{\partial x^2}V(t,x,y)+\pi xy\frac{\partial}{\partial x}V(t,x,y)\hspace{2cm}&\\
\label{eq:HJBstochReturns}
+\rho \pi x\beta\sigma\frac{\partial^2}{\partial x\partial y}V(t,x,y)\Bigg)&=0\ ,\\
\nonumber
V(T,x,y)&=U(x)\ .
\end{align}
The first-order condition for $\pi$ yields the optimal 
\[\pi_t  =  -\frac{xy\frac{\partial}{\partial x}V(t,x,y)+\rho \pi x\beta\sigma\frac{\partial^2}{\partial x\partial y}V(t,x,y)}{\sigma^2x^2\frac{\partial^2}{\partial x^2}V(t,x,y)}\ .\]
For the power utility
\[U(x)  = \frac{x^{1-\gamma}}{1-\gamma}, \]
we have the ansatz $V(t,x,y) = U(x)g(t,y)$ with 
\begin{align*}
\frac{\partial}{\partial t}V& = \frac{\partial}{\partial t}g(t,y)U(x)\ ,\\
\frac{\partial}{\partial x}V& =\frac{1-\gamma}{x} g(t,y)U(x)\ ,\\
\frac{\partial^2}{\partial x^2}V& =-\frac{(1-\gamma)\gamma}{x^2} g(t,y)U(x)\ ,\\
\frac{\partial^2}{\partial x\partial y}V& =\frac{1-\gamma}{x} \frac{\partial}{\partial y}g(t,y)U(x)\ ,
\end{align*}
all of which are inserted into equation \eqref{eq:HJBstochReturns} to get an equation for $g$:
\begin{align}
\nonumber
\left(\frac{\partial}{\partial t}+\frac{\beta^2}{2}\frac{\partial^2}{\partial y^2}-\kappa y\frac{\partial}{\partial y}\right)g(t,y)+\frac{1-\gamma}{2\sigma^2\gamma}\left(y+\frac{\rho \beta\sigma\frac{\partial}{\partial y}g(t,y)}{g(t,y)}\right)^2g(t,y)&=0\ ,\\
\nonumber
g(T,y)&=1\ .
\end{align}
We now apply another ansatz $g(t,y) = e^{a(t)y^2+b(t)}$, which when inserted into the equation for $g(t,y)$ yields the following system:
\begin{align*}
y^2&:~a'(t)=-2\beta^2\left(1+\frac{(1-\gamma)\rho^2}{\gamma}\right)a^2(t)-2\left(\frac{\rho\beta(1-\gamma)}{\sigma\gamma}-\kappa\right)a(t)-\frac{1-\gamma}{2\sigma^2\gamma}\ ,\\
1&:~b'(t)=-\beta^2a(t)\ ,
\end{align*}
with terminal conditions $a(T)=b(T)=0$. The solution $a(t)$ can be written as a ratio,
\[a(t) = \frac{v'(t)}{2\beta^2 \left(1+\frac{(1-\gamma)\rho^2}{\gamma}\right)v(t)}\ ,\]
where $v(t)$ is the solution to a 2nd-order ODE,
\[v''(t) +2\left(\frac{\rho\beta(1-\gamma)}{\sigma\gamma}-\kappa\right)v'(t)+\frac{(1-\gamma)\beta^2}{\sigma^2\gamma}\left(1+\frac{(1-\gamma)\rho^2}{\gamma}\right)v(t)=0\ . \]
The roots of this equation are
\[m_\pm =-\left(\frac{\rho\beta(1-\gamma)}{\sigma\gamma}-\kappa\right)\pm\sqrt{\kappa^2-\frac{\beta(1-\gamma)}{\sigma\gamma}\left(2\kappa\rho+\frac{\beta}{\sigma}\right)}\ ,\]
which gives the general solution 
\[v(t) = C_1e^{m_+(T-t)}+C_2e^{m_-(T-t)}\ .\]
It is not necessary to fully determine constants $C_1$ and $C_2$ because we are mainly interested in the ratio $v'(t)/v(t)$.\\

\noindent\textbf{Finite-Time Blowup.} Complex valued $m_\pm$ leads to finite-time blowup for the optimization problem. If the roots are complex then let $c = -\left(\frac{\rho\beta(1-\gamma)}{\sigma\gamma}-\kappa\right)$ and $d = \frac{\beta(1-\gamma)}{\sigma\gamma}\left(2\kappa\rho+\frac{\beta}{\sigma}\right)-\kappa^2$ so that the general solution is
\[v(t) = e^{c(T-t)}\Big(C_1\cos(d(T-t))+C_2\sin(d(T-t))\Big)\ ,\]
and with $v'(T)=-cC_1-dC_2=0$ to satisfy the terminal condition $a(T)=0$, so that
 
\[v(t) = C_1e^{c(T-t)}\left(\cos(d(T-t))-\frac{c}{d}\sin(d(T-t))\right)\ .\]
 The solution $a(t)$ will blow at time $t^*$ such that $\cos(d(T-t^*))-\frac{c}{d}\sin(d(T-t^*))=0$, that is $\tan(d(T-t^*))=\frac{d}{c}$ or 
\[T-t^* = \frac{1}{d}\left(\pi\indicator{c\leq0}+\tan^{-1}\left(\frac{d}{c}\right)\right)\ .\]

%%%%%%%%%%%%%%%%%%%%%%%%%%
\section{Stochastic Volatility}
Now let's consider the same optimal terminal wealth problem as the Merton problem, with exponential utility
\[U(x) = -\frac1\gamma e^{-\gamma x}\qquad\hbox{where }\gamma>0\ ,\]
and in the incomplete market of stochastic volatility, 
\begin{eqnarray}
\label{eq:SVM_stochControl}
dS_t&=&\mu S_tdt+\sigma(Y_t)S_tdW_t\\
\label{eq:dY_stochControl}
dY_t &=&\alpha(Y_t)dt+\beta(Y_t)dB_t\ ,
\end{eqnarray}
where $dW_t\cdot dB_t = \rho dt$. From \eqref{eq:SVM_stochControl} and \eqref{eq:dY_stochControl}, we have the wealth process,
\[dX_t = rX_tdt+\pi_t\left(\frac{dS_t}{S_t}-rdt\right)\ ,\]
no longer enforcing the non-negativity constraint. The optimization problem is
\[V(t,x,y) = \max_\pi\mathbb E\left[U(X_T)\Big|X_t=x,Y_t=y\right]\ ,\]
but the technique used in Example \ref{ex:optimalLog} does not apply because there is some local martingale behavior in the stochastic integrals. Instead, we arrive at the optimal solution using the HJB equation. The HJB equation is
\begin{align}
\nonumber
\left(\frac{\partial}{\partial t}+rx\frac{\partial}{\partial x}+\frac{\beta^2(y)}{2}\frac{\partial^2}{\partial y^2}+\alpha(y)\frac{\partial}{\partial y}\right)V(t,x,y)\hspace{4cm}&\\
\nonumber
+\max_{\pi}\Bigg(\frac{\sigma^2(y)\pi^2}{2}\frac{\partial^2}{\partial x^2}V(t,x,y)+\pi(\mu-r)\frac{\partial}{\partial x}V(t,x,y)\hspace{2cm}&\\
\label{eq:HJBstochVol}
+\rho \pi \beta(y)\sigma(y)\frac{\partial^2}{\partial x\partial y}V(t,x,y)\Bigg)&=0\ ,\\
\nonumber
V(T,x)&=U(x)\ .
\end{align}
Using the ansatz $V(t,x,y) = U(xe^{r(T-t)})g(t,y)$, we have
\begin{align*}
\frac{\partial}{\partial t}V& = \gamma rx e^{r(T-t)} V+U(xe^{r(T-t)})\frac{\partial}{\partial t}g(t,y)\ ,\\
\frac{\partial}{\partial x}V& = -\gamma e^{r(T-t)} V\ ,\\
\frac{\partial^2}{\partial x^2}V& = \gamma^2 e^{2r(T-t)}V\ ,\\
\frac{\partial^2}{\partial x\partial y}V& = -\gamma e^{r(T-t)} U(xe^{r(T-t)})\frac{\partial}{\partial y}g(t,y)\ ,\\
\end{align*}
which we insert into \eqref{eq:HJBstochVol} to find a PDE for $g$,

\begin{align}
\nonumber
\left(\frac{\partial}{\partial t}+\frac{\beta^2(y)}{2}\frac{\partial^2}{\partial y^2}+\alpha(y)\frac{\partial}{\partial y}\right)g(t,y)\hspace{5cm}&\\
\label{eq:HJBstochVol_g}
+\min_{\pi}\Bigg(\frac{\gamma^2e^{2r(T-t)}\sigma^2(y)\pi^2}{2}g(t,y)-\gamma e^{r(T-t)}\pi\left((\mu-r)g(t,y)+\rho  \beta(y)\sigma(y)\frac{\partial}{\partial y}g(t,y)\right)\Bigg)&=0\ ,\\
\nonumber
g(T,y)&=1\ ,
\end{align}
and the optimal strategy is
\[\pi_t = e^{-r(T-t)}\left(\frac{\mu-r}{\gamma\sigma^2(y)}+\rho \frac{ \beta(y)}{\gamma\sigma(y)}\frac{\frac{\partial}{\partial y}g(t,y)}{g(t,y)}\right)\ .\]
Inserting this optimal $\pi_t$ into \eqref{eq:HJBstochVol_g} yields the nonlinear equation
\begin{align}
\label{eq:stochVol_Merton_g}
\left(\frac{\partial}{\partial t}+\frac{\beta^2(y)}{2}\frac{\partial^2}{\partial y^2}+\alpha(y)\frac{\partial}{\partial y}\right)g(t,y)-\frac{\sigma^2(y)}{2}\left(\frac{\mu-r}{\sigma^2(y)}+\rho \frac{ \beta(y)}{\sigma(y)}\frac{\frac{\partial}{\partial y}g(t,y)}{g(t,y)}\right)^2 g(t,y)&=0\ .
\end{align}
This equation can be reduced to a linear PDE if we look for a function $\psi(t,y)$ such that
\[g(t,y) = \psi(t,y)^q\ ,\]
where $q$ is a parameter. Differentiating yields,
\begin{align*}
\frac{\partial}{\partial t}g &= \frac{qg}{\psi}\frac{\partial}{\partial t}\psi\\
\frac{\partial}{\partial y}g &= \frac{qg}{\psi}\frac{\partial}{\partial y}\psi\\
\frac{\partial^2}{\partial y^2}g &= qg\left(\frac1\psi\frac{\partial^2}{\partial y^2}\psi+\frac{q-1}{\psi^2}\left(\frac{\partial}{\partial y}\psi\right)^2\right)\ ,
\end{align*}
and then plugging into \eqref{eq:stochVol_Merton_g} with chosen parameter $q=1/(1+\rho^2)$ yields a linear equation:
\begin{align}
\label{eq:linearPDEstochVolControl}
\left(\frac{\partial}{\partial t}+\frac{\beta^2(y)}{2}\frac{\partial^2}{\partial y^2}+\left(\alpha(y)-\rho \frac{(\mu-r) \beta(y)}{\sigma(y)}\right)\frac{\partial}{\partial y}\right)\psi(t,y)-\frac{(\mu-r)^2}{2q\sigma^2(y)}\psi(t,y)&=0\ .
\end{align}
%%%%%%%%%%%%%
\begin{example}[Fully Affine Heston Model]
Consider a futures contract $F_{t,T}$ with settlement date $T$ and stochastic volatility and returns,
\begin{align*}
\frac{dF_{t,T}}{F_{t,T}}&=\mu Y_tdt+\sqrt{Y_t}dW_t\\
dY_t&=\kappa(\bar Y-Y_t)dt+\beta\sqrt{Y_t}dB_t
\end{align*}
where $\beta^2\leq 2\kappa\bar Y$ and $dW_tdB_t=\rho dt$. The wealth process for futures trading is
\[dX_t = rX_tdt+ \pi_t\frac{dF_{t,T}}{F_{t,T}}\ .\]
For $U(x) = -\frac1\gamma e^{-\gamma x}$ the optimal terminal expected utility is \[V(t,x,y)=U(xe^{r(T-t)})\psi(t,y)^q\ ,\]
where $\psi(t,y)$ is similar to a solution to equation \eqref{eq:linearPDEstochVolControl}, except the equation has no $r$,
\begin{align*}
\left(\frac{\partial}{\partial t}+\frac{\beta^2y}{2}\frac{\partial^2}{\partial y^2}+\left(\kappa(\bar Y-y)-\rho \mu\beta y\right)\frac{\partial}{\partial y}\right)\psi(t,y)-\frac{\mu^2y}{2q}\psi(t,y)&=0\ .
\end{align*}
It can be further shown that the solution to this equation is of the form
\[\psi(t,y) = e^{a(t)y+b(t)}\ ,\]
with $a(T)=b(T)=0$, and where $a(t)$ and $b(t)$ satisfy ODEs,
\begin{align*}
a'(t)+\frac{\beta^2}{2}a^2(t)-\left(\kappa+\rho \mu\beta \right)a(t)-\frac{\mu^2}{2q}&=0\\
b'(t)+\kappa\bar Ya(t)&=0\ ,
\end{align*}
both of which can be solved explicitly.
\end{example}

%%%%%%%%%%%%%%%%%%%%%%%%%%
\section{Indifference Pricing}
Stochastic control for terminal wealth can be implemented to find the the price of a call option under stochastic volatility,
\[V^{h}(t,x,y,s) = \max_\pi\mathbb E\left[U(X_T-(S_T-K)^+)\Big|X_t=x,Y_t=y,S_t=s\right]\ ,\]
where the investor now hedges a short position in a call option with strike $K$. Compared to the same investor's value function that is not short the call
\[V^0(t,x,y) = \max_\pi\mathbb E\left[U(X_T)\Big|X_t=x,Y_t=y\right]\ ,\]
we look for the amount of cash $\$p$ such that the 
\[V^{h}(t,x+p,y,s) = V^{0}(t,x,y)\ .\]
The extra cash makes the hedger \textit{utility indifferent} to the short position. With exponential utility there is a separation of variables,
\begin{align*}
V(t,x,y,s)&=\max_\pi\mathbb E\left[-\frac1\gamma e^{-\gamma(X_T-(S_T-K)^+)}\Big|X_t=x,Y_t=y,S_t=s\right]\\
&=-\frac1\gamma e^{-\gamma xe^{r(T-t)}}\min_\pi\mathbb E\left[e^{-\gamma\left(\int_t^Te^{r(T-u)}\pi_u\left(\frac{dS_u}{S_u}-rdu\right)-(S_T-K)^+\right)} \Big|Y_t=y,S_t=s\right]\\
&=U\left(xe^{r(T-t)}\right)g^h(t,y,s)\ ,
\end{align*}
where we've used the differential $d\left(e^{r(T-t)}X_t\right) = e^{r(T-t)}\pi_t\left(\frac{dS_t}{S_t}-rdt\right)$. Hence we find a price $\$p$ such that $U\left(pe^{r(T-t)}\right) = g(t,y)/g^h(t,y,s)$, where $g(t,y)$ is the solution from \eqref{eq:HJBstochVol_g}. 

Depending on the risk-aversion coefficient $\gamma$, there will be different prices $\$p$. This brings us back to the price of volatility risk $\Lambda(t,s,x)$ from Proposition \ref{prop:stochasticVolPDE} of Chapter \ref{chapt:stochVol}. Namely, investors with different risk aversion will have a different $\Lambda$ for their martingale evaluation of the call option.

If an indifference price is obtained then there is a solution to both optimization problems, and hence there is no-arbitrage and the range of prices for $\$c$ will be a no-arbitrage interval. For complete markets there will be a single price $\$c$ for all levels of risk aversion.

%%%%%%%%%%%%%%%%%%%%%%%%%%
%\section{Verification Lemma}

%%%%%%%%%%%%%%%%%%%%%%%%%%
%\section{Drawdown Constraint}

\appendix
\addappheadtotoc
\appendixpage
\noappendicestocpagenum
\chapter{Martingales and Stopping Times}
\label{app:martingalesStoppingTimes}
Let $\mathcal F_t$ denote a $\sigma$-algebra. A process $X_t$ is an $\mathcal F_t$-martingale if and only if
\[\mathbb E_tX_T = X_t\qquad\forall t\leq T\ ,\]
where $\mathbb E_t=\mathbb E[~\cdot~|\mathcal F_t]$. A process $X_t$ is a submartingale if and only if
\[\mathbb E_tX_T \geq X_t\qquad\forall t\leq T\ ,\]
and a supermartingale if and only if
\[\mathbb E_tX_T \leq X_t\qquad\forall t\leq T\ .\]
Note that true martingale is both a sub and supermartingale. 

%%%%%%%%%%%%%%%%%%%%%%%%%%%%%%%
\section{Stopping Times}
In probability theory, a stopping time is a a stochastic time that is non-anticipative of the underlying process. For instance, for a stock price $S_t$ a stopping time is the first time the price reaches a level $M$,
\[\tau = \inf\{t>0|S_t\geq M\}\ .\]
The non-anticipativeness of the stopping is important because there are some events that are seemingly similar but are not stopping times, for instance 
\[\nu=\sup\{t>0|S_t<M\}\]
\textbf{is not a stopping time.}

Stopping times are useful when discussing martingales. For example, so-called stopped-processes inherit the sub or supermartingale property. Namely, $X_{t\wedge\tau}$ is a sub or supermartingale of $X_t$ is a sub or supermartingale, respectively. There is also the optional stopping theorem:
\begin{theorem}[Optional Stopping Theorem] Let $X_t$ be a submartingale and let $\tau$ be a stopping time. If $\tau<\infty $ a.s. and $X_{t\wedge\tau}$ uniformly integrable, then $\mathbb EX_0\leq \mathbb EX_\tau$ with equality if $X_t$ is a martingale.
\end{theorem}
An example application of the optional stopping theorem is Gambler's ruin: Let 
\[\tau=\inf\{t>0|W_t\notin (a,b)\}\ ,\]
where $0<b<\infty$ and $-\infty<a<0$. Then $\mathbb P(\tau<\infty)=1$ and by optional stopping,
\begin{align*}
0&=W_0=\mathbb EW_\tau\\
&= a\mathbb P(W_\tau=a)+b(1-\mathbb P(W_\tau=a))\\
&=(a-b)\mathbb P(W_\tau=a)+b\ ,
\end{align*}
which can be simplified to get
\[\mathbb P(W_\tau=a) = \frac{b}{b-a}\ .\]

%%%%%%%%%%%%%%%%%%%%%%%%%%%%%%%%%
\section{Local Martingales}
First define a local martingale:
\begin{definition}[Local Martingale]
A process $X_t$ is a local martingale if there exists a sequence of finite and increasing stopping times $\tau_n$ such that $\mathbb P(\tau_n\rightarrow\infty\hbox{ as }n\rightarrow\infty)=1$ and $X_{t\wedge\tau_n}$ is a true martingale for any $n$. 
\end{definition}
Some remarks are in order:
\begin{remark}
In discrete time there are no local martingales; a martingale is a martingale.
\end{remark}

\begin{remark}
A true martingale $X_t$ is a local martingale, and any bounded local martingale is in fact a true martingale.
\end{remark}

The It\^o stochastic integral is in general a local martingale, not necessarily a true martingale. That is, 
\[I_t = \int_0^t\sigma_udW_u\ ,\]
is only a local martingale, but there exists stopping times $\tau_n$ such that
\[I_{t\wedge\tau_n} =  \int_0^{t\wedge\tau_n}\sigma_udW_u\ ,\]
is a true martingale. For It\^o integrals there is the following theorem for a sufficient (but not necessary) condition for true martingales:
\begin{theorem}
\label{thm:finiteItoIsometry}
The It\^o integral $ \int_0^t\sigma_udW_u$ is a true martingale on $[0,T]$ if
\[\mathbb E\int_0^T\sigma_s^2ds<\infty \ ,\]
i.e., the It\^o isometry is finite.
\end{theorem}
For the stochastic integral $I_t = \int_0^t\sigma_udW_u$ we can define
\[\tau_n = \inf\left\{t>0\Bigg|\int_0^t\sigma_u^2ds\geq n\right\}\wedge T\ ,\]
for which we have a bounded It\^o isometry, and hence Theorem \ref{thm:finiteItoIsometry} applies to make $I_{t\wedge\tau_n}$ a martingale on $[0,T]$.

An example of a local martingale is the constant elasticity of volatility (CEV) model,
\[dS_t = \sigma S_t^{\alpha}dW_t\ ,\]
with $0\leq\alpha\leq 2$; $S_t$ is strictly a local martingale for $1<\alpha\leq 2$. For $\alpha=2$ one can check using PDEs that the transition density is
\[p_t(z|s) = \frac{s}{z^3\sqrt{2\pi t\sigma^2}}\left(e^{-\frac{\left(\frac1z-\frac1s\right)^2}{2t\sigma^2}}-e^{-\frac{\left(\frac1z+\frac1s\right)^2}{2t\sigma^2}}\right)\ .\]
One can check that $\mathbb ES_t^4=\infty$ for all $t>0$ so that Theorem \ref{thm:finiteItoIsometry} does not apply, but to see that it is a strict local martingale one must also check that $\mathbb E_tS_T<S_t$ for all $t<T$.

%%%%%%%%%%%%%%%%%%%%%%%%%%%%%%%%%%%
%\section{Dynkin's Formula}

\chapter{Some Notes on Fourier Transforms}
For a real-valued function $f(x)$ with $\int_{-\infty}^\infty |f(x)|^2dx<\infty$, Fourier transforms are defined as follows
\begin{align}
\label{eq:fourierIntegral_A}
\widehat f(u)&=\int_{-\infty}^\infty  e^{iux}f(x)dx\\
\label{eq:inverseFourierIntegral_A}
f(x)&=\frac{1}{2\pi}\int_{-\infty}^\infty  e^{-iux}\widehat f(u)du\ .
\end{align}
The Fourier transforms in \eqref{eq:fourierIntegral_A} and \eqref{eq:inverseFourierIntegral_A} are somewhat different from traditional definitions, which are
\begin{align*}
\widehat f(u)&=\int_{-\infty}^\infty  e^{-2\pi iux}f(x)dx\\
f(x)&=\int_{-\infty}^\infty  e^{2\pi iux}\widehat f(u)du\ .
\end{align*}
The reason we choose \eqref{eq:fourierIntegral_A} in finance is because we work so much with probability theory, and as probabilists we like to have a Fourier transform of a density be equal to a characteristic function. However, nothing changes and we are able to carry out all the same calculations; the difference amounts merely to a change of variable inside the integrals.
%%%%%%%%%%%%%%%%%%%%%%%%%%%%%%%%%%
\section{Some Basic Properties}
\begin{itemize}
\item Linearity: if $\psi(s) = af(s)+bg(s)$, then 
\[\widehat \psi(u) = a\widehat f(u)+b\widehat g(u)\ .\]
\item Translation/Time Shifting: if $f(s) = \psi(s-s_0)$, then 
\[\widehat f(u)= e^{ius_0}\widehat \psi(u)\ .\]
\item Convolution: for $f*g(x) = \int_{-\infty}^\infty f(y)g(x-y)dy$,
\begin{align*}
\widehat {f*g}(u) &= \int_{-\infty}^\infty e^{iux} \int_{-\infty}^\infty f(y)g(x-y)dydx\\
&=\int_{-\infty}^\infty e^{iuy}f(y)\int_{-\infty}^\infty e^{iu(x-y)} g(x-y)dxdy\\
& = \int_{-\infty}^\infty e^{iuy}f(y)\widehat g(u)dy\\
&=\widehat f(u)\widehat g(u)\ .
\end{align*}
\item Reverse Convolution: 
\begin{align*}
\widehat{\widehat f*\widehat g}(x) &= \frac{1}{2\pi}\int_{-\infty}^\infty e^{-iux} \int_{-\infty}^\infty \widehat f(v)\widehat g(u-v)dvdu\\
&= \frac{1}{2\pi}\int_{-\infty}^\infty e^{-ivx} \widehat f( v) \int_{-\infty}^\infty e^{-i(u- v)x}\widehat g(u- v)dudv\\
&=g(x)\int_{-\infty}^\infty e^{-i vx} \widehat f( v)dv\\
&= 2\pi g(x)f(x)\ .
\end{align*}
\end{itemize}
An important concept to realize about Fourier is that functions $(e^{iux})_{u\in\mathbb R}$ can be thought of as orthonormal basis elements in a linear space. Similar to finite-dimensional vector spaces, we can write a function as a some inner products with the basis elements. Indeed, that is what we accomplish with the inverse Fourier transform; we can think of the integral as a sum,
\[f(x)=\frac{1}{2\pi}\int_{-\infty}^\infty  e^{-iux}\widehat f(u)du\approx\frac{1}{2\pi} \sum_{n=1}^N e^{-iu_nx}\widehat f(u_n) \ ,\]
where $u_n$ are discrete points $\mathbb C$ (this is a similar idea to Fast Fourier Transforms (FFT)). The basis elements are orthogonal in the sense that
\[\frac{1}{2\pi}\int_{-\infty}^\infty e^{iux}e^{-ivx}dx = \delta_0(u-v)\ .\]
Finally, it needs to be pointed out that the function $\delta_0$ is something that defined under integrals, and is rather hard to formalize outside. For instance, in Parseval's identity, we used $\delta_0$, but only inside the integral:
\begin{proposition}
For a real-valued function $f(x)$ with $\int_{-\infty}^\infty |f(x)|^2dx<\infty$, 
\[\int_{-\infty}^\infty |f(x)|^2ds =\frac{1}{2\pi}\ \int_{-\infty}^\infty |\widehat f(u)|^2du\ .\]
\end{proposition}
\begin{proof}
\begin{align*}
\frac{1}{2\pi}\int_{-\infty}^\infty |\widehat f(u)|^2du&=\frac{1}{2\pi}\int_{-\infty}^\infty \widehat f(u)\overline{ \widehat f(u)}du\\
&=\frac{1}{2\pi}\int_{-\infty}^\infty \int_{-\infty}^\infty \int_{-\infty}^\infty e^{iux}f(x)e^{-iuy}f(y)dxdydu\\
&=\frac{1}{2\pi}\int_{-\infty}^\infty \int_{-\infty}^\infty f(x)f(y) \left(\int_{-\infty}^\infty e^{iux}e^{-iuy}du\right)dxdy\\
&=\int_{-\infty}^\infty \int_{-\infty}^\infty f(x)f(y) \delta_0(y-x)dxdy\\
&=\int_{-\infty}^\infty  f(x)f(x) dx \ .
\end{align*}
\end{proof}

%%%%%%%%%%%%%%%%%%%%%%%%%%%%%%%%%%
\section{Regularity Strips}
For some function $f:\mathbb R\rightarrow \mathbb R$ we modify \eqref{eq:fourierIntegral_A} and \eqref{eq:inverseFourierIntegral_A} to accommodate non-integrability and/or non-differentiability. We write
\begin{align*}
\widehat f(u)&=\int_{-\infty}^\infty  e^{iux}f(x)dx\\
f(x)&=\frac{1}{2\pi}\int_{iz-\infty}^{iz+\infty}  e^{-iux}\widehat f(u)du\ .
\end{align*}
where $z=\Im(u)$. The region of $z$ values where the Fourier transform is defined is called \textit{the regularity strip.} This comes in handy for functions like the European call and put payoffs.
\begin{example}[Call Option]
For a call option, $f(x) = (e^s-K)^+$. This function has $\int_{\log K}^\infty |f(s)|^2ds =\infty$ and non-differentiability at $s=\log(K)$. However, we can still write a Fourier transform:
\begin{align*}
\widehat f(u)&= \int_{-\infty}^\infty e^{ius}(e^s-K)^+ds\\
&= \int_{\log K}^\infty e^{s+ius}ds-K\int_{\log K}^\infty e^{ius}ds\\
&=  \frac{e^{s+ius}}{1+iu}\Big|_{s=\log K}^\infty -\frac{Ke^{ius}}{iu}\Big|_{s=\log K}^\infty\ .
\end{align*}
where the anti-derivative of $e^{ius}$ and $e^{s+ius}$ are the same as they are for real variables because the function is analytic in the complex plain. For $z=\Im(u)>1$, the anti-derivative will be zero when evaluated at $s=\infty$, will be infinite for $z<1$, and undefined for $z=1$. Taking $z>1$, we have
\begin{align*}
\widehat f(u)
&= -\frac{e^{(1+iu)\log(K)}}{1+iu}+\frac{Ke^{iu\log(K)}}{iu}\\
&= -\frac{\left(e^{\log(K)}\right)^{(1+iu)}}{1+iu}+\frac{K\left(e^{\log(K)}\right)^{iu}}{iu}\\
&= \frac{K^{1+iu}}{iu-u^2}\ ,
\end{align*}
as shown in Table \ref{tab:payoffs} of Chapter \ref{chapt:stochVol}.
\end{example}

\begin{example}[PutOption]
For a put option, $f(x) = (K-e^s)^+$. This function has $\int_{\log K}^\infty |f(s)|^2ds <\infty$, like the put option is non-differentiability at $s=\log(K)$. Repeating the steps from the previous example, it follows that the regularity strip is $z=\Im(u)<0$ (see Table \ref{tab:payoffs} of Chapter \ref{chapt:stochVol}).
\end{example}
%%%%%%%%%%%%%%%%%%%%%%%%%%%%%%%%%%
\section{The Wave Equation}
Consider the wave (transport) equation,
\begin{align*}
\frac{\partial}{\partial t}V(t,x)&=a\frac{\partial}{\partial x}V(t,x)\\
V(0,x)&=f(x)\ .
\end{align*}
The solution can be found with Fourier transforms:
\begin{align*}
\left(\frac{\partial}{\partial t}-a\frac{\partial}{\partial x}\right)V(t,x) &=\frac{1}{2\pi}\int_{-\infty}^\infty \left(\frac{\partial}{\partial t}-a\frac{\partial}{\partial x}\right)e^{-iux}\widehat V(t,u)\\
&=\frac{1}{2\pi}\int_{-\infty}^\infty e^{-iux} \left(\frac{\partial}{\partial t}+aiu\right)\widehat V(t,u)\ ,
\end{align*}
or simply
\begin{align*}
\frac{\partial}{\partial t}\widehat V(t,u)&=-aiu\widehat V(t,u)\ ,\\
\widehat V(0,u)&=\widehat f(0)\ .
\end{align*}
The solution is
\[\widehat V(t,u) = e^{-aiu t}\widehat f(u)\ ,\]
which we invert to obtain the solution to the wave equation,
\begin{align*}
V(t,x) &= \frac{1}{2\pi}\int_{-\infty}^\infty e^{-iux}e^{-aiu t}\widehat f(u)du\\
&=\frac{1}{2\pi}\int_{-\infty}^\infty e^{-iu(x+ at)}\widehat f(u)du\\
&=f(x+at)\ .
\end{align*}
The lines $x = c-at$ are the characteristic lines so that for any $c$ we have $v(t,c-at) = f(c)$; the initial information $f(c)$ at point $c$ is propagated along the characteristic line.

%%%%%%%%%%%%%%%%%%%%%%%%%%%%%%%%%%
\section{The Heat Equation}
Consider the heat equation

\begin{align*}
\frac{\partial}{\partial t}V(t,x)&=\frac12\Delta V(t,x)\\
V(0,x)&=f(x)\ ,
\end{align*}
where and $\int_{-\infty}^\infty |f(x)|^2dx<\infty$. For scalar $x$ we have $\Delta =\frac{\partial^2}{\partial x^2}$, and the inverse Fourier transform yields an ODE under the integral sign,
\begin{align*}
\left(\frac{\partial}{\partial t}-\frac12\frac{\partial^2}{\partial x^2}\right)V(t,x)&= \frac{1}{2\pi}\int_{-\infty}^\infty  \left(\frac{\partial}{\partial t}-\frac12\frac{\partial^2}{\partial x^2}\right)e^{-iux}\widehat V(t,u)du\\
&= \frac{1}{2\pi}\int_{-\infty}^\infty  e^{-iux}\left(\frac{\partial}{\partial t}+\frac{u^2}{2}\right)\widehat V(t,u)du \\
&=0\ ,
\end{align*}
or simply,
\begin{align}
\nonumber
\frac{d}{dt}\widehat V(t,u)&=-\frac{u^2}{2}\widehat V(t,u)\\
\label{eq:Vhat_ODE}
\widehat V(0,u)&=\widehat f(u) \ .
\end{align}
The solution to \eqref{eq:Vhat_ODE} is 
\[\widehat V(t,u) = \widehat f(u) e^{-\frac{u^2}{2}t}\ .\]
Applying the inverse Fourier transform yields the solution,
\begin{align*}
V(t,x)& = \frac{1}{2\pi}\int_{-\infty}^\infty e^{-iux}\widehat V(t,u)du\\
& = \frac{1}{2\pi}\int_{-\infty}^\infty e^{-iux}\widehat f(u) e^{-\frac{u^2}{2}t}du\\
& = \frac{1}{2\pi}\int_{-\infty}^\infty e^{-iux}\left(\int_{-\infty}^\infty e^{iuy}f(y) dy\right)e^{-\frac{u^2}{2}t}du\\
& = \frac{1}{\sqrt{2\pi t}}\int_{-\infty}^\infty f(y) \left(\sqrt{\frac{t}{2\pi}}\int_{-\infty}^\infty e^{iu(y-x)} e^{-\frac{u^2}{2}t}du\right)dy\\
& = \frac{1}{\sqrt{2\pi t}}\int_{-\infty}^\infty f(y)\left( \mathbb E e^{i\mathcal Z(y-x)/\sqrt t}\right) dy\qquad\hbox{where $\mathcal Z\sim N(0,1)$}\\
& = \frac{1}{\sqrt{2\pi t}}\int_{-\infty}^\infty f(y) e^{-\frac{(y-x)^2}{2t}}dy\ .
\end{align*}
In fact, it can be seen as the expectation of a Brownian motion,
\[V(t,x) = \mathbb Ef(x+W_t)\]
where $W_t$ is a Brownian motion. If $f(x) = \delta_0(x)$, then the we have the \textit{fundamental solution},
\[\Phi(t,x) = \frac{1}{\sqrt{2\pi t}}e^{-\frac{x^2}{2t}}\ ,\]
from which solutions for different initial conditions are convolutions,
\[V(t,x) = f*\Phi(t,x) = \int_{-\infty}^\infty f(y)\Phi(t,x-y)\ .\]

%%%%%%%%%%%%%%%%%%%%%%%%%%%%%%%%%%
%%%%%%%%%%%%%%%%%%%%%%%%%%%%%%%%%%
\section{The Black-Scholes Equation}
The Black-Scholes model is
\[\frac{dS_t}{S_t}=rdt+\sigma dW_t^Q\]
where $W^Q$ is a risk-neutral Brownian motion. Letting $X_t=\log(S_t)$ the Black-Scholes equation is
\begin{align*}
\left(\frac{\partial}{\partial t}+\frac{\sigma^2}{2}\frac{\partial^2}{\partial x^2}+\left(r-\frac{\sigma^2}{2}\right)\frac{\partial}{\partial x}-r\right)V(t,x)&=0\\
V(T,x)&=\psi(x) \ ,
\end{align*}
where $\psi(X_T)$ is the claim, e.g. a call option $\psi(x) = (e^x-K)^+$. Using the same Fourier techniques as we did with the wave and heat equations, we have
\begin{align*}
\left(\frac{\partial}{\partial t}-\frac{\sigma^2u^2}{2}-iu\left(r-\frac{\sigma^2}{2}\right)-r\right)\widehat V(t,u)&=0\\
\widehat V(T,u)&=\widehat \psi(u) \ ,
\end{align*}
which has solution 
\[\widehat V(t,u)=\widehat\psi(u)e^{-(T-t)\left(\frac{\sigma^2u^2}{2}+iu\left(r-\frac{\sigma^2}{2}\right)+r\right)}\ .\]
Applying the inverse Fourier transform, we arrive at the risk-neutral pricing formula (i.e. Feynman-Kac),
\begin{align*}
V(t,x)&=\frac{1}{2\pi}\int_{-\infty}^\infty e^{-iux}\widehat V(t,u)du\\
&=\frac{e^{-r(T-t)}}{2\pi}\int_{-\infty}^\infty e^{-iux}\widehat\psi(u)e^{-(T-t)\left(\frac{\sigma^2u^2}{2}+iu\left(r-\frac{\sigma^2}{2}\right)\right)}du\\
&=\frac{e^{-r(T-t)}}{2\pi}\int_{-\infty}^\infty e^{-iux}\left(\int_{-\infty}^\infty e^{iuy}\psi(y)dy\right)e^{-(T-t)\left(\frac{\sigma^2u^2}{2}+iu\left(r-\frac{\sigma^2}{2}\right)\right)}du\\
&=\frac{e^{-r(T-t)}}{2\pi}\int_{-\infty}^\infty\int_{-\infty}^\infty \psi(y) e^{iu\left(y-x-(T-t)\left(r-\frac{\sigma^2}{2}\right)\right)}e^{-(T-t)\frac{\sigma^2u^2}{2}} dudy\ .
\end{align*}
Take change of variable $y'=y-x-(T-t)\left(r-\frac{\sigma^2}{2}\right)$ such that $dy'=dy$, then 
\begin{align*}
V(t,x)&=\frac{e^{-r(T-t)}}{2\pi}\int_{-\infty}^\infty\int_{-\infty}^\infty \psi\left(y'+x+(T-t)\left(r-\frac{\sigma^2}{2}\right)\right) e^{iuy'}e^{-(T-t)\frac{\sigma^2u^2}{2}} dudy'\\
&=\frac{e^{-r(T-t)}}{2\pi}\int_{-\infty}^\infty \psi\left(y'+x+(T-t)\left(r-\frac{\sigma^2}{2}\right)\right)\left( \int_{-\infty}^\infty e^{iuy'}e^{-(T-t)\frac{\sigma^2u^2}{2}} du\right)dy'\\
&=\frac{e^{-r(T-t)}}{\sqrt{2\pi\sigma^2(T-t)}}\int_{-\infty}^\infty \psi\left(y'+x+(T-t)\left(r-\frac{\sigma^2}{2}\right)\right)\left(\mathbb E e^{i\frac{\mathcal Z}{\sigma\sqrt{T-t}}y'}\right)dy'\\
&=\frac{e^{-r(T-t)}}{\sqrt{2\pi\sigma^2(T-t)}}\int_{-\infty}^\infty \psi\left(y'+x+(T-t)\left(r-\frac{\sigma^2}{2}\right)\right)e^{-\frac12\frac{\left(y'\right)^2}{\sigma^2(T-t)}}dy'\\
&=e^{-r(T-t)}\mathbb E^Q\psi\left(x+(T-t)\left(r-\frac{\sigma^2}{2}\right)+\sigma (W_T^Q-W_t^Q)\right)\\
&=e^{-r(T-t)}\mathbb E^Q\left[\psi\left(X_T\right)\Big|X_t=x\right]\ .
\end{align*}

\small{\bibliography{refs}}
\end{document}